\documentclass{article}

\usepackage[maxnames=50,  sorting=none]{biblatex}

\nocite{*}

\usepackage{tabularx}
\newcommand{\ssection}[1]{\smallskip\phantomsection\addcontentsline{toc}{section}{#1}\textit{#1.---}}
\usepackage[T1]{fontenc}    
\usepackage{url}            
\usepackage{booktabs}       
\newcommand{\rectangle}{\,\fboxsep0pt\fbox{\rule{0.6 em}{0.0  pt}\rule{0 pt}{1 ex}}}
\usepackage{amsfonts}       
\usepackage{nicefrac}       
\usepackage{microtype}      
\usepackage{lipsum}
\usepackage[ascii]{inputenc}
\usepackage{bbm,braket,microtype,mathtools}
\usepackage[T1]{fontenc}
\usepackage{nomencl}
\usepackage{lmodern}
\usepackage[USenglish]{babel}
\usepackage{amsfonts}
\usepackage{inputenc}
\usepackage[pdfencoding=auto,bookmarksnumbered,hypertexnames=false,colorlinks,colorlinks=true,linkcolor=blue,urlcolor=black,citecolor=blue,anchorcolor=green]{hyperref}
\usepackage{dcolumn,bbm,url,booktabs,braket,microtype,aliascnt,amsthm,mathtools,mathrsfs,etoolbox,cancel,capt-of,algpseudocode,float}
\usepackage[T1]{fontenc}
\usepackage{graphicx}
\usepackage{bbm,braket,microtype,mathrsfs,amsmath,amssymb,color,amsthm,mathtools,fullpage,enumitem,ae,aecompl,bm,thm-restate}
\usepackage[USenglish]{babel}
\usepackage[capitalize]{cleveref}
\usepackage{newpxtext}
\usepackage{csquotes}

\DeclareMathOperator{\rank}{rank}
\DeclareMathOperator{\Sym}{Sym}

\DeclareMathOperator{\tr}{tr}

\newcommand{\appropto}{\mathrel{\vcenter{
  \offinterlineskip\halign{\hfil$##$\cr
    \propto\cr\noalign{\kern2pt}\sim\cr\noalign{\kern-2pt}}}}}

\newcommand{\CC}{\mathbb C}
\newcommand{\RR}{\mathbb R}
\newcommand{\ZZ}{\mathbb Z}
\newcommand{\NN}{\mathbb N}

\newcommand{\Hil}{\mathcal H}
\newcommand{\ot}{\otimes}

\newcommand{\eps}{\varepsilon}
\newcommand{\id}{\mathbbm 1}

\newcommand{\GL}{\mathrm{GL}}

\newcommand{\EE}{\mathbb E}
\newcommand{\Perm}{\operatorname{Perm}}
\newcommand{\Gauss}[1]{{\mathcal G^{#1 \times #1}}}
\newcommand{\Usub}[2]{{U(#1)^{#2 \times #2}}}
\newcommand{\TC}[1]{}

\newcommand{\hook}{\text{hook}}

\newcommand{\ketbra}[2]{\ket{#1}\!\!\bra{#2}}
\newcommand{\WD}{\text{WD}}

\newcommand{\Pl}{\text{Pl}}

\crefname{algo}{Algorithm}{Algorithms}\crefname{fig}{Figure}{Figures}
\newtheorem{thm}{Theorem}[section]\crefname{thm}{Theorem}{Theorems}
\newtheorem{lem}[thm]{Lemma}\crefname{lem}{Lemma}{Lemmas}
\newtheorem{conj}[thm]{Conjecture}\crefname{conj}{Conjecture}{Conjectures}
\crefname{prp}{Proposition}{Proposition}
\newtheorem{rem}[thm]{Remark}\crefname{rem}{Remark}{Remarks}
\newtheorem{cor}[thm]{Corollary}\crefname{cor}{Corollary}{Corollaries}
\newtheorem{dfn}[thm]{Definition}\crefname{dfn}{Definition}{Definitions}
\crefname{exa}{Example}{Examples}
\newtheorem*{thm*}{Theorem}\crefname{thm}{Theorem}{Theorems}
\newtheorem*{lem*}{Lemma}\crefname{lem}{Lemma}{Lemmas}
\newtheorem*{conj*}{Conjecture}\crefname{conj}{Conjecture}{Conjectures}

\floatstyle{boxed}\newfloat{algo}{htbp}{alg}\floatname{algo}{Algorithm~}

\title{Permanent of random matrices from representation theory: moments, numerics, concentration, and comments on hardness of boson-sampling}

\author{
  Sepehr Nezami\footnote{ California Institute of Technology. Email: nezami@caltech.edu.} }
\begin{document}

\maketitle
\begin{abstract}
Computing the distribution of permanents of random matrices has been an outstanding open problem for several decades. In quantum computing, ``anti-concentration'' of this distribution is an unproven input for the proof of hardness of the task of boson-sampling. Using a hybrid representation-theoretic and combinatorial approach, we study permanents of random i.i.d. complex Gaussian matrices, and more broadly, submatrices of random unitary matrices. We prove strong lower bounds for all moments of the permanent distribution. Moreover, we provide substantial evidence that our lower bounds are close to being tight, and therefore, constitute accurate estimates for the moments. Let $U(d)^{k\times k}$ indicate the distribution of $k\times k$ submatrices of $d\times d$ Haar distributed random unitary matrices, and $\mathcal  G^{k\times k}$ be the distribution of $k\times k$ matrices with matrix elements being i.i.d. standard complex Gaussian random numbers. {\bf(1)} Using the Schur-Weyl duality (or more precisely, the Howe's $\text{GL}_k \times \text{GL}_t$ duality), we provide an explicit expansion formula for the $2t$-th moment of $|\text{Perm}(M)|$, i.e., $\mathbb E|\text{Perm}(M)|^{2t}$, when $M$ is drawn from $U(d)^{k\times k}$ or $\mathcal G^{k\times k}$. {\bf(2)} When the matrices are drawn from the above distributions, we prove a surprising size-moment duality: the $2t$-th moment of the permanent of random $k\times k$ matrices is equal to the $2k$-th moment of the permanent of $t\times t$ matrices, {\bf(3)} We design an algorithm to exactly compute high moments of relatively small matrices. {\bf(4)} We prove strong lower bounds for arbitrary moments of permanents of matrices drawn from $\mathcal G^{ k\times k}$ or $U(k)$, and conjecture that our lower bounds are close to saturation up to a small multiplicative error. We provide extensive numerical and analytical evidence for our conjecture. {\bf (5)} Assuming our conjectures, we use the large deviation theory to compute the tail of the log-permanent probability density function of Gaussian matrices for the first time. {\bf (6)} We argue that it is unlikely that the permanent distribution can be uniquely determined from the integer moments and one may need to supplement the integer moment calculations with extra assumptions in order to prove the permanent anti-concentration conjecture. 
\end{abstract}
\tableofcontents
\section{Introduction}

Random matrix theory has been very successful in computing the distribution of various quantities derived from random matrices. These include, for example, the derivation of eigenvalue distributions\TC{?}, low order polynomials of matrix elements~\cite{collins2006integration}, and the distribution of matrix determinants~\cite{girko1978central,girko2012theory,tao2006random,nguyen2014random}. On the other hand, the distribution of the permanent of random matrices (when the matrix elements have a vanishing mean) is understood to a much lesser extent~\cite{tao2009permanent,vu2020recent}, despite several decades of research~\cite{girko1978central,girko2012theory,ji2019approximating,borovskikh1994random,rempala1999limiting,szekely1995limit}. 
\par
The case of permanents is more difficult than the other mentioned quantities for several reasons: (1) The permanent is a high order polynomial of many terms, hence, the computation of its moments quickly becomes intractable. (2) The permanent, unlike the determinant, lacks a geometric description and cannot be expressed in terms of matrix eigenvalues. Therefore, most of the tools of the random matrix theory cannot be directly applied to this case. (3) It is known that calculation of the permanent of individual matrices is hard even in the average case~\cite{valiant1979complexity}, and therefore, sampling from the permanent distribution of large matrices is computationally prohibited. 
\par 
In addition to interest from mathematics and computer science communities, understanding certain aspects of the permanent distribution is essential for one of the leading proposals of demonstrating quantum supremacy with photonic systems: the hardness of the task of ``boson sampling'' relies on the ``permanent anti-concentration'' conjecture. Let $\Gauss{k}$ be the distribution of random $k\times k$ matrices with matrix elements being i.i.d. complex Gaussian numbers with mean $0$ and variance $1$, then the conjecture is:
\begin{conj}[The permanent anti-concentration conjecture~\cite{aaronson2011computational}]\label{conj:arkh}
\[\mathrm{Pr}_{M \sim \Gauss{k}}\left( \left|\Perm(M) \right|^2 \leq \epsilon k! \right) \leq A k^B \epsilon^C,\quad \text{for some constants $A$, $B$, and $C>0$ and all $\eps>0$}.\]
\end{conj}
This conjecture states that the permanent distribution should not have a considerable weight near the origin. The best known result approaching this conjecture is proven by Tao and Vu~\cite{tao2009permanent}, which states that the permanent of $\pm 1$ matrices is of order $k^{k(\frac12+o(1))}$ with probability $1-o(1)$, which in particular means that the probability of vanishing is $o(1)$. Although it is very plausible that the techniques of~\cite{tao2009permanent} can be extended to the case of random Gaussian matrices, this result does not show the polynomial anti-concentration as required in~\cref{conj:arkh}.
\par 
While proving~\cref{conj:arkh} has been the main motivation of the author, most of the results and heuristics arguments will provide insight into the parts of the distribution away from the origin, and therefore, the conjecture will remain open. 
\par
Our main technical contribution is understanding and bounding the moments of the permanent distribution. Exact results are known for the cases of either small matrices ($1\times 1$ or $2 \times 2$) or small moments ($2$ or $4$):
\begin{align}\label{eq:moments12}
    &\EE_{M \sim \Gauss{k}} \left|\Perm(M) \right|^2=\EE_{M \sim \Gauss{1}} \left|\Perm(M) \right|^{2k} = k!,\quad \text{and},\\
    &\EE_{M \sim \Gauss{k}} \left|\Perm(M) \right|^{2\times 2}=\EE_{M \sim \Gauss{2}} \left|\Perm(M) \right|^{2k} = k!(k+1)!.
\end{align}
Few other moments have also been calculated for small-sized matrices\TC{Cite old papers}. 
\par
To understand and bound the higher moments and large matrices, we take a hybrid representation-theoretic and combinatorial approach, and will be guided by exact moment calculations performed in~\cref{sec:alg}.
\par
We summarize our results here:
\begin{itemize}
    \item {\bf Results regarding the permanent of matrices with i.i.d. standard Gaussian matrix elements:} 
    \begin{enumerate}
        \item We show that $2t$-th moment of permanents of $k \times k$ matrices can be computed by solving a counting problem:
        Start by a $k \times t$ table filled with numbers $1,\cdots, kt$ and let $S_{kt}$ be the permutation group that permutes cells of the table. Consider the subgroups $R ,C \in S_{kt}$ that preserve rows and columns of the table, respectively (see~\cref{fig:rc_tables}). Then
        \[\EE_{M \sim \Gauss{k}} \left|\Perm(M) \right|^{2t} = \text{Number of solutions of }(r_1c_1r_2c_2 = e), \quad r_1,r_2 \in R,c_1,c_2 \in C, \]
        with $e$ being the identity element of $S_{kt}$. In the next steps, we wish to argue that for $k,t\geq 3$, the four step process $\pi \rightarrow r_1c_1r_2c_2 \pi$, $\pi \in S_{kt}$ is sufficiently randomizing in $S_{kt}$, such that the probability of obtaining $e$ after starting with $e$ is close to $1/|S_{kt}|$. 
        \item This relation exhibits the moment-size duality: Exchanging $t$ and $k$ does not change the result of the moment calculation. 
        \item We use representation theory to attack the above combinatorial counting problem. If $\rho_{\text{reg}}$ is the regular representation of $S_{kt}$, then, 
        \[\text{Number of solutions of }(r_1c_1r_2c_2 = e) =\frac1{|S_{kt}|} \sum_{r_1,r_2\in R,c_1,c_2 \in C} \tr[\rho_{\text{reg}}(r_1c_1r_2c_2)]. \]
        We decompose the regular representation into the direct sum of irreps and calculate the first few terms of this sum. This calculation is involved and constitutes the bulk of the paper. We will observe that the terms in the sum get exponentially smaller and the first few terms provide a good approximation of the permanent moments. We show that
        \[ \EE_{M \sim \Gauss{k}} \left|\Perm(M) \right|^{2t} \geq 1.625 \times k!^{2t}t!^{2k}/(kt)!, \text{ for }k,t\geq 4.\]
        \item We develop an algorithm to exactly compute the permanent moments for the cases of $t=3, k\leq 100$ and $t=4,k\leq 10$.
        \item Backed by our analytical and numerical results, we conjecture that 
        \[\EE_{M \sim \Gauss{k}} \left|\Perm(M) \right|^{2t}\leq 2\times k!^{2t}t!^{2k}/(kt)!, \text{ for }t,k\geq 3,\] and that $\EE_{M \sim \Gauss{k}} \left|\Perm(M) \right|^{2t} / [k!^{2t}t!^{2k}/(kt)!]$ quickly approaches a constant. This conjecture confirms the prediction of part (1) about the row and column permutations being sufficiently randomizing.
        \item Assuming our conjecture, we use large deviation theory to compute the tail of the distribution of log-permanent of random Gaussian matrices (See~\cref{sec:concentration}).
    \end{enumerate}
    \item {\bf Results regarding permanents and determinants of minors of Haar random unitary matrices.}
    \begin{enumerate}
        \item Using Howe's $\GL_d \times \GL_t$ duality, we prove an explicit expansion formula for calculating the moments of permanents, a direct generalization of our combinatorial counting results for the i.i.d. Gaussian case.
        \item We show that the size-moment duality holds for the case of moments of permanents of unitary minors.
        \item We prove a lower bound for the moments of permanents in this case. When the minor is equal to the whole unitary matrix, we show that the \emph{Hunter-Jones} conjecture on the moments of the permanent of random unitary matrices indicates that our bound is very close to being tight.
        \item As a bonus, our result provides a simple explicit formula for the moments of the determinant of minors of random unitary matrices.
    \end{enumerate}
\end{itemize}
\subsection{A combinatorial counting problem}\label{sec:intro_1}
As the first step, we will translate the moment calculation problem to a counting problem involving the symmetric group\footnote{Translating calculations involving the unitary group to the ones involving the symmetric group is a central theme in the representation theory of the unitary and the symmetric groups. See, for example, the Schur-Weyl duality or the Frobenius character map~\cite{fulton2013representation,bump2004lie}.}. We only focus on the Gaussian ensemble $\Gauss k$ in this section and discuss the case of minors of Haar random unitary matrices in~\cref{sec:subintro}. Recall the definition of the permanent: 
\[ \Perm(M) = \sum_{\pi \in S_k} \prod_{i=1}^k M_{i ,\pi(i)},\]
Define $L_{k,t}$ to be the set of all $k\times t$ tables, where each column is a permutation of $1,\cdots , k$. For example, a typical element of $L_{k,t}$ could be the following table:
\[k \text{ rows}\left\{\,\,\,\begin{tabular}{|c|c|c|c|c|}
\hline
   $1$  &  $k$ &  $k-1$&  $\cdots$&  $k-4$ \\\hline
$4$     &  $2$ &  $2$&  $\cdots$&  $k-1$ \\\hline 
  $\vdots$ & &   &   &  $\vdots$ \\\hline
$2$     &  $4$ & $1$ & $\cdots$ &  $6$  \\ \hline
\end{tabular} \right.\]
Expanding $\Perm(M)^t$ into the individual monomials, we get $\Perm(M)^t = \sum_{l \in L_{k,t}} \prod_{i=1,j=1}^{t,k} M_{i,l_{ij} }$. Therefore, 
\[ \EE_{M\sim \Gauss k}\left| \Perm(M)\right|^{2t} = \sum_{l,\overline l \in L_{k,t}} \prod_{j}^{k} \left(\EE_{M\sim \Gauss k} \prod_{i}^{t} M_{j,l_{ij} }\overline{M_{j,\overline l_{ij} }}\right).\]
Let $\mathcal G$ be the standard complex normal distribution. Using $\EE_{x\sim \mathcal G} x^a \overline{x}^b=a!\delta_{ab}$, it is straightforward to see that the term in the parenthesis is nothing but the number of row preserving permutations that map the table $l$ to $\overline l$. 
\par
\begin{dfn}[Row and column preserving subgroups]\label{dfn:row-col-pres}
Given a $k\times t$ rectangular table, consider the group $S_{kt}$ of all permutations of the table cells. We define two important subgroups:
\begin{figure}
    \centering
    \includegraphics[width=12cm]{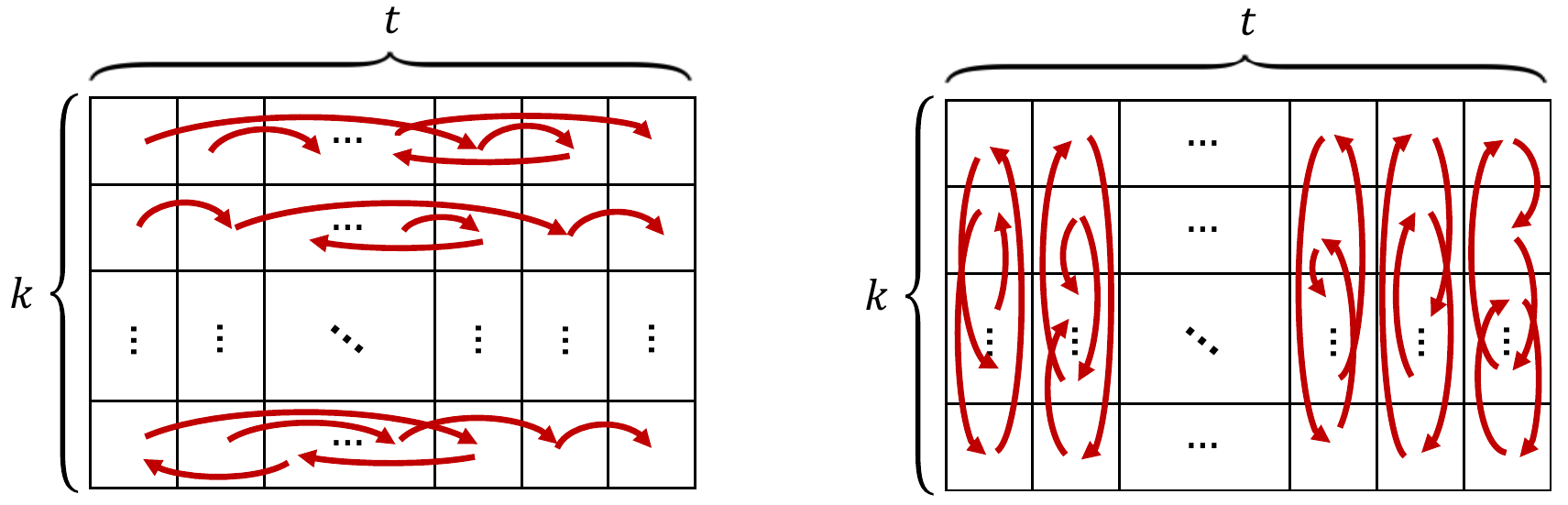}
    \caption{Row preserving and column preserving subgroups of $S_{kt}$. The group $S_{kt}$ naturally acts on the table cells by permuting them. {\bf Left.} A typical element of the row preserving subgroup $R$. These permutations do not mix cells from different columns. {\bf Right.} A typical element of the column preserving subgroup $C$.}
    \label{fig:rc_tables}
\end{figure}
\begin{itemize}
    \item \emph{$R$, the row preserving subgroup}: This subgroup contains all permutation that permute the table elements within individual rows. See~\cref{fig:rc_tables}, left. This group has $t!^k$ elements.
    \item \emph{$C$, the column preserving subgroup}: This subgroup permutes the elements within the columns. See~\cref{fig:rc_tables}, right. Order of this group is $k!^t$.
\end{itemize}
\end{dfn}
\par 
With this notation, we conclude that $\EE_{M\sim \Gauss k}\left| \Perm(M)\right|^{2t} =\sum_{l,\overline l \in L_{k,t}} \sum_{r\in R} \delta_{l,r(l')}$, where $r(l)$ indicates the action of the permutation $r$ on the table $l$. This expression can be simplified even further: let $l_0$ be a canonical table where the first row is all $1$'s, second row is all $2$'s, and so on. Using the fact that all tables in $L_{k,t}$ can be constructed by the action on an element of $C$ on $l_0$, we have that  $\sum_{l,\overline l \in L_{k,t}} \sum_{r\in R} \delta_{l,r(l')}= \sum_{c_1,c_2\in C} \delta_{l_0,c_1 r c_2(l_0)}$. From this, it is easy to see that the following theorem holds:
\begin{thm} \label{thm:comb}
Let $R$ and $C$ be the row preserving and column preserving permutations of a $k\times t$ grid. Then 
\begin{equation}\label{eq:comb_rel}
    \EE_{M\sim \Gauss k}\left| \Perm(M)\right|^{2t} = \sum_{r_1,r_2 \in R,\, c_1,c_2 \in C} \delta(r_1c_1r_2c_2,e),
\end{equation} 
where $e$ is the identity element of the symmetric group $S_{kt}$.
\end{thm}
The above expression has a remarkable $t\leftrightarrow k$ symmetry:
\begin{cor}[Moment-size duality]\label{cor:gauss_duality}
The $2 t$-th moment of the permanent of random complex Gaussian $k\times k$ matrices is equal to $2 k$-th moment of the permanent of random Gaussian $t\times t$ matrices:
\begin{equation}
    \EE_{M\sim \Gauss k}\left| \Perm(M)\right|^{2t} = \EE_{M\sim \Gauss t}\left| \Perm(M)\right|^{2k}.
\end{equation} 
\end{cor}
We will later see that the analog of~\cref{cor:gauss_duality} holds for the case of submatrices of random unitary matrices, however, the proof will not be as simple.
\par
Let us momentarily discuss one interpretation of~\cref{thm:comb}. Consider the process of randomly choosing $r_1,r_2 \in R$ and $c \in C$, and forming the combination $r_1 c\, r_2$. The resulting permutations will have a probability distribution which we call $p(\pi)$. More precisely, $p(\pi)$ is defined by the following equality:
\begin{equation}\label{eq:defprob}
    \EE_{r_1,r_2 \in R,\, c \in C}\,\, r_1 c \,r_2 = \sum_{\pi \in S_{kt}} p(\pi) \pi.
\end{equation}  
Using $p(\pi) = p(\pi^{-1})$, is straightforward to see that
$\EE_{r_1,r_2 \in R,\, c_1,c_2 \in C} \,\delta(r_1c_1r_2c_2,e)=\sum_{\pi\in S_{kt}} p(\pi)^2$,
and therefore,
\begin{equation}\label{eq:2norm}
    \EE_{M\sim \Gauss k}\left| \Perm(M)\right|^{2t} = k!^{2t}t!^{2k} \| p \|_2^2\geq \frac{k!^{2t}t!^{2k}}{(kt)!},
\end{equation}
where $\| p\|_2=(\sum_\pi p(\pi)^2)^{1/2}$ is the $2$-norm of the distribution $p$, and we used $ \|p \|_2^2 \geq |S_{kt}|^{-1}$. The equality in~\cref{eq:2norm} happens when $p(\pi)$ is a constant function. 
\par 
We will argue that the inequality~(\ref{eq:2norm}) is not saturated for permanents, but surprisingly, it is close to being saturated. When $k,t\geq 3$, we conjecture that  $\|p\|_2^2\leq 2|S_{kt}|^{-1}$, which means that the 2-norm distance between the distribution of $r_1 c r_2$ (or similarly $c_1 r c_2$) and the uniform distribution is smaller than $1/\sqrt{(kt)!}$.
\par
\subsection{Representation theory and the moment bounds}\label{sec:intrep}
We wish to analyze~\cref{eq:comb_rel} using the language of representation theory. Recall the regular representation of the symmetric group, $\rho_{\text{reg}}(\pi)$ for $\pi \in S_{kt}$, which assigns a $(kt)!$ dimensional matrix to each element of $S_{kt}$ (not to be confused with the standard $kt$ dimensional representation). Importantly, we have that $\tr [\rho_{\text{reg}}(\pi)]=(kt)!\delta(\pi,e)$. With this, we can write~\cref{eq:comb_rel} as,
\begin{equation}\label{eq:gauss-bad}
    \EE_{M\sim \Gauss k}\left| \Perm(M)\right|^{2t} = \frac1{(kt)!}\sum_{r_1,r_2\in R,c_1,c_2 \in C}\tr[ \rho_{\text{reg}}(r_1c_1r_2c_2)].
\end{equation} 
\par
\begin{rem}[Notational remark]\label{rem:notational-remark}
We frequently replace $(\sum_{r\in R} r)$ and $(\sum_{c\in C} c)$ by $R$ and $C$ when there is no confusion. For instance,~\cref{eq:gauss-bad} will be written as:
\begin{equation}\label{eq:gauss-good}
    \EE_{M\sim \Gauss k}\left| \Perm(M)\right|^{2t} = \frac1{(kt)!} \tr[ \rho_{\text{reg}}(RCRC)].
\end{equation} 
\end{rem}
The regular representation of a group can be decomposed into irreducible representations. The irreducible representations of $S_{kt}$ are indexed by the ordered partitions of the number $kt$, or as it commonly referred to, the~\emph{Young diagrams} of $kt$ boxes (see~\cref{sec:basym}). Hence, we obtain the following \emph{expansion formula} for the moments of permanents:
\begin{equation}\label{eq:expansion-formula}
\EE_{M\sim \Gauss k}\left| \Perm(M)\right|^{2t} = \sum_{\lambda \vdash kt, l(\lambda)\leq \min(k,t)} \frac{1}{(kt)!} f^\lambda \tr \left[\rho_\lambda(RCRC) \right],
\end{equation}
where $\lambda \vdash kt$ means that $\lambda$ is an ordered partition of $kt$ (or a Young diagram of $kt$ boxes), $l(\lambda)$ is number of parts in the partition $\lambda$ (or the depth of the corresponding Young diagram), and $f^\lambda$ is the dimension of the irreducible representation indexed by $\lambda$ (the constraint $ l(\lambda)\leq \min(k,t)$ is not trivial and will be proven later). All terms in the expansion formula~\cref{eq:expansion-formula} are positive as $\rho_\lambda(R)$ and $\rho_\lambda(C)$ are proportional to projectors\footnote{We can assume that $\rho_\lambda$ is a unitary representation.}.  
\par 
Much of this paper is concerned with studying this sum term by term. The first term corresponds to $\lambda = (kt)$, which identifies the trivial representation with $f^\lambda = 1$. This representation assigns $1$ to every group element, therefore, the total contribution from this term is simply $k!^{2t}t!^{2k}/(kt)!$. Hence,
\begin{equation}\label{eq:lower-bound-1}
\EE_{M\sim \Gauss k}\left| \Perm(M)\right|^{2t} \geq \frac{k!^{2t}t!^{2k}}{(kt)!}.
\end{equation}
\par 
Which is identical to the inequality in~\cref{eq:2norm}. Other terms in the sum can be computed with increasing difficulty in order to improve this simple lower bound. Our techniques for computing these terms are too involved to be explained in the introduction, and heavily rely on the Cauchy's identity, plethysm, and symmetric polynomials. Most of the calculations are necessarily computerized, and the results and techniques are described in details in~\cref{sec:permiid}. Some of the calculated contributions are as follows:
\begin{itemize}
    \item For $\lambda=(kt-1,1)$, $((kt)!)^{-1}f^\lambda \tr [\rho_\lambda(RCRC)] = 0 $.
    \item For $\lambda=(kt-2,2)$, $((kt)!)^{-1}f^\lambda \tr [\rho_\lambda(RCRC)] = \frac{k!^{2t}t!^{2k}}{(kt)!}(1/2-3/(2kt))$.
    \item For $\lambda=(kt-3,3)$ where $t,k\geq  3$, \[((kt)!)^{-1}f^\lambda \tr [\rho_\lambda(RCRC)] = \frac{k!^{2t}t!^{2k}}{(kt)!} \times \frac83 (1/(kt)-6/(k^2t^2)+5/(k^3t^3)).\]
    \item For $\lambda=(kt-4,4)$ where $t,k\geq4$, \[((kt)!)^{-1}f^\lambda \tr [\rho_\lambda(RCRC)] = \frac{k!^{2t}t!^{2k}}{(kt)!} \times \frac1{24} (1+4/k+4/t+\text{Lower order terms}),\]
    see~\cref{sec:RCRC_calc} for the complete expansion.
    \item For $\lambda=(kt-4,2,2)$ with $t,k\geq 3$,
    \[((kt)!)^{-1}f^\lambda \tr [\rho_\lambda(RCRC)] = \frac{k!^{2t}t!^{2k}}{(kt)!} \times \frac1{12} (1-2/k-2/t+\text{Lower order terms}).\]
\end{itemize}
Unfortunately, finding a general formula seems challenging. We can make two important observations at this point: (1) The contribution of terms to $\EE |\Perm (M)|^{2t}/(k!^{2t}t!^{2k}/((kt)!))$ is always $O(1)$ (2) Even the $O(1)$ coefficient of terms get less and less significant as one considers the partitions of $kt$ with smaller first part. As a consequence, we can derive strong lower bounds for the moments of permanents by truncating~\cref{eq:expansion-formula} to include a few explicitly calculated terms. One such bound is the following theorem:  
\begin{thm}[Example moment lower bound]
Let $k,t\geq 4$, then 
\begin{equation}\label{eq:lower-bound-2}
    \EE_{M\sim \Gauss k }\left|\Perm(M)\right|^{2t} \geq \frac{13}8 \frac{k!^{2t}t!^{2k}}{(kt)!}.
\end{equation}
\end{thm}
We refer the reader to~\cref{sec:lbmgc} for more involved bounds.
\par
In the next subsection, we provide further evidence for approximate tightness of this lower bound. 
\subsection{Exact results, permanent moment growth conjecture, and concentration results}
To argue that our lower bounds~\cref{eq:lower-bound-1} and~\cref{eq:lower-bound-2} closely mimic the correct behavior of the moments of the permanent of random matrices, we need to explicitly compute some of these moments. Unfortunately, a direct calculation of large moments of the permanent of even small matrices requires massive computational resources, as they concentrate very slowly. On the other hand, computing the permanent of large matrices is known to be hard, even for individual instances~\cite{valiant1979complexity}. 
\par 
To avoid such obstructions, we develop an algorithm for exact computation of the moments of permanents. See~\cref{sec:alg} for more details. This algorithm is effective when $\min(k,t)\leq 4$, but quickly becomes intractable for other cases. 
\par
We can compare our best lower bounds derived from the expansion formula~\cref{eq:expansion-formula} to the exact results. See~\cref{fig:first}.
\begin{figure}
    \centering
    \includegraphics[width=14 cm] {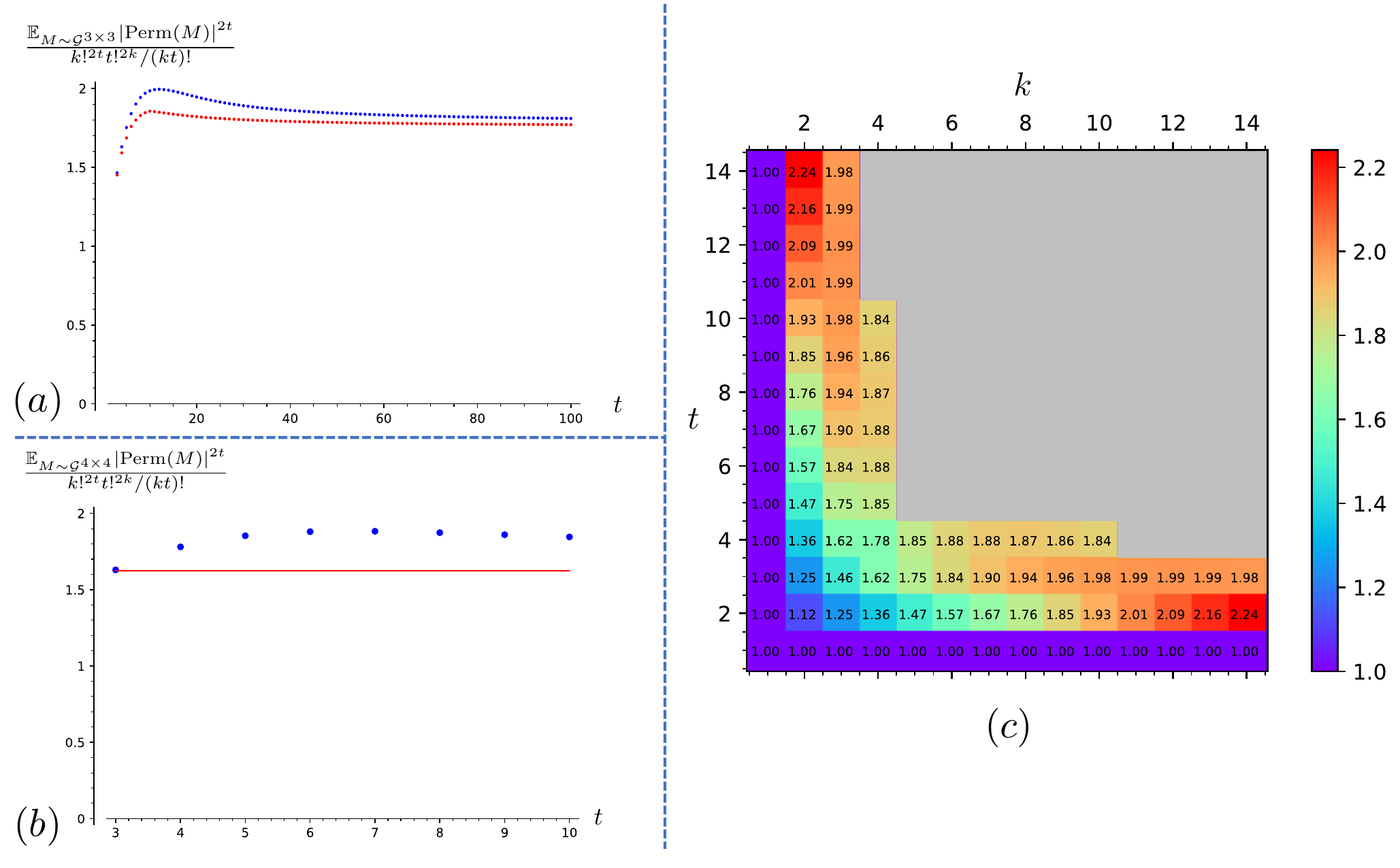}
    \caption{Exact values of the moments of permanents and our lower bounds. Note that due to $t\leftrightarrow k$ symmetry we can exchange $t$ and $k$ in all of the plots and what follows. {\bf (a)} Blue: The plot of $\EE_{M\sim \Gauss k}\left| \Perm(M)\right|^{2t}/(k!^{2t}t!^{2k}/((kt)!))$ for the case of $k=3$. Red: The lower bound derived from including terms in the~\cref{eq:expansion-formula} corresponding to $\lambda = (kt),(kt-1,1), \cdots, (kt-10,10), \text{ and } (kt-4,2,2)$. This lower bound asymptotes to $8849/5040\approx 1.76$ . {\bf (b)} Blue: Exact values for the case of $k=4$. Red: Our universal $13/8 \times k!^{2t}t!^{2k}/(kt)!$ lower bound. {\bf (c)} The plot of $\EE_{M\sim \Gauss k}\left| \Perm(M)\right|^{2t}/(k!^{2t}t!^{2k}/((kt)!))$ for some of the known values of $k$ and $t$. Grey indicates the places that we do not know the exact results, and where our lower bound dictates a value larger than $13/8=1.625$. Our conjecture predicts that the value in the grey area should be smaller than $2$. Note the exceptional case of $t =2$ outside this region, where the ratio grows as $\sqrt{k/\pi}$. As can be seen in part (a), the ratio drops to a value smaller than $2$ as we move to $t=3$. }
    \label{fig:first}
\end{figure}
These results, along with the fast decay of the terms in~\cref{eq:expansion-formula}, motivate the following conjecture:
\begin{conj}[The permanent moment growth conjecture.]\label{conj:intro}
Suppose that $t,k\geq 3$, then the quantity
\begin{equation}\label{eq:permmomnorm}
    \frac{\EE_{M\sim \Gauss k }\left|\Perm(M)\right|^{2t} }{{k!^{2t}t!^{2k}}/{(kt)!}} 
\end{equation}
is always smaller than $2$ and quickly asymptotes to a constant $C$ as we increase both $k$ and $t$. 
\end{conj}
We know that the constant of~\cref{conj:intro} is larger than $1.625 = 13/8$. See~\cref{sec:lbmgc,sec:conjsection} for details and more evidence. We emphasize that this conjecture is nothing but the statement that when one chooses $r_1,r_2$ randomly form $R$ and $c$ randomly from $C$, then $r_1cr_2$ is close to be randomly distributed in $S_{kt}$ (see the paragraph before~\cref{sec:intrep}).  
\par 
We can provide yet another interpretation of normalized permanent moment~\cref{eq:permmomnorm}. Consider randomly choosing $t$ permutations in $S_k$, and summing their standard $k$-dimensional representation matrices. The resulting matrix will be a \emph{weak $t$-magic square}, which is a $k \times k$ matrix where row and column sums are equal to $t$. As we increase $t$, one expects that the individual matrix elements of the resulting magic squares become independent. In~\cref{sec:conjsection}, we show that the normalized permanent moment~\cref{eq:permmomnorm} is directly related to a measure of independence of such magic square matrix elements.
\par
Lastly, we discuss the shape of the distribution of random log-permanents when $k$ is very large. Define the random variable 
\begin{equation*}
   Y_k := \frac{1}{k} \log \,\left[\frac{|\Perm(M)|}{\sqrt{k!}}\right]\quad \text{for}\quad M\sim \Gauss{k},
\end{equation*}  
and let $p_{Y_k}(y)$ be the probability density function of $Y_k$. We normalized $\Perm$ by a factor $\sqrt{k!}$, because $\sqrt{k!} = \sqrt{\EE |\Perm(M)|^2}$, and also $\sqrt{k!}$ is close to where the most of the mass of the permanent probability distribution is expected to be according to~\cite{tao2009permanent}. Using the techniques of large deviation theory and assuming~\cref{conj:intro}\footnote{To be more accurate, we use a slightly different version of this conjecture. See~\cref{conj:ext}.}, we show that
 \[p_{Y_k}(y) = e^{-ke^{2y+1}\omega(y)+o(k)}, \text{ for }y>0.21.\]
 Here, we used the small $o$ notation, and $\omega(y)$ is a function that we can compute and is very close to $1$ for $y>2$ and starts decaying to zero as we decrease $y$ to smaller values. See~\cref{fig:omegaplot} for a plot of $\omega(y)$.
 \par 
 Moreover, we argue that under milder assumptions, 
 \[ e^{-4y k+o(k)} \leq p_{Y_k}(y)\leq e^{-6y k+o(k)}, \text{ for }0<y<0.048.\]
 \par 
 See~\cref{sec:concentration} for detailed arguments. Interestingly, if we replace the permanent by the determinant in the definition of $Y_k$, we see that the tail of the distribution decays mush faster, which suggests that the tail of random permanent distribution is much heavier than the tail of the random determinant distribution.

\subsection{Permanent of submatrices of random unitary matrices}\label{sec:subintro}
In this section, we intend to generalize our expansion formula~\cref{eq:expansion-formula} to the case of submatrices of Haar random unitary matrices. Define $\Usub{d}{k}$ to be the distribution of the leading $k\times k$ minors of Haar random unitary $d\times d$ matrices. In the limit of $d\rightarrow \infty$, the distribution $\Usub{d}{k}$ approaches $\Gauss k$ normalized by a factor of $d^{-1/2}$. On the other hand, when $k=d$, $\Usub d k$ is simply the distribution of Haar random unitary matrices. 
\par
Using a representation theory duality, called the ``Howe $\GL_d\times \GL_t$ duality'', we prove a generalization of our expansion formula~\cref{eq:expansion-formula}:
If $k\leq d$, we have 
\begin{equation}\label{eq:generalized-expansion-formula}
\EE_{M\sim \Usub d k}\left| \Perm(M)\right|^{2t} = \sum_{\lambda \vdash kt, l(\lambda)\leq \min(k,t)} \frac{1}{\WD_\lambda(d)} \left(\frac{f^\lambda}{(kt)!}\right)^2 \tr \left[\rho_\lambda(RCRC) \right],
\end{equation}
where $\WD_\lambda(d)$ is the dimension of the irreducible representation of $U(d)$ indexed by the partition $\lambda$ (it is defined in~\cref{eq:weyl_dim_formula_definition}). Moreover, it is straightforward to derive~\cref{eq:expansion-formula} from~\cref{eq:generalized-expansion-formula}. Note that this expansion enjoys the $k\leftrightarrow t$ symmetry:
\begin{cor}[Generalized moment-size duality]\label{cor:gen-gauss_duality}
Let $t,k\leq d$. Then,
\begin{equation}
    \EE_{M\sim \Usub d k}\left| \Perm(M)\right|^{2t} = \EE_{M\sim \Usub d t}\left| \Perm(M)\right|^{2k}.
\end{equation} 
\end{cor}
\par
Later in the paper we use ~\cref{eq:generalized-expansion-formula} to compute $2$nd and $4$th moments of permanents of submatrices of random unitary matrices, as well as $6$th moment of a few small matrices.
\par
Using~\cref{eq:generalized-expansion-formula} we can prove a different lower bound for the moments of permanents in the general case:
\begin{thm}\label{eq:general-lower-bound}
Let $k\leq d$, then 
\begin{align}
    \EE_{M\sim \Usub d k} \left| \Perm(M) \right|^{2t}  \geq \binom{\binom{d+k-1}{k}+t-1}{t}^{-1}.
\end{align}
\end{thm}
The proof of this theorem is technical and will be discussed in~\cref{sec:lbhj}. 
\par 
This bound is much weaker than~\cref{eq:lower-bound-2} when $d\rightarrow \infty$, but becomes more and more relevant as $k$ approaches $d$. To be more quantitative, note the following conjecture suggested by Nick Hunter-Jones and supported by numerical experiments:
\begin{conj}[The Hunter-Jones conjecture]\label{conj:hjConj}
    For the random unitary matrices, the following holds:
    \[\EE_{M\sim U(d)} \left| \Perm(M) \right|^{2t} \approx \frac{t!}{\binom{2d-1}{d}^t}\cong \binom{\binom{2d-1}{d}+t-1}{t}^{-1}.\]
\end{conj}
The value of the permanent moment predicted by~\cref{conj:hjConj} is very close to the lower bound predicted by~\cref{eq:general-lower-bound} for $k=d$. This suggests approximate tightness of~\cref{eq:lower-bound-2} for $k=d$. 
\par 
The Hunter-Jones conjecture further suggests that the distribution of the square permanent of random unitary matrices should be close to an exponential distribution with mean $t!/\binom{2d-1}{d}$. This is because the $t$-th moment of such distribution is ${t!}/{\binom{2d-1}{d}^t}$, matching the value predicted by the conjecture.
\par
Lastly, we point out that~\cref{eq:generalized-expansion-formula} can be easily modified to the case of determinants, and in that case, the sum greatly simplifies. In~\cref{sec:detmom} we prove the following exact formula:
\begin{thm}[Moments of determinants of minors of random unitary matrices]\label{thm:mom-det}
Let $k\leq d$ be integers. Then, 
\begin{align}
\EE_{U \sim \Usub{d}{k}} \left| \det \left(M\right) \right|^{2t} =\prod_{i=1,j=1}^{k,t}\frac{i+j-1}{(d-k)+i+j-1}.
\end{align}
\end{thm}
\section{Background on representation theory}
In this section, we review the essential representation theory background and tools that we use in the rest of the paper. To avoid a very lengthy section, we omit the proofs and refer interested readers to the relevant resources~\cite{fulton2013representation,fulton1997young,bump2004lie, howe1995perspectives}.
\subsection{Basics of the representation theory}
Let us start with an abstract group $G$\footnote{We assume basic knowledge of group theory.}. The group is naturally defined by its product rule, i.e., given any two elements $g,h \in G$, one can find a third element $k \in G$ which is the result of composition of $g$ and $h$. We write this relation as $k = g h$. In many cases, one can assign a matrix to each group element, such that the standard matrix multiplication of the matrices mimics the abstract group multiplication laws. More precisely, for an integer $n$, one aims to find a map $\rho$ from the group $G$ to the space of $n\times n$ matrices, such that,
\[ \rho(g) \rho(h) = \rho(g h), \text{for all }g,h \in G.\]
Multiplication on the left-hand-side is the matrix multiplication, while the multiplication of the right-hand-side is the group composition. We call a map $\rho$ a \emph{representation of the group $G$} when it satisfies the above property.
\par
If $\rho_1$ and $\rho_2$ are representations of the group $G$, then their direct sum $\rho_1 \oplus \rho_2$ is also a representation. Conversely, given a generic representation, one can try to decompose it into the direct sum of smaller representations until further decomposition is not possible. In this way, one always ends up with representations that cannot be decomposed into smaller ones, which are called the \emph{irreducible representations} or \emph{irreps}\footnote{We ignore the important distinction between irreducible and indecomposable representations, as the notions match for the representations that we study in this paper.}.
\par
If $\rho$ is a representation of dimension $n$, i.e., it maps any element of $G$ to a $n\times n$ matrix, then one naively needs $|G|n^2$ numbers to uniquely identify a representation. However, there is a more compact and elegant way of describing the representations using~\emph{characters}, which are nothing but the trace of the representation matrices:
\[ \chi(g) := \tr \rho(g).\]
The character $\chi$ is a vector in a complex vector space of dimension $|G|$\footnote{For infinite groups, it might be more convenient for some to think of the character as a complex valued function defined on $G$.}, where its value on the identity element is equal to the dimension of the representation matrix.  Surprisingly, one can always reconstruct a representation from its characters. 
\par
The characters of irreducible representations form an orthogonal set of vectors in the $|G|$ dimensional complex vector space with respect to the normalized inner product
\[ \EE_{g}\, \overline{\chi_i(g)} \chi_j(g)=\frac{1}{|G|} \sum_{g\in G} \overline{\chi_i(g)} \chi_j(g) = \delta_{i,j},\text{ for irreducible character $\chi_i$ and $\chi_j$.}\]
Moreover, the character of the direct sum of representations is simply the sum of the characters of the representations. Hence, one can always read the irreducible content of a representation by writing its character vector in the orthogonal basis given by the irreducible characters. If the decomposition of a representation $\rho$ with character $\chi$ has $n_i$ copies of an irrep $\rho_i$ with character $\chi_i$, then,
\[\EE\, \overline{\chi(g)}\chi_i(g) = n_i.\]
We call $n_i$ the \emph{degeneracy} of the representation $\rho_i$ in $\rho$. 
\par
Some important representations are the followings:
\begin{itemize}
\item The \emph{trivial} representation, $\rho_{\text{trivial}}$, which assigns $1$ to all group elements: $\rho_{\text{trivial}}(g) = 1$ for all $g \in G$. 
\item The \emph{regular} representation: This representation acts on a vector space of dimension $|G|$ with basis elements indexed by the group elements, $\{e_g\}$. The representation is simply defined as: \[\rho_{\text{reg}} (g) e_ {g'} = e_{gg'}.  \]
The character of the regular representation is $\chi_{\text{reg}}(g) = |G|\delta_{g,e}$, where $e$ is the identity element of the group. For any irreducible character $\chi_i$ with dimension $d_i$ and degeneracy $n_i$ in the regular representation, we have,
\begin{equation} \label{eq:reg-dec}
n_i = \frac{1}{|G|} \sum_{g} \overline{\chi_{\text{reg}}(g)} \chi_i(g) = \chi_i(e) = d_i.
\end{equation}
Therefore, any irrep appears in the regular representation, and it appears with a degeneracy equal to its dimension.
\end{itemize}
\subsection{The symmetric group and the unitary group}\label{sec:basym}
This paper is mostly concerned with two groups: the symmetric group on $t$ elements, $S_t$, and the group of unitary $d\times d$ matrices, $U(d)$. 
\par
We start by discussing the Young diagrams that index the irreps of both groups. A~\emph{Young diagram} is an ordered partition of an integer, which is usually presented as a collection of boxes as shown in~\cref{fig:yt}. left.
We usually indicate Young diagrams with Greek letters, and use ``$\vdash$'' or ``$|\cdot|$'' to show the number of boxes, e.g., $\lambda \vdash n$ or $|\lambda|=n$. We use subscripts to identify individual parts of the partitions. For example, $\lambda_i$ is the $i$-th largest element in the partition, or equivalently, length of the $i$-th row of the corresponding diagram. \emph{Depth} of the diagram, $l(\lambda)$, is defined as the number parts in the partition or the number of rows in the corresponding diagram (see~\cref{fig:yt}. left). Lastly, we define $\tilde \lambda$ to be the transpose Young diagram, where the rows and columns are exchanged. See~\cref{fig:yt}. right.
\par
\begin{figure}
    \centering
    \includegraphics[width=10cm]{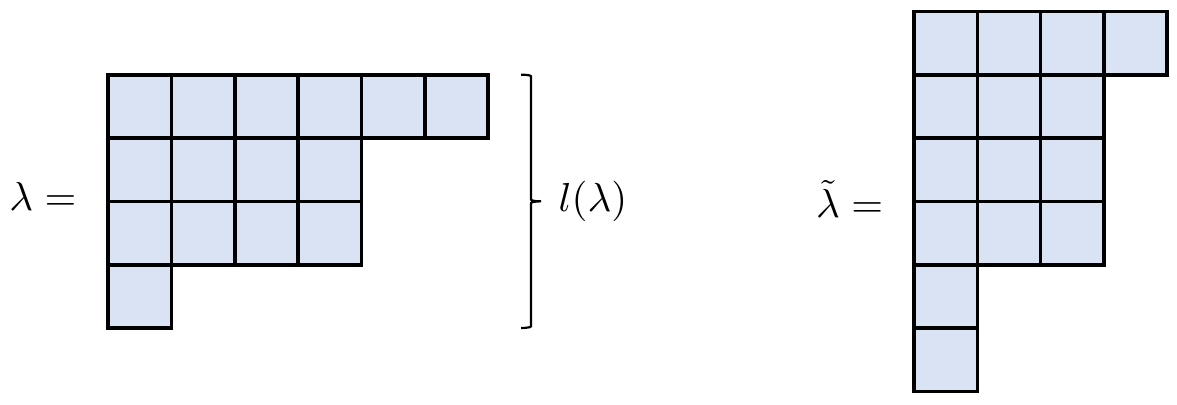}
    \caption{{\bf Left.} The Young diagram $\lambda$ with $15$ boxes. We indicate the number of boxes by $|\lambda|=15$ or $\lambda \vdash 15$. The corresponding partition is $\lambda = (6,4,4,1)$, and the depth of the diagram is the number of rows: $l(\lambda)=4$. {\bf Right}. $\tilde \lambda$, the transpose of $\lambda$.}
    \label{fig:yt}
\end{figure}
\par 
A Young diagram can be filled with integers to form a \emph{Young tableau}. There are two important types of Young tableaux:
\begin{enumerate}
    \item \emph{Standard Young tableau} is a Young diagram filled with integers $1,\cdots, n$, each one appearing once. We also assume that the numbers in each row and column are strictly increasing (see~\cref{fig:yt2} .left). For any Young diagram $\lambda$, the number of standard tableaux is given by the \emph{hook length formula}: for every box in the Young diagram, count the number of boxes directly below, or directly to the to the right of that box (including the box itself). Multiply the numbers assigned to all of the boxes and call the result $\hook(\lambda)$. Number of standard Young tableaux, $f^\lambda$, is given by~\cite{fulton1997young,fulton2013representation}:
    \begin{equation}\label{eq:hook_length_definition}
        f^\lambda = \frac{n!}{\hook(\lambda)} = \frac{n! \prod_{1\leq i<j\leq l(\lambda)}(\lambda_i-\lambda_j) }{\lambda_1! \lambda_2! \cdots \lambda_{l(\lambda)}! }.
    \end{equation} 
    \item \emph{Semi-standard Young tableau} is a Young diagram of $n$ boxes filled with a subset of numbers $1,\cdots, d$, for some integer $d$. This time, the constraint is that the numbers appearing in each row are weakly increasing, while the numbers appearing in each column are strictly increasing. See~\cref{fig:yt2}. right. The number of such tableaux is given by the Weyl dimension formula~\cite{fulton1997young}:
    \begin{equation}\label{eq:weyl_dim_formula_definition}
        \WD_\lambda(d) = \prod_{1\leq i<j\leq d} \frac{\lambda_i - \lambda_j+j-i}{j-i} = f^\lambda \frac{\prod_{(i,j)\in \lambda}(d+j-i)}{n!},
    \end{equation}
    when $(i,j)\in \lambda$ indicates that there is a box in the $i$-th row and $j$-th column of $\lambda$.
\end{enumerate}
\begin{figure}
    \centering
    \includegraphics[width=9cm]{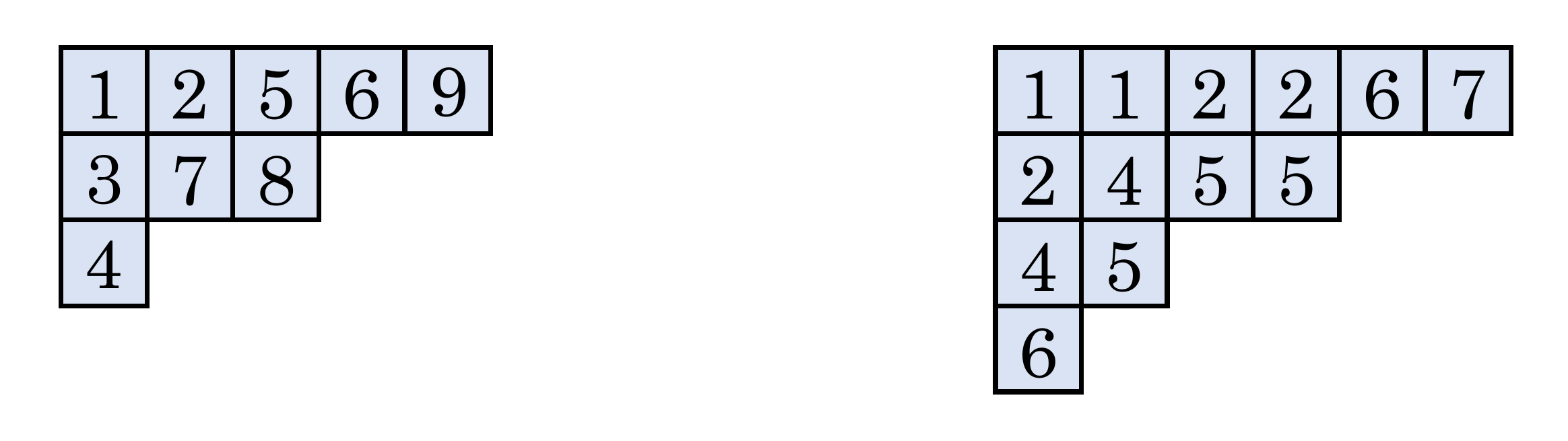}
    \caption{{\bf Left.} A standard Young tableau of $9$ boxes. The diagram is filled with numbers $1,\cdots, 9$, and rows and columns are strictly increasing. {\bf Right.} A semi-standard Young tableau of $13$ boxes, with content $1,\cdots, 7$. Here the rows are weakly increasing, while columns are strictly increasing.}
    \label{fig:yt2}
\end{figure}
The Hook length formula and the Weyl dimension formula are related by $n! \WD_\lambda(d) =f^\lambda\prod_{(i,j)\in \lambda}(d+(j-i))$, see~\cite{fulton1997young}, page 55, Eq. (9). Using this equality, we can derive the following expression that relates $\WD_\lambda(d)$ as $d\rightarrow \infty$ to $f^\lambda$: 
\begin{equation}\label{eq:limiting_behavior_dimension}
  \lim_{d\rightarrow \infty} \frac{\WD_\lambda(d)}{d^n} = \frac{f^\lambda}{n!}  \lim_{d\rightarrow\infty}\prod_{(i,j)\in \lambda}(1+(j-i)/d) = \frac{f^\lambda}{n!}. 
\end{equation}
\par
Representations of the symmetric group, $S_t$, are indexed by the Young diagrams $\lambda \vdash t$. We usually indicate these irreps by $\rho_\lambda(\pi)$, for $\pi \in S_t$. Dimension of $\rho_\lambda(\pi)$ is $f^\lambda$ which we defined in \cref{eq:hook_length_definition}. There are a few notable representations:
\begin{itemize}
    \item The \emph{trivial representation}, which is the irreducible representation corresponding to $\lambda = (t)$. 
    \item The \emph{sign representation} corresponds to $\lambda = (1,1,1,\cdots ,1)$ (with $t$ ones). This is a one-dimensional irreducible  representation that assigns $\pm 1$ to each permutation according to its sign. 
    \item The \emph{standard representation} is the conventional matrix representation of the permutations. It assigns a $t\times t$ matrix $\rho_{\text{std}}$ to each $\pi \in S_t$, where $\rho_{\text{std}} (\pi) \ket i =\ket {\pi(i)}$. This representation is reducible and decomposes to a trivial irrep and one copy of $\rho_{(t-1,1)}$. 
\end{itemize}
One can use Young diagrams to construct the representations matrices. One way to do so is by defining the Young symmetrizer. First, define two subgroups of $S_t$ as follows:
\begin{itemize}
    \item \emph{The row preserving subgroup $R_\lambda$.} Fill the Young diagram $\lambda$ with numbers $1$ to $t$, placing $1, \cdots, \lambda_1$ in the first row, $\lambda_1+1,\cdots, \lambda_1+\lambda_2$ in the second row, and so on. With this filling, define $R_\lambda$ to be the subgroup of $S_t$ that preserves rows of this Young tableau. It is easy to see that $|R_\lambda| =\prod_{i} \lambda_i!$.
    \item \emph{The column preserving subgroup $C_\lambda$.} Consider the Young tableau defined above. Similar to $R_\lambda$, we define the column preserving subgroup to be the subgroup of $S_t$ that preserves the columns of $\lambda$.
\end{itemize}
The Young symmetrizer, $c_\lambda$, is
\[ c_\lambda = \sum_{\pi \in R_\lambda, \sigma \in C_\lambda} \text{sgn}(\sigma) \sigma \pi,\]
where $\pi$ and $\sigma$ are considered to be abstract group elements (i.e., the elements of the group algebra), or equivalently, the matrices of the regular representation of the group. ``$\text{sgn}$'' is the sign of the permutation.
\par
We can observe that $c_\lambda$ is an unnormalized and non-orthogonal projector:
\begin{equation}\label{eq:cchook}
    c_\lambda c_\lambda = \hook(\lambda) c_\lambda,
\end{equation} 
where the image of $c_\lambda$ is the irrep $\rho_\lambda$ of the symmetric group. See, e.g.,~\cite{fulton2013representation} chapter 4 for details.
\par
The characters of any element of the symmetric group is identified by its conjugacy class, which is in turns determined by the cycle type of the permutation. For $\pi \in S_t$, suppose that $\pi$ has cycles of length $\mu_1 \geq \mu_2 \geq \cdots \geq \mu_l$. Then $\mu=(\mu_1,\mu_2,\cdots,\mu_l)$ is a partition of $t$, or equivalently, a Young diagram of size $|\lambda| =t$. Therefore, the conjugacy classes of the symmetric group are also indexed by the Young diagrams of size $t$. We indicate the character of the symmetric group by the matrix $\chi$:
\[ \chi_{\mu \lambda} = \text{character of the representation }\rho_\lambda \text{ on an element with the cycle type } \mu.\]
\par
Interestingly, the representations of the unitary group are also indexed by the Young diagrams, but a different set of them. The group $U(d)$ has an infinite number of finite dimensional irreducible representations indexed by the Young diagrams of depth $\leq d$. The dimension of the irreducible representation $\sigma^{(d)}_\lambda(U)$ of $U(d)$ is equal to the number of semi-standard Young tableaux and is given by the Weyl dimension formula $\WD_\lambda(d)$. 
\subsection{Symmetric polynomials}
Before proceeding, we need to briefly review the theory of symmetric polynomials. A discussion of this theory is essential for any in-depth analysis of the representation theory of the symmetric and the unitary groups. The material in this section can be found in~\cite{fulton1997young,fulton2013representation}. 
\par 
Consider the space of symmetric polynomials of $d$ variables, i.e., the polynomials $p(x_1,\cdots, x_d)$ that remain invariant under the permutation of their $d$ variables. The goal of this subsection is to remind the reader of different bases for this space, as well as transition matrices between them. 
\par 
\begin{itemize}
    \item \emph{Monomial symmetric polynomials.} Monomial symmetric polynomials are the simplest basis for the space of symmetric polynomials. For any partition $\lambda \vdash t$, define $m_\lambda$ to be 
\[m_\lambda = x_1^{\lambda_1} x_2^{\lambda_2} \cdots x_d^{\lambda_d} + \text{all distinct monomials derived by permuting the variables.} \]
For instance, if $d=3$, $m_{(1,1,0)} = x_1 x_2 +x_1 x_3 + x_2x_3$. It can be easily seen that $m_\lambda$ is a polynomial of degree $|\lambda|$, and the set $\{m_\lambda\}$ forms a basis for all symmetric polynomials.
\item \emph{Power-sum symmetric polynomials.} First, define $p_i = x_1^i + \cdots + x_d^i$. Next, we define $p_\lambda$ for a partition $\lambda$ to be 
\[ p_\lambda = p_{\lambda_1} p_{\lambda_2} \cdots p_{\lambda_{l(\lambda)}}.\]
The set of all polynomials $p_\lambda$ with $\lambda \vdash t$ forms a basis for the space of symmetric polynomials of degree $t$ in $d$ variables.
\item \emph{Elementary symmetric polynomials.} Define $e_i = \sum_{1\leq l_1<l_2<\cdots<l_i\leq d} x_{l_1}x_{l_2}\cdots x_{l_i}$. For instance, $e_2 = x_1 x_2 + x_1x_3 +\cdots +x_{d-1}x_d$. We define the elementary symmetric polynomial $e_\lambda$ as follows:
\begin{equation*} 
e_\lambda = e_{\lambda_1} e_{\lambda_2} \cdots e_{\lambda_{l(\lambda)}}.
\end{equation*}
Again, these polynomials form a basis for the space of symmetric polynomials. The change of basis between $m_\mu$ and $e_\mu$ is given by
\begin{equation}\label{eq:transition_m_to_e}
    e_\lambda = \sum_{\mu} \text{IB}_{\lambda\mu} m_\mu,
\end{equation} 
where $\text{IB}_{\lambda\mu}$ is the number of $0,1$ matrices with the row sums of $\mu$ and the column sums given by $\nu$. 
\item \emph{Complete homogeneous symmetric polynomials.} 
Let $ h_i  = \sum_{1\leq l_1\leq l_2\leq \cdots\leq l_i\leq d} x_{l_1}x_{l_2}\cdots x_{l_i}$, and $h_\lambda =  h_{\lambda_1} \cdots h_{\lambda_{l(\lambda)}}$. The change of basis from the symmetric monomials to this basis is given by $\text{IM}_{\lambda\mu}$ coefficients that count the number of integer valued matrices with the row sum equal to $\lambda$ and the column sum of $\mu$:
\begin{equation}\label{eq:transition_m_to_h}
h_\lambda = \sum_{\mu} \text{IM}_{\lambda\mu} m_\mu. 
\end{equation}
\item \emph{Schur polynomials.} Lastly, we discuss the arguably most important polynomial, the Schur polynomial $s_\lambda$. One way to define the Schur polynomials is by starting with the determinantal expression
\begin{align*}
a_{(\lambda_1+d-1, \lambda_2+d-2, \dots , \lambda_d)} (x_1, x_2, \dots , x_d) :=
\det \left[ \begin{matrix} x_1^{\lambda_1+d-1} & x_2^{\lambda_1+d-1} & \dots & x_d^{\lambda_1+d-1} \\
x_1^{\lambda_2+d-2} & x_2^{\lambda_2+d-2} & \dots & x_d^{\lambda_2+d-2} \\
\vdots & \vdots & \ddots & \vdots \\
x_1^{\lambda_d} & x_2^{\lambda_d} & \dots & x_d^{\lambda_d} \end{matrix} \right].
\end{align*}
In particular, $a_{(d-1, d-2, \dots , 0)} (x_1, x_2, \dots , x_d)  = \prod_{1\leq i <j \leq d} (x_i-x_j) $ is the Vandermonde determinant. Then, the Schur polynomial is defined as
\begin{equation}\label{eq:schur-def-1}
    s_\lambda  = \frac{a_{(\lambda_1+d-1, \lambda_2+d-2, \dots , \lambda_d)}}{a_{(d-1, d-2, \dots , 0)}}.
\end{equation}
There is a different way of defining the Schur polynomial using the semi-standard tableaux:
\begin{equation}\label{eq:schur-def-2}
    s_\lambda = \sum_{T \text{ is a semi-standard filling of }\lambda} \quad \prod_{i=1}^d x_i^{\text{number of times }i \text{ appears in }T}.
\end{equation}
It is not obvious from this definition that $s_\lambda$ is symmetric. The degree of $s_\lambda$ is $|\lambda|$.
\par
Next, we discuss the transition matrices between the different polynomial bases and the Schur basis. If $|\lambda|= |\mu|$, define the \emph{Kostka} coefficient $K_{\lambda\mu}$ to be the number of Young tableaux which are the semi-standard fillings of the Young diagram of shape $\lambda$ with content $\mu$ (i.e., there are $\mu_1$ number of $1$'s, $\mu_2$ number of $2$'s, etc., in the tableau). Then, one can see that 
\begin{equation}\label{eq:transition_m_to_s}
s_\lambda = \sum_\mu K_{\lambda\mu} m_\mu.
\end{equation}
\par
Another important example is the transition matrix between $p_\lambda$ and $s_\lambda$. This is given by the characters of the symmetric group:
\begin{equation}\label{eq:transition_s_to_p}
p_\lambda = \sum_\mu \chi_{\lambda \mu } s_\mu.
\end{equation}
 It is possible to use this relation to recover the Frobenius formula for the character of the symmetric group~\cite{fulton2013representation}.
Lastly, we have the following two relations:
\begin{align}
    &h_\lambda = \sum_\mu K_{\mu \lambda} s_\mu, \text{ and,} \label{eq:transition_s_to_h}\\
    &e_\lambda = \sum_\mu K_{\tilde \mu \lambda } s_\mu. \label{eq:transition_s_to_e}
\end{align}
\end{itemize}
In addition to what is mentioned above, the Schur polynomials have another significance: they are the characters of the unitary group $U(d)$. If a unitary $U \in U(d)$ has eigenvalues $a_1,a_2,\cdots,a_d$, the character of the irrep $\sigma^{(d)}_\lambda(U)$ is simply given by $s_\lambda(a_1,a_2,\cdots,a_d)$. 
\par
Interestingly, the transition matrices can be used to derive formulas for $\text{IB}_{\lambda\mu}$ and $\text{IM}_{\lambda\mu}$. Combining~\cref{eq:transition_s_to_e} and~\cref{eq:transition_m_to_s}, we get $e_\lambda = \sum_\mu K_{\tilde \mu \lambda} s_\mu= \sum_{\mu\nu} K_{\tilde\mu\lambda} K_{\mu\nu} m_\nu$. From~\cref{eq:transition_m_to_e} we also know that $e_\lambda = \sum_{\nu} \text{IB}_{\lambda\nu}m_{\nu}$. Hence, we find a simple formula for the number of $0-1$ matrices with prescribed row and column sums:
\begin{equation}\label{eq:IB_formula}
    \text{IB}_{\mu\nu} = \sum_{\lambda} K_{\tilde \lambda\mu} K_{\lambda\nu}.
\end{equation}
Similarly,~\cref{eq:transition_m_to_h} gives the number of integer matrices with prescribed row and column sums as
\begin{equation}\label{eq:IM_formula}
    \text{IM}_{\mu\nu} = \sum_{\lambda} K_{ \lambda\mu} K_{\lambda\nu}.
\end{equation}
\par 
We end this section by discussing a condition for vanishing of the Kostka numbers. Let us define the lexicographical order on the Young diagrams:
\[\mu > \nu \text{ if the first non zero }\mu_i-\nu_i\text{ is positive}. \]
Then, one can see that 
\begin{equation}\label{eq:zerokost} K_{\mu\nu} = 0 \quad\text{ if } \mu<\nu.
\end{equation}
Moreover, $K_{\mu\mu}=1$ for all $\mu$.
This means that if we order the Young diagrams by the lexicographical order, the Kostka matrix is a lower triangular matrix with identity on the diagonals. The inverse of the matrix $K$, which we indicate by $(K^{-1})_{\mu\nu}$, exists and is a lower triangular matrix as well:
\begin{equation}\label{eq:zerokostinv} (K^{-1})_{\mu\nu} = 0 \quad\text{ if } \mu<\nu.
\end{equation}
\subsection{Schur-Weyl duality, Howe duality, and the left-right action on the regular representation}\label{sec:rep_backgroun_dualities}
In this section, we discuss some of the deep dualities and results in the representation theory of $S_t$ and $U(d)$. We will extensively use these results in the rest of this manuscript.
\subsubsection{\texorpdfstring{$U(d) \text{ and } S_t$}{}: the Schur-Weyl duality} \label{sec:sw} Let $\Hil = \CC^d$. For any $U\in U(d)$, the operator $U^{\ot t}$ acts on $\Hil^{\ot t}$. Similarly, for any $\pi \in S_t$, $\pi$ acts on $\Hil^{\ot t}$ through a representation $r$ that permutes the $t$ replicas of the vector space:
\begin{equation}\label{eq:def_r}
r(\pi) \,\left[{\psi_1} \ot  {\psi_2 }\ot \cdots \ot  {\psi_d} \right]= {\psi_{\pi^{-1}(1)}} \ot  {\psi_{\pi^{-1}(2)}} \ot \cdots \ot {\psi_{\pi^{-1}(d)} }\quad \text{for }\ket{\psi_i} \in \Hil.
\end{equation}
The action of the permutation group and the symmetric group commute, and therefore, one can decompose $\Hil^{\ot t}$ into irreducible representations of both groups. The celebrated Schur-Weyl duality gives the form of this decomposition and states that:
\begin{equation}\label{eq:def_SW}
    \Hil^{\ot t} = \bigoplus_{\lambda \vdash t,\, l(\lambda) \leq d} V^{\text{Sym}}_\lambda \ot V^{U(d)}_\lambda,  
 \end{equation} 
 where $V^{\text{Sym}}_\lambda$ is the vector space that the symmetric group acts on with the irrep $\rho_\lambda$, and $V^{U(d)}_\lambda$ is a space that the irrep $\sigma^{(d)}_\lambda$ of $U(d)$ acts on.  
In other words, the vector space is decomposed into a direct sum of so called \emph{isotypic} subspaces $V^{\text{Sym}}_\lambda \ot V^{U(d)}_\lambda$, where the groups $S_t$ and $U(d)$ act as a simple tensor product. This duality lies at the heart of the representation theory of symmetric and unitary groups.
\par 
\subsubsection{\texorpdfstring{$U(d_1) \text{ and }  U(d_2)$}{}: the Howe duality}\label{sec:howe} Consider a slightly different scenario and let $\Hil = \CC^{d_1} \ot \CC^{d_2}$. The groups $U(d_1)$ and $U(d_2)$ have commuting actions on $\Hil$, where $U(d_1)$ acts on $\CC^{d_1}$ and $U(d_2)$ acts on $\CC^{d_2}$ tensor factor. One can consider $t$-th symmetric power of $\Hil$, i.e., the image of $\Hil^{\ot t}$ under the projection to the symmetric subspace $\frac{1}{t!}\sum_{\pi \in S_t} r(\pi)$ (see~\cref{eq:def_r} for the definition of $r(\pi)$). We call this space $\Sym_t(\Hil)$. The commuting actions of $U(d_1)$ and $U(d_2)$ extend to the symmetric subspace, and $\Sym_t$ decomposes into isotypic subspaces
\begin{equation}\label{eq:howe_duality}
    \Sym_t( \CC^{d_1} \ot \CC^{d_2}) = \bigoplus_{\lambda \vdash t, \,l(\lambda) \leq \min(d_1,d_2)} V^{U(d_1)}_\lambda \ot V^{U(d_2)}_\lambda.
\end{equation}
Again, $V^{U(d)}_\lambda$ is a subspace that the irrep $\sigma_\lambda^{(d)}$ acts on. This duality and a number of similar dualities involving other classical groups are usually called Howe dualities~\cite{howe1995perspectives,goodman2009symmetry}.~\Cref{eq:howe_duality} and the Schur-Weyl duality~\cref{eq:def_SW} are very similar in nature, and in fact, can be derived from each other (see~\cite{howe1995perspectives} and~\cref{sec:expproof}).  
\subsubsection{\texorpdfstring{$S_t \text{ and }  S_t$}{}: the left-right action on the regular representation}\label{sec:pw} Here we discuss the decomposition of the regular representation. This is the most basic of the relations that we discuss in this section, and generically, is not considered a duality. We include it here as it has a similar structure to the above dualities.
\par
Consider a $t!$ dimensional vector space with a basis indexed by the group elements $\{e_ \pi\}_{\pi\in S_t}$. The left regular representation act on this space by $\rho_L (\alpha) e_ \pi =  e_ {\alpha \pi}$. Similarly, the right regular representation $\rho_R$ acts by $\rho_R(\alpha) e_ \pi  = e_ {\pi \alpha^{-1}}$. It is immediate to see that these two actions commute, and $\CC^{t!}$ bears a representation of $S_t\times S_t$. In a similar fashion as the above dualities, one can see that $\CC^{t!}$ decomposes in subspaces carrying a tensor factorization of the irreps of left and right symmetric groups:
\[ \CC^{t!} = \bigoplus_{\lambda \vdash t} V^{\text{Sym}}_\lambda \ot V^{\text{Sym}}_\lambda.\]
The first tensor factor in the summand corresponds to the action of the permutation group form the left, while the second one corresponds to the right action $\rho_R$.

\section{From permanents to the symmetric group and the Howe duality}
Equipped with the basics of the representation theory, we can start our analysis of the permanent moments. We still need more tools, which will be introduced as we proceed. 
\par
In this section, we prove one of our central results (\cref{eq:generalized-expansion-formula}) that provides a framework for systematically computing the moment of permanents of minors of random unitary matrices. As we will observe, our result can be used to give an alternative proof of the Gaussian expansion formula~\cref{eq:expansion-formula}.
\par
Recall $\Usub{d}{k}$, the ensemble of leading $k\times k$ minors of Haar random $d\times d$ unitary matrices\footnote{Choosing the leading minor has no significance, and any other minor could be used as well. This is because the random unitary ensemble is invariant under the permutation of rows and columns.}.  
We are interested in computing $\EE _{M \sim \Usub{d}{k}}\left|\Perm M\right| ^{2t}$, and
as a first step, we wish to write it in a different form. Using the simple identity $\Perm (A\ot \id_t) = \Perm(A)^t$, we have
\[ \EE _{M \sim \Usub{k}{d}}\left|\Perm M\right| ^{2t} = \EE _{M \sim \Usub{k}{d}}\left|\Perm (M \ot \id_t )\right| ^{2}.  \]
Let $\Hil = \CC^{dt}=\CC^d \ot \CC^t$. One basis for $\Hil$ is given by $e_ i \ot e_ j := e_{i,j}$, where $i \in \{1,\cdots, d\}$ and $j \in \{1,\cdots, t \}$. Consider a vector $\Omega \in \Hil^{\ot kt}$ defined as
\begin{equation}\label{eq:omega-def}  
\Omega = \bigotimes_{ 1 \leq i \leq k,1 \leq j \leq t} e_{i,j}  = \arraycolsep=1.2pt\def\arraystretch{1.5} \begin{array}{cccccccccc}
e_ {1,1} &\ot& e_ {1,2}&\ot& e_ {1,3} &\ot& \cdots &\ot& e_ {1,t}&\ot\\
e_ {2,1} &\ot& e_ {2,2}&\ot& e_ {2,3} &\ot& \cdots &\ot& e_ {2,t}&\ot\\
e_ {3,1} &\ot& e_ {3,2}&\ot& e_ {3,3} &\ot& \cdots &\ot& e_ {3,t}&\ot\\
\vdots && \vdots&& \vdots && \ddots && \vdots&\\
e_ {k,1} &\ot& e_ {k,2}&\ot& e_ {k,3} &\ot& \cdots &\ot& e_ {k,t}&\\
\end{array}.
\end{equation}
Let $V$ and $U$ be $t\times t$ and $d\times d$ matrices, respectively, and define $M$ to be the leading $k\times k$ minor of $U$. The permutation group $S_{tk}$ acts on the set of $kt$ pairs $\{(i,j)\}_{1\leq k, 1\leq t}$, and one can see that:
\[ \Perm (M\ot V) = \sum_{\pi \in S_{tk}}\,\,\prod_{i=1, j = 1, \pi((i,j)) = (r,s)}^{t,k} U_{ir}V_{js} = \langle \Omega_{\Sym}, \Sym_{kt}(U\ot V)\Omega_{\Sym}\rangle,  \]
where $\Omega_{\Sym} = \left(\frac{1}{\sqrt {(kt)!}} \sum_{r(\pi) \in S_{tk}} \pi\right) \Omega$ is the symmetrized and normalized version of $\Omega$ (see~\cref{eq:def_r} for the definition of $r$), and $\langle \cdot,\cdot \rangle$ is the standard complex inner product.

Using the Howe duality (\cref{eq:howe_duality}), we can decompose $\Sym_{kt} (U\ot V)$ into irreps of $U(d)$ and $U(t)$:
\[\Sym_{kt} (\CC^d\ot \CC^t) = \bigoplus_{\lambda \vdash kt, l(\lambda) \leq \min (d,t)} V_\lambda^{U(d)}\ot V_\lambda^{U(t)}, \]
where $U(d)$ acts on $V_\lambda^{U(d)}$ and $U(t)$ acts on $V_\lambda^{U(t)}$. Written in terms of the representation matrices, we have 
\begin{align}\label{eq:howe_body}
    \Sym_{kt} (U\ot V) = \bigoplus_{\lambda \vdash kt, l(\lambda) \leq \min (d,t)} \sigma_\lambda^{(d)}(U) \ot \sigma_\lambda^{(t)}(V).
\end{align}
For our calculations we set $V = \id_t$, and therefore,
\begin{multline*}
\EE _{U \sim \Usub d k}\left|\Perm (M \ot \id_t )\right| ^{2} =\\
\EE _{U \sim \Usub d k}\, \langle \Omega_{\Sym} ,\Sym_{kt}(U\ot \id_t) \Omega_{\Sym}\rangle \,\langle \Omega_{\Sym}, \Sym_{kt}(U^\dagger\ot \id_t) \Omega_{\Sym}\rangle= \\
\sum_{\lambda \vdash kt, l(\lambda) \leq \min (d,t)} \EE _{U \sim \Usub d k} \,\langle \Omega_{\Sym}, \sigma_\lambda(U) \ot \sigma_\lambda(\id_t)\Omega_{\Sym} \rangle\,\langle \Omega_{\Sym},  \sigma_\lambda(U^\dagger) \ot \sigma_\lambda(\id_t) \Omega_{\Sym}\rangle.
\end{multline*}
To derive the third line from the second line, we used~\cref{eq:howe_body} for $\Sym_{kt}(U\ot \id_t)$ and $\Sym_{kt}(U^\dagger\ot \id_t)$, and used Schur's lemma to argue that only the terms with coinciding irreps survive. Therefore, we only have one sum over the irreps.
\par
Let us define the projection of $ \Omega_{\Sym}$ to the isotypic subspace indexed by $\lambda$, i.e., $ \Omega_{\Sym}^{\lambda} := \sigma_\lambda(\id) \ot \sigma_\lambda(\id_t)\,\Omega_{\Sym}$, and define $P_{\Omega_{\Sym}^{\lambda}}$ to be the projector to $\Omega_{\Sym}^{\lambda}$ (In conventional quantum mechanics bra-ket notation, $P_{\Omega_{\Sym}^{\lambda}} = \ketbra{\Omega^\lambda_{\Sym}}{\Omega^\lambda_{\Sym}}$, when $\ket {\Omega^\lambda_{\Sym}} :=\Omega^\lambda_{\Sym}$). Using Schur's lemma, it is straightforward to see that 
\begin{multline}\label{eq:first_expansion}
    \sum_{\lambda \vdash kt, l(\lambda) \leq \min (d,t)} \EE _{U \sim \Usub d k} \,\langle \Omega_{\Sym}, \sigma_\lambda(U) \ot \sigma_\lambda(\id_t)\Omega_{\Sym} \rangle\,\langle \Omega_{\Sym},  \sigma_\lambda(U^\dagger) \ot \sigma_\lambda(\id_t) \Omega_{\Sym}\rangle = \\ 
    \sum_{\lambda \vdash kt, l(\lambda) \leq \min (d,t)} \frac{1}{\WD_\lambda(d)} \tr \left[ \left( \tr_{ V_\lambda^{U(d)}}(P_{\Omega^\lambda_{\Sym}}) \right)^2\right].
\end{multline}
Where $ \tr_{ V_\lambda^{U(d)}}$ is the partial trace. Let us explain the tensor factorization once more in details: in order to compute the moments of permanents, one has to consider the isotypic subspaces individually, and consider the restriction $\Omega_{\Sym}^{\lambda}$ of $\Omega_{\Sym}$. According to Howe's duality, this vector lives in the tensor product vector space  $V_\lambda^{U(d)}\ot V_\lambda^{U(t)}$, and one can compute the partial trace of the projector to this vector with respect to one of these tensor factors: $ \tr_{ V_\lambda^{U(d)}}(P_{\Omega^\lambda_{\Sym}})$. The rest is straightforward algebra.
\par 
After simplifying this expression, we obtain the following result:
\begin{thm} \label{thm:main_expansion_formula}
Consider a $k \times t$ rectangle and let $R$ and $C$ be the row preserving and the column preserving subgroups (see~\cref{dfn:row-col-pres}). Then 
\begin{align}\label{eq:main_avg_formula}
\EE_{M\sim \Usub d k} \left| \Perm(M) \right|^{2t} = \sum_{\lambda \vdash kt, l(\lambda)\leq \min(k,t)} \frac{1}{\WD_\lambda(d)} \left(\frac{f^\lambda}{(kt)!}\right)^2\tr \left[\rho_\lambda(RCRC) \right],
\end{align}
where $f^\lambda$ is the dimension of the irrep of $S_{tk}$ with the Young diagram $\lambda$, $\WD_\lambda(d)$ is the dimension of the irrep $U(d)$ with the same Young diagram, and $R = \sum_{r \in R} r$, and $C= \sum_{c\in C} c$ according to~\cref{rem:notational-remark} 
\end{thm}
The proof of this theorem is lengthy and will be reported in the next subsection. Before going to the proof we discuss some of the consequences of~\cref{thm:main_expansion_formula}.
\par
We can immediately derive the formula for the average of permanents of Gaussian matrices from~\cref{eq:main_avg_formula}. Consider the limit $d\rightarrow \infty$ while keeping $k$ fixed. Matrix elements of $M\sim \Usub d k$ become more and more independent and will be distributed according to the Gaussian distribution. However, it is easy to see that the standard deviation of the individual matrix elements of $M$ is $\frac{1}{\sqrt d}$, therefore, we need to normalize $M$ by a factor of $d^{1/2}$ to obtain a random i.i.d. complex Gaussian matrix (where the standard deviation of individual elements is $1$). Hence, we have the following easy remark:
\begin{rem}\label{rem:gauss_usub}
Let $M'$ be the leading $k\times k$ minor of $d\times d$ Haar random unitary matrices. Let $M$ be $d^{1/2} M'$, as $d\rightarrow\infty$. Then, the different matrix elements of $M$ are i.i.d Gaussians with standard deviation of $1$. 
\end{rem}
We conclude that, 
\begin{multline*} 
\EE_{M \sim \Gauss k} \left| \Perm(M) \right|^{2t} = \lim_{d\rightarrow \infty} \EE_{M \sim \Usub d k} \left| \Perm(d^{1/2} M) \right|^{2t} =
\\
 \sum_{\lambda} \left[\lim_{d\rightarrow \infty} \frac{d^{kt/2}}{\WD_\lambda(d)} \right] \left(\frac{f^\lambda}{(kt)!}\right)^2\tr \left[\rho_\lambda(RCRC) \right] =  \frac{1}{(kt)!}\sum_{\lambda} f^\lambda\tr \left[\rho_\lambda(RCRC) \right] =\frac{\tr (RCRC)}{(kt)!},
\end{multline*}
where the first equality follows from~\cref{rem:gauss_usub}, the second one uses~\cref{eq:main_avg_formula}, the third one is consequence of~\cref{eq:limiting_behavior_dimension}, and the last one is simply the decomposition of the regular representation (see, e.g.,~\cref{eq:limiting_behavior_dimension} and the text that follows). This provides an alternative proof of~\cref{thm:comb}.
\par 
We also obtain the generalized version of the size-moment duality:
\begin{cor}[size-moment duality]\label{cor:size-moment-general}
From~\cref{thm:main_expansion_formula}, it is easy to observe that as long as $k,t\leq d$,
\[\EE_{M \sim \Usub d k} \left| \Perm(M) \right|^{2t} =\EE_{M \sim \Usub d t} \left| \Perm(M) \right|^{2k}. \]
In other words, one can exchange the matrix dimension and the moment when calculating moments of permanents.
\end{cor}
\subsection{Proof of\texorpdfstring{~\cref{thm:main_expansion_formula}}{}}\label{sec:expproof}
In this section, we start from~\cref{eq:first_expansion} and prove~\cref{thm:main_expansion_formula}. Unlike the previous discussion, we use the Schur-Weyl duality as we need to obtain more detailed information than what the Howe duality naively presents. In fact, the early parts of our calculations mimic the proof of the Howe duality from the Schur-Weyl duality~\cite{howe1995perspectives}. This section is rather detailed, and we encourage the reader to skip it in the first reading of the manuscript.
\par
We wish to analyze $P_{\Omega_{\Sym}^\lambda}$ (see the text above~\cref{eq:first_expansion}) as our first step. Recall the definition of our initial vector space $\Hil = \CC^d \ot \CC^t$, and the definition of $\Omega$ in~\cref{eq:omega-def}. 

One may naturally assume that the tensor factors of $\Hil^{\ot kt}$ are ordered in following way:
\[ \Hil^{\ot kt}  = \arraycolsep=1.2pt\def\arraystretch{1.5} \begin{array}{cccccccccc}
(\CC^d\ot \CC^t) &\ot& (\CC^d\ot \CC^t)&\ot& (\CC^d\ot \CC^t) &\ot& \cdots &\ot& (\CC^d\ot \CC^t)&\ot\\
(\CC^d\ot \CC^t) &\ot& (\CC^d\ot \CC^t)&\ot& (\CC^d\ot \CC^t) &\ot& \cdots &\ot& (\CC^d\ot \CC^t)&\ot\\
\vdots && \vdots&& \vdots && \ddots && \vdots&\\
(\CC^d\ot \CC^t) &\ot& (\CC^d\ot \CC^t)&\ot& (\CC^d\ot \CC^t) &\ot& \cdots &\ot& (\CC^d\ot \CC^t)&
\end{array}.\]
To proceed, we need to re-arrange the tensor factors in a different way:
\[ \Hil^{\ot kt}  = (\CC^d)^{\ot kt} \ot (\CC^t)^{\ot kt}  =  \arraycolsep=1.2pt\def\arraystretch{1.5} \begin{array}{cccccccccc}
\CC^d &\ot& \CC^d&\ot& \CC^d &\ot& \cdots &\ot& \CC^d&\ot\\
\CC^d &\ot& \CC^d&\ot& \CC^d &\ot& \cdots &\ot& \CC^d&\ot\\
\vdots && \vdots&& \vdots && \ddots && \vdots&\\
\CC^d &\ot& \CC^d&\ot& \CC^d &\ot& \cdots &\ot& \CC^d&
\end{array} \quad \ot\quad \begin{array}{cccccccccc}
\CC^t &\ot& \CC^t&\ot& \CC^t &\ot& \cdots &\ot& \CC^t&\ot\\
\CC^t &\ot& \CC^t&\ot& \CC^t &\ot& \cdots &\ot& \CC^t&\ot\\
\vdots && \vdots&& \vdots && \ddots && \vdots&\\
\CC^t &\ot& \CC^t&\ot& \CC^t &\ot& \cdots &\ot& \CC^t&
\end{array}.\]
We call the first sector, i.e., $(\CC^d)^{\ot kt}$, the vector space $\Hil_A$, and the second sector the vector space $\Hil_B$. When acting by $(U\ot \id_t)^{\ot kt}$ on $\Hil^{\ot kt}$, we assume that $U$ acts on $\Hil_A$.
Correspondingly, we write $ \Omega$ as the tensor product of $\Omega^A$ and $\Omega^B$:
\[ \Omega  = \Omega^A \ot  \Omega^B  =  \arraycolsep=1.2pt\def\arraystretch{1.5} \begin{array}{cccccccccc}
e_ {1} &\ot& e_ {1}&\ot& e_ {1} &\ot& \cdots &\ot& e_ {1}&\ot\\
e_ {2} &\ot& e_ {2}&\ot& e_ {2} &\ot& \cdots &\ot& e_ {2}&\ot\\
e_ {3} &\ot& e_ {3}&\ot& e_ {3} &\ot& \cdots &\ot& e_ {3}&\ot\\
\vdots && \vdots&& \vdots && \ddots && \vdots&\\
e_ {k} &\ot& e_ {k}&\ot& e_ {k} &\ot& \cdots &\ot& e_ {k}&\\
\end{array}
\quad \ot\quad 
\begin{array}{cccccccccc}
e_ {1} &\ot& e_ {2}&\ot& e_ {3} &\ot& \cdots &\ot& e_ {t}&\ot\\
e_ {1} &\ot& e_ {2}&\ot& e_ {3} &\ot& \cdots &\ot& e_ {t}&\ot\\
e_ {1} &\ot& e_ {2}&\ot& e_ {3} &\ot& \cdots &\ot& e_ {t}&\ot\\
\vdots && \vdots&& \vdots && \ddots && \vdots&\\
e_ {1} &\ot& e_ {2}&\ot& e_ {3} &\ot& \cdots &\ot& e_ {t}&\\
\end{array}.\]
\par 
According to the Schur-Weyl duality (\cref{sec:sw}) the vector spaces $(\CC^d)^{\ot kt}$ and $(\CC^t)^{\ot kt}$ decompose into the irreps:
\begin{align*}
    &\Hil_A = \bigoplus_{\lambda \vdash kt, l(\lambda)\leq d}  V^{\text{Sym},A}_\lambda \ot V^{U(d),A}_\lambda, \quad \text{and},\\
    &\Hil_B = \bigoplus_{\lambda \vdash kt, l(\lambda)\leq t} V^{\text{Sym},B}_\lambda \ot V^{U(t),B}_\lambda.
\end{align*} 
The symmetric group $S_{kt}$ acts on the $V^{\Sym}$ sectors through $\rho_\lambda^A$ and $\rho_\lambda^B$, $U(d)$ acts on $V^{U(d),A}$ sector via $\sigma_\lambda^A$, and $U(t)$ acts on $V^{U(t),B}$ with $\sigma_\lambda^B$. With this notation, 
\begin{equation}\label{eq:defHil} \Hil^{\ot kt} = \bigoplus_{\lambda_1 \vdash kt, l(\lambda_1)\leq d,\lambda_2 \vdash kt, l(\lambda_2)\leq t}  V^{\Sym,A}_{\lambda_1} \ot V^{U(d),A}_{\lambda_1} \ot V^{\Sym,B}_{\lambda_2} \ot V^{U(t),B}_{\lambda_2} .
\end{equation}
\par 
As usual, we represent the action of the permutation group on the $\Hil_A$ by $r_A(\pi)$, and on $\Hil_B$ by $r_B(\pi)$, for $\pi \in S_{tk}$.
\par 
Now, we re-write $ \Omega_{\Sym}$ as:
\begin{equation}\label{eq:4597} 
\Omega_{\Sym} = \frac{1}{\sqrt{(kt)!}} \sum_{\pi \in S_{kt}} r_A(\pi)\ot r_B(\pi) \,{\Omega} = \frac{1}{\sqrt{(kt)!}} \sum_{\pi \in S_{kt}} \bigoplus_{\lambda_1,\lambda_2} \rho^A_{\lambda_1}(\pi) \ot \sigma^A_{\lambda_1}(\id) \ot \rho^B_{\lambda_2}(\pi) \ot \sigma^B_{\lambda_2}(\id)\,\Omega.
\end{equation}
Consider the vector spaces $V_\lambda^{\Sym,A}$ and $V_\lambda^{\Sym,B}$. One can find real and identical basis vectors $e_i^{A,\lambda}$ and $e_i^{B,\lambda}$, with $i=1,\cdots, f^\lambda$ for these vector spaces. With this, we can define the standard maximally entangled vector between these spaces:
\[\phi^+_\lambda = \frac{1}{\sqrt{f^\lambda}}\sum_{i=1}^{f^\lambda}  e_i^{A,\lambda}\ot e_i^{B,\lambda}. \]
Furthermore, we can define the rank $1$ projector to $\phi^+_\lambda$ as $P^+_\lambda$ (Again, using Dirac's notation $P^+_\lambda = \ketbra{\phi^+_\lambda}{\phi^+_\lambda}$). As a result of Schur's lemma and reality of permutation representations, we can easily see that $\phi^+_\lambda$ is the unique $+1$ eigenvalue of all $\rho^A_\lambda(\pi)\ot \rho^B_\lambda(\pi)$, for all $\pi \in S_{kt}$. Therefore, it is easy to see that
\[ \frac{1}{(kt)!} \sum_{\pi \in S_{kt}} \rho^A_{\lambda_1}(\pi) \ot \rho^B_{\lambda_2}(\pi) = \delta_{\lambda_1,\lambda_2} \,P^+_{\lambda_1} \,.\]
\par 
Substituting this relation into~\cref{eq:4597}, we get that,
\[  \Omega_{\Sym} = \sqrt{(kt)!} \,\bigoplus_{\lambda} P^+_{\lambda}  \ot \sigma^A_{\lambda}(\id) \ot \sigma^B_{\lambda}(\id)\,  \Omega. \]
Note that there is only one sum over the Young diagrams, which means that the support of $\Omega_{\Sym}$ lies in the diagonal $\lambda_1=\lambda_2$ sector of the vector space~\cref{eq:defHil}. 
\par
It is straightforward to see that for every vector $\Omega$, one can always find a vector $\omega_\lambda \in V^{U(d),A}_\lambda \ot V^{U(d),B}_\lambda$, such that \[ P^+_\lambda \Omega = \phi^+_\lambda \ot \omega_\lambda.\]
(Again, using Dirac's notation $\ket \omega= \braket {\phi^+_\lambda| \Omega}$, where $\ket \Omega=\Omega$). Therefore, 
\begin{equation}\label{eq:29081}
    \Omega_{\Sym }= \bigoplus_{\lambda} \phi^+_\lambda  \ot  \omega_\lambda
\end{equation}
\par 
Recall that our goal was to derive an explicit expression for $\Omega_{\Sym}^\lambda$. We can read $\Omega_{\Sym}^\lambda$ from $\Omega_{\Sym}$ using the relation $\frac{f^\lambda}{(kt)!} \sum_{\pi \in S_{kt}} \chi_\lambda(\pi) r_A(\pi) = \id_{V_\lambda^{\Sym,A}} \ot \id_{V_\lambda^{U(d),A}}$, where $\chi_\lambda$ is the character of the symmetric group. Using this, we get that
\[  \Omega_{\Sym}^\lambda  =\frac{f^\lambda}{(kt)!} \sum_{\pi \in S_{kt}} \chi_\lambda(\pi) r_A(\pi) \,{\Omega}_{\Sym}  = \frac{f^\lambda}{((kt)!)^{3/2}} \sum_{\pi_1,\pi_2 \in S_{kt}} \chi_\lambda(\pi_1 \pi_2^{-1}) r_A(\pi_1) \ot r_B(\pi_2)\, \Omega.\]
Now, we can compute the quantity $\tr \left[ \left( \tr_{ V_\lambda^{U(d)}}(P_{\Omega^\lambda_{\Sym}}) \right)^2\right]$. Consider two copies of the state $ P_{\Omega^\lambda_{\Sym}}$: \[ (P_{\Omega^\lambda_{\Sym}})_1 \ot (P_{\Omega^\lambda_{\Sym}})_2 \in V^{\Sym,A_1}_{\lambda_1} \ot V^{U(d),A_1}_{\lambda_1} \ot V^{\Sym,B_1}_{\lambda_2} \ot V^{U(d),B_1}_{\lambda_2} \ot V^{\Sym,A_2}_{\lambda_1} \ot V^{U(d),A_2}_{\lambda_1} \ot V^{\Sym,B_2}_{\lambda_2} \ot V^{U(d),B_2}_{\lambda_2} . \]
Let $F_{V_\lambda^{U(d),A}}$ be the operator that swaps the vector spaces $V_\lambda^{U(d),A_1}$ and $V_\lambda^{U(d),A_2}$, and $F_{V_\lambda^{\Sym,A}}$ be the corresponding operator for the vector spaces $V^{\Sym,A_2}_{\lambda}$ and $V^{\Sym,A_1}_{\lambda}$. We have 
\[\tr \left[ \left( \tr_{ V_\lambda^{U(d)}}(P_{\Omega^\lambda_{\Sym}}) \right)^2\right] = \tr{\left[ (P_{\Omega^\lambda_{\Sym}})^{\ot 2} F_{V_\lambda^{U(d),A}}\right]} = \tr{\left[ (P_{\omega^\lambda_{\Sym}})^{\ot 2} F_{V_\lambda^{U(d),A}}\right]},\] 
where $P_{\omega^\lambda_{\Sym}}$ is the projector to ${\omega^\lambda_{\Sym}}$. Since $f^\lambda \tr{\left[ (P^+)^{\ot 2} F_{V_\lambda^{\Sym,A}}\right] }=1$, we conclude that,
\begin{multline*}
\tr \left[ \left( \tr_{ V_\lambda^{U(d)}}(P_{\Omega^\lambda_{\Sym}}) \right)^2\right] =
\tr{\left[ (P_{\omega^\lambda_{\Sym}})^{\ot 2} F_{V_\lambda^{U(d),A}}\right]}
=\\
f^\lambda \tr{\left[ (P_{\omega^\lambda_{\Sym}})^{\ot 2} F_{V_\lambda^{U(d),A}} \ot\left[ (P^+)^{\ot 2} F_{V_\lambda^{\Sym,A}}\right] \right]} =\\
f^\lambda \tr{\left[ (P_{\Omega^\lambda_{\Sym}})^{\ot 2} (F_{V_\lambda^{\Sym,A}}  \ot F_{V_\lambda^{U(d),A}} ) \right]}=  f^\lambda \tr{\left[ (P_{\Omega^\lambda_{\Sym}})^{\ot 2} F_A \right]},
\end{multline*}
with $F_A$ being the operator that simply swaps $\Hil_A$ tensor factors of the two copy vector space.
Substituting~\cref{eq:29081} into this expression leads to the following equation:
\begin{multline*}
    \tr \left[ \left( \tr_{ V_\lambda^{U(d)}}(P_{\Omega^\lambda_{\Sym}}) \right)^2\right]  = \\
    \frac{(f^\lambda)^5}{((kt)!)^6}\sum_{\alpha_1,\alpha_2,\beta_1,\beta_2,\gamma_1,\gamma_2,\delta_1,\delta_2 \in S_{kt}} \chi (\alpha_1 \alpha_2^{-1}) \chi (\beta_1 \beta_2^{-1}) \chi (\gamma_1 \gamma_2^{-1}) \chi (\delta_1 \delta_2^{-1}) \times \\ \tr{\left[r_{A_1}(\alpha_1)r_{B_1}(\alpha_2) r_{A_2}(\beta_1)r_{B_2}(\beta_2) (P_{\Omega_A} \ot P_{\Omega_B})^{\ot 2} r_{A_1}(\gamma_1)r_{B_1}(\gamma_2) r_{A_2}(\delta_1)r_{B_2}(\delta_2)    F_A     \right] }. 
\end{multline*}
The trace in the above expression reduces to
\begin{multline*}
\tr{\left[r_{A_1}(\alpha_1)r_{B_1}(\alpha_2) r_{A_2}(\beta_1)r_{B_2}(\beta_2) (P_{\Omega_A} \ot P_{\Omega_B})^{\ot 2} r_{A_1}(\gamma_1)r_{B_1}(\gamma_2) r_{A_2}(\delta_1)r_{B_2}(\delta_2)    F_A     \right] } = \\
\tr{\left[r_{A_1}(\alpha_1) r_{A_2}(\beta_1) (P_{\Omega_A} )^{\ot 2} r_{A_1}(\gamma_1) r_{A_2}(\delta_1)F_A \right] } \times 
\tr{\left[ r_{B}(\beta_2) P_{\Omega_B} r_{B}(\delta_2) \right]} \times \tr{\left[r_{B}(\alpha_2)P_{ \Omega_B}  r_{B}(\gamma_2) \right]} = \\
\tr{\left[ r_{A}(\beta_1) P_{\Omega_A}  r_{A}(\gamma_1)\right] }  \times \tr{\left[ r_{A}(\alpha_1) P_{\Omega_A}  r_{A}(\delta_1)\right] }\times 
\tr{\left[ r_{B}(\beta_2) P_{\Omega_B} r_{B}(\delta_2) \right]} \times \tr{\left[r_{B}(\alpha_2) P_{\Omega_B}  r_{B}(\gamma_2) \right]} = \\
\tr{\left[P_{\Omega_A}  r_{A}(\gamma_1\beta_1)\right] }  \times \tr{\left[ P_{\Omega_A}  r_{A}(\delta_1\alpha_1)\right] }\times 
\tr{\left[ P_{\Omega_B} r_{B}(\delta_2\beta_2) \right]} \times \tr{\left[ P_{\Omega_B}  r_{B}(\gamma_2\alpha_2) \right]}.
\end{multline*}
As define in~\cref{dfn:row-col-pres}, let $R$ be the row preserving group and $C$ be the column preserving group for the $k \times t$ rectangle above. It is evident that for any $\pi \in S_{kt}$, $\tr \left[P_{\Omega_A} r_A(\pi) \right] = \delta_{\pi \in R}$, and $\tr \left[P_{\Omega_B} r_B(\pi) \right] = \delta_{\pi \in C}$. Therefore, 
\begin{multline*}
    \tr \left[ \left( \tr_{ V_\lambda^{U(d)}}(P_{\Omega^\lambda_{\Sym}}) \right)^2\right]   = \\
    \frac{(f^\lambda)^5}{((kt)!)^6}\sum_{\alpha_1,\alpha_2,\beta_1,\beta_2,\gamma_1,\gamma_2,\delta_1,\delta_2 \in S_{kt}} \chi (\alpha_1 \alpha_2^{-1}) \chi (\beta_1 \beta_2^{-1}) \chi (\gamma_1 \gamma_2^{-1}) \chi (\delta_1 \delta_2^{-1})\,  \delta_{\gamma_1 \beta_1 \in R}\, \delta_{\delta_1 \alpha_1 \in R} \,\delta_{\delta_2 \beta_2 \in C}\, \delta_{\gamma_2 \alpha_2 \in C} \\ = 
    \frac{(f^\lambda)^5}{((kt)!)^6}\sum_{r_1,r_2 \in R, c_1,c_2 \in C, \alpha_1,\alpha_2,\beta_1,\beta_2 \in S_{tk}} \chi (\alpha_1 \alpha_2^{-1}) \chi (\alpha_2 c_1 r_1 \beta_1^{-1} ) \chi (\beta_1 \beta_2^{-1})  \chi (\beta_2  c_2 r_2 \alpha_1^{-1}). 
\end{multline*}
The following expression uses the Schur's lemma to "stitch" the different terms:
\begin{multline*}
 \sum_{\pi \in S_{tk}} \chi_\lambda ( \alpha \pi) \chi_\lambda (\pi^{-1} \beta)  = \sum_{\pi \in S_{tk}} \chi_\lambda ( \alpha \pi) \chi_\lambda (\beta^{-1} \pi) = \tr \left( \left[ \rho_\lambda ( \alpha) \ot  \rho_\lambda ( \beta^{-1}) \right]\left[ \sum_{\pi \in S_{tk}} \rho_\lambda(\pi)\ot\rho_\lambda(\pi) \right]\right) = \\
 ((kt)!) \tr \left( \left[ \rho_\lambda ( \alpha) \ot  \rho_\lambda ( \beta^{-1}) \right] P^+_{\lambda}  \right) = ((kt)!) \tr \left( \left[ \rho_\lambda ( \alpha)  \rho_\lambda ( \beta) \ot  \id \right] P^+_{\lambda}  \right)  = \frac{(kt)!}{f^\lambda} \chi_\lambda(\alpha \beta),
\end{multline*}
where $P^+_\lambda$ is again the maximally entangled vector between two copies of the vector space where $\rho_\lambda$ acts on, and we used the transpose trick for the one to the last equality. Using this relation multiple times, we conclude that,
\begin{equation*}
    \tr \left[ \left( \tr_{ V_\lambda^{U(d)}}(P_{\Omega^\lambda_{\Sym}}) \right)^2\right]    = 
    \frac{(f^\lambda)^2}{((kt)!)^2}\sum_{r_1,r_2 \in R, c_1,c_2 \in C} \chi_\lambda ( c_1 r_1  c_2 r_2 ) = \frac{(f^\lambda)^2}{((kt)!)^2} \tr_\lambda(RCRC). 
\end{equation*}
We used the notation introduced in~\cref{rem:notational-remark}.
\par
It remains to prove that if $l(\lambda) > k$ then $\tr_\lambda(RCRC) = 0$. To do so, it is sufficient to show that $\rho_\lambda(R)$ is zero if $l(\lambda)>k$. Consider the vector space $(\CC^d)^{\ot t k }$, and the commuting action of $U(d)$ and $S_{tk}$. We have that
 \begin{align}\label{eq:98687} \left(\sum_{r\in R} r\right) U^{\ot kt} = \bigoplus_{\lambda \vdash kt,  l(\lambda)\leq d} \rho_\lambda(R) \ot \sigma_\lambda(U).  
 \end{align}
 On the other hand, 
 \[ \frac{1}{|R|} \left(\sum_{r\in R} r\right) U^{\ot kt} = \left[\left(\frac{1}{t!}\sum_{\pi \in S_t} \pi \right)U^{\ot t}\right]^{\ot k} = \left[\Sym_t(U)\right]^{\ot k}. \]
By the Pieri's formula~\cite{fulton2013representation}, if $\lambda$ is any Young diagram, then $\rho_\lambda \ot \Sym_t = \bigoplus_\mu \rho_\mu$, where $\mu$ is constructed by adding $t$ boxes to $\lambda$ in a way that no two boxes are added to the same column. Therefore, the length of the irreps appearing in $\rho_\lambda \ot \Sym_t$ is at most one more than the length of $\lambda$. With this, we can easily conclude that no irrep of length more than $k$ appear in $\left[ \Sym_t (U)\right] ^{\ot k}$. Comparing with~\cref{eq:98687}, this means that $\rho_\lambda(R) =0 $ when $l(\lambda)>k$.

\section{Permanent of random i.i.d. Gaussian matrices}\label{sec:permiid}
In this section, we describe our results on lower bounding, computing, and estimating $\EE_{M \sim \Gauss k} |\Perm (M)|^{2t}$. Our starting point is the expansion formula in~\cref{rem:gauss_usub}, which we repeat here for the reader's convenience:
\begin{align}\label{eq:repeated-expansion-formula}
\EE_{M \sim \Gauss k} \left| \Perm(M) \right|^{2t} =  \frac{1}{(kt)!}\sum_{\lambda} f^\lambda\tr \left[\rho_\lambda(RCRC) \right].
\end{align}
Terms in the right-hand side sum are all positive, and therefore, any subset of them will constitute a valid lower bound for the moments of random permanents. The challenging task is to explicitly compute $\tr \left[\rho_\lambda(RCRC)\right]$. We can only compute this quantity for the Young diagrams of depth $2$ and $3$, where there are only a few boxes in the second and the third row. The first goal of this section is to show how this calculation is done.
\par 
In~\cref{sec:pleth}, we introduce the notion of plethysm and show how it can be used to bound $\tr \left[\rho_\lambda(RCRC)\right]$ in terms of $\tr \left[\rho_\lambda(RC)\right]$. Next, in~\cref{sec:citw,sec:comments}, we use Cauchy identity to provide formulas for $\tr \left[\rho_\lambda(RC)\right]$ when $l(\lambda)=2$, and a few explicit results for $l(\lambda)=3$. We extend these results to the case of $\tr \left[\rho_\lambda(RCRC)\right]$ in~\cref{sec:RCRC_calc}, and prove a lower bound for the moments of permanent in~\cref{sec:lbmgc}.
\par 
The second goal of this section is to argue that the lower bounds that we drive in~\cref{sec:lbmgc} are very close to being tight. To justify this assertion, we provide explicit calculation of the permanent moments using a new algorithm that we develop in~\cref{sec:alg}. Next, in~\cref{sec:conjsection}, we formally state our permanent moment growth conjecture and provide further analytical arguments to support it. Lastly, in~\cref{sec:concentration}, we assume our moment growth conjecture and use large deviation theory to predict the form of the tail of the log-permanent distribution. 

\subsection{Plethysm} \label{sec:pleth}
In this section, we first discuss the notion of plethysm and then proceed to connect it to our calculations. 
\par
Consider a complex vector space $\CC^d$ (with $d$ being large) where the $\GL(d)$ group acts on. One can construct the order $t$ symmetric representation of this group that acts on $\Sym_t (\CC^d)$. This is an irreducible representation as we mentioned before, so it cannot be decomposed into smaller irreps of $\GL(d)$.
\par
Now, proceed one step further and consider $k$-th symmetric power of $\Sym_t(\CC^d)$, i.e., $\Sym_k(\Sym_t(\CC^d))$. This space carries a representation of $\GL(d)$ and can be decomposed into its irreps $\sigma_\lambda$ (that act on $V_\lambda^{U(d)}$), each one repeated $\Pl^{k,t}_\lambda$ times:
\begin{equation}\label{eq:pleth}
    \Sym_k(\Sym_t(\CC^d)) =  \bigoplus_{\lambda\vdash kt} \CC^{{\Pl^{k,t}_\lambda}}  \otimes V_\lambda^{U(d)}.
\end{equation}
We call the coefficients $\Pl_\lambda^{k,t}$ the ``plethysm'' coefficients. Finding a combinatorial formula for the plethysm coefficients is a major open problem in representation theory, and except for the cases where $k<4$, no explicit relation is known for general $t$.
\par
Now, we state the main result of this subsection:
\begin{thm}\label{thm:pleth_ineq}
Let $\Pl_\lambda^{k,t}$ be the plethysm coefficient defined in~\cref{eq:pleth}, $\rho_\lambda$ be the irrep of $S_{kt}$ corresponding to the Young diagram $\lambda$, and $R$ and $C$ be the row and the column preserving subgroups defined in~\cref{dfn:row-col-pres}. Then, the following inequalities hold:
\begin{equation}\label{eq:pleth_ineq}
    (\tr[ \rho_\lambda(RC)])^2/\Pl_\lambda^{k,t}\leq \tr[ \rho_\lambda(RCRC)]\leq (\tr[ \rho_\lambda(RC)])^2.
\end{equation}
\end{thm}
\begin{proof}
Recall that the row preserving group $R$ (see~\cref{dfn:row-col-pres}) is constructed as the product of $k$ permutation groups that individually permute the content of rows of a $k\times t$ rectangle. Let us consider the \emph{diagonal} subgroup of $R$, $R_d$, where the same permutation acts on different rows. To be more explicit, each element of $R$ can be identified by a sequence of $k$ permutations in $S_t$: $(\pi_1,\cdots, \pi_k)$, where $\pi_i$ simply permutes the element of the $i$-th row. With this notation, each element of $R_d$ is of the form $(\pi,\pi,\cdots, \pi)$, with $\pi \in S_t$. Therefore, $|R_d| = t!$. 
\par
Moreover, it is easy to see that if $c\in C$ and $r_d \in R_d$, then $cr_d = r_dc$, therefore, two groups $C$ and $R_d$ commute and one can construct the product group $\hat C = C \times R_d$ of order $t! (k!)^t$. In a similar fashion, we can define the diagonal $C_d$ group and the corresponding product group $\hat R$ of order $k! (t!)^k$.
\par
Let us consider the action of the group $\hat R$ on the vector space $(\CC^d)^{\ot kt}$. The subgroup $R$ of $\hat R$ will project to $(\Sym_t(\CC^d))^{\ot k}$, and then the diagonal $C_d$ will further project to $\Sym_k(\Sym_t(\CC^d))$. Using the Schur-Weyl duality for permutations in $\hat R$, we get:
\[\Sym_k(\Sym_t(\CC^d)) = \text{support}(r(\hat R)) = \bigoplus \text{support}\left(\rho_\lambda(\hat R)\right) \ot V_\lambda^{U(d)}. \]
When comparing with the definition of the plethysm coefficients, this gives the following important result: \[\rank(\rho_\lambda(\hat R)) = \Pl_\lambda^{k ,t}.\]
Let us return to the quantity of interest, $\tr [\rho_\lambda(RCRC)]$. We have that,
\begin{equation}\label{eq:pleth-0}
\tr[ \rho_\lambda(RCRC)] = |R|^2 |C|^2 \tr\left[ \rho_\lambda\left(\frac{R}{|R|}\frac{C}{|C|}\frac{R}{|R|}\frac{C}{|C|}\right)\right] = |R|^2 |C|^2 \tr\left[ \left(\rho_\lambda\left(\frac{\hat R}{|\hat R|}\frac{C}{|C|}\frac{\hat R}{|\hat R|}\right)\right)^2\right], 
\end{equation} 
where we used the fact that $R/|R|$ and $C/|C|$ are projectors, and $C = C_dCC_d/(k!)^2$. Note that,
\begin{equation}\label{eq:pleth-1}
    \rank \rho_\lambda\left(\frac{\hat R}{|\hat R|}\frac{C}{|C|}\frac{\hat R}{|\hat R|}\right)\leq \rank \rho_\lambda(\hat R) = \Pl_\lambda^{k,t}.
\end{equation} 
In general, for any positive Hermitian matrix $A$ we have $\tr (A)^2/\rank(A)\leq\tr(A^2) \leq \tr(A)^2$. Setting $A = \rho_\lambda\left(\frac{\hat R}{|\hat R|}\frac{C}{|C|}\frac{\hat R}{|\hat R|}\right)$ we obtain:
\begin{equation}\label{eq:pleth-2}
    \frac{\left(\tr \left[\rho_\lambda\left(\frac{\hat R}{|\hat R|}\frac{C}{|C|}\frac{\hat R}{|\hat R|}\right)\right]\right)^2}{\Pl_\lambda^{k,t}}\leq\tr\left[ \left(\rho_\lambda\left(\frac{\hat R}{|\hat R|}\frac{C}{|C|}\frac{\hat R}{|\hat R|}\right)\right)^2\right] \leq \left(\tr \left[\rho_\lambda\left(\frac{\hat R}{|\hat R|}\frac{C}{|C|}\frac{\hat R}{|\hat R|}\right)\right]\right)^2.
\end{equation} 
But, $\tr \left[\rho_\lambda\left(\frac{\hat R}{|\hat R|}\frac{C}{|C|}\frac{\hat R}{|\hat R|}\right)\right] = \tr \left[\rho_\lambda\left(\frac{ R}{| R|}\frac{C}{|C|}\frac{ R}{|R|}\right)\right] = \frac{1}{|R||C|}\tr \left[\rho_\lambda\left(RC\right)\right]$. Combining~\cref{eq:pleth-0} and~\cref{eq:pleth-2}, we get the desired result.
\end{proof}
The significance~\cref{thm:pleth_ineq} is in the fact that, generically, the plethysm coefficients are much smaller than $\rank \rho_\lambda(R)$, and $\tr \rho_\lambda(RC)$ is easier to compute than $\tr \rho_\lambda(RCRC)$. For several important cases $\Pl_\lambda^{k,t} =1$, which fully reduces the much harder calculation of $\tr \rho_\lambda(RCRC)$ to the calculation of $\tr \rho_\lambda(RC)$.
\par
\subsubsection{Special cases of plethysm}\label{sec:pleth-ex}
Before ending this subsection, we wish to compute $\Pl_\lambda^{k,t}$ for $l(\lambda)\leq 2$ as well as a few cases where $l(\lambda)=3$.
\par
As our first step, note that the character of $\Sym_k(\Sym_t(\CC^d))$ can be easily computed. Let $\lambda_1,\cdots, \lambda_d$ be the eigenvalues of an element of $\GL(d)$. Then, the eigenvalues of $\Sym_t(\CC^d)$ are all monomials $u_i := \lambda_1^{t_1^i} \cdots \lambda_d^{t_d^i}$, where $\{t^i\}$ is the set of vectors of integers that are solutions to $t_1^i + \cdots +t_d^i =  t$. There are $q_t  =\binom{d+t-1}{d-1}$ solutions to this equation, so $i$ runs from $1$ to $q_t$. Then, the character of this representation is simply 
\begin{equation}\label{eq:pleth-character}
f(\lambda_1,\cdots, \lambda_d) = h_k(u_1,\cdots,u_{q_t})=\sum_{i_1\leq \cdots \leq i_k} u_{i_1} \cdots u_{i_k}.
\end{equation}
\par
We wish to expand $f$ in terms of Schur polynomials, and it can be concluded from~\cref{eq:pleth} that the coefficient of each polynomial is equal to the plethysm coefficient:
\[f(\lambda_1,\cdots,\lambda_d) = \sum_\mu \Pl_\mu^{k,t} \,s_\mu(\lambda_1,\cdots,\lambda_d). \]
\par 
Let us count how many times the term $\lambda_1^{kt}$ appears in $f$. To produce such a term in the expansion of $h_k(u)$, we are only allowed to include one term, i.e., $u_1^k$, where $u_1=\lambda_1^t$. Therefore, this coefficient appears once. We can generalize this argument to count the coefficient of the term $\lambda_1^{kt-a} \lambda_2^{a}$ in $f$. To make such coefficients, we are only allowed to use the terms $u_{\omega_i}:= \lambda_1^{t-i}\lambda_2^{i}$. Assume that in the expansion of $h_k$ these coefficients appear with the powers of $k_i$, i.e., we are considering terms such as $u_{\omega_1}^{k_1}\cdots u_{\omega_t}^{k_t}$. These coefficients satisfy $\sum_{i=0}^t k_i = k$, and $\sum_{i=0}^t i k_i = a$. This is equal to the number of partitions of $a$ to at most $k$ parts, where the size of each part is limited by $t$. In simpler words, this is the number of Young diagrams of the $a$ boxes that fit into a $k\times t$ rectangle. Let us call the number of Young diagrams of $a$ box that fit in a $t\times k$ rectangle $q(t,k,a)$. Therefore, we have
\[ \text{coefficient of }\lambda_1^{kt-a}\lambda_2^a \text{ in }f = q(t,k,a):=\text{\# Young diagrams of }a\text{ boxes that fit in a }k\times t \text{ rectangle}.\]
\par
On the other hand, we can see from the definition of the Schur polynomials (\cref{eq:schur-def-2}) that the term $\lambda_1^{kt-a}\lambda_2^a$ can only appear in $s_\mu$ where the diagram $\mu$ can be filled with the content $(kt-a,a)$. This means that $\mu$ should be of the form $(kt-b,b)$, with $b\leq a$. More precisely, the coefficient of $\lambda_1^{kt-a}\lambda_2^a$ in $s_{(kt-b,b)}$ is given by $K_{(kt-b,b),(kt-a,a)} = 1$ as long as $b\leq a$. Hence, \[Pl_{(kt-a,a)}^{k,t} = \text{ Coefficient of }\lambda_1^{kt-a}\lambda_2^a \text{ in }f - \text{ Coefficient of }\lambda_1^{kt-a+1}\lambda_2^{a-1}\text{ in }f =  q(t,k,a)- q(t,k,a-1).\]
In particular, assuming that $k\geq t \geq 2$, the first few plethysm coefficients are:
\begin{itemize}
    \item $\lambda = (kt)$: $\Pl_\lambda^{k,t} = 1$.
    \item $\lambda = (kt-1,1)$: $\Pl_\lambda^{k,t} = q(t,k,1)-q(t,k,0) =1-1= 0$.
    \item $\lambda = (kt-2,2)$: $\Pl_\lambda^{k,t} = q(t,k,2)-q(t,k,1) =2-1= 1$.
    \item $\lambda = (kt-3,3)$: $\Pl_\lambda^{k,t}=q(t,k,3)-q(t,k,2)=q(t,k,3)-2$. If $t=2$, then $q(2,k)=2$, and if $t\geq 3$, then $q(2\geq 3,k)=3$. Therefore, 
    \begin{align*}
        &\Pl^{t,k}_{(kt-3,3)} = 1 \text{ if } t\geq 3,\\
        &\Pl^{t,k}_{(kt-3,3)} = 0 \text{ if } t=2.
    \end{align*}
    \item $\lambda = (kt-4,4)$. Using the same logic as the previous case and counting the number of Young diagrams that fit in a $t\times k$ box, we get the following values for the plethysm coefficients:
    \begin{align*}
        &\Pl^{t,k}_{(kt-4,4)} = 2 \text{ if } t\geq 4,\\
        &\Pl^{t,k}_{(kt-4,4)} = 1 \text{ if } t= 3, \\
        &\Pl^{t,k}_{(kt-4,4)} = 1 \text{ if } t= 2.
    \end{align*}

\end{itemize}

\par
Similar logic can be used to compute a few plethysm coefficients with $l(\lambda)>2$:
\begin{itemize}
    \item $\lambda=(kt-2,1,1)$. In this case, we need to compute the coefficient of $\lambda_1^{kt-2}\lambda_2 \lambda_3$ in $f$. Let us define $u_1 = \lambda_1^t$, $u_2 = \lambda_1^{t-1}\lambda_2$, $u_3 = \lambda_1^{t-1}\lambda_3$, and $u_4 = \lambda_1^{t-2} \lambda_2 \lambda_3$. Using the definition of the character~\cref{eq:pleth-character}, there are only two ways to produce $\lambda_1^{kt-2}\lambda_2 \lambda_3$ in $f$: either as $u_1^{k-1}u_4$, or $u_1^{k-2}u_2u_3$. On the other hand, there are only $4$ Schur polynomials that contain this term: $s_{(kt)},s_{(kt-1,1)},s_{(kt-2,2)},s_{(kt-2,1,1)}$.
    Hence,
    \begin{multline*}
    f = \Pl_{(kt)} s_{(kt)}+\Pl_{(kt-1,1)} s_{(kt-1,1)}+\Pl_{(kt-2,2)} s_{(kt-2,2)}+\Pl_{(kt-2,1,1)} s_{(kt-2,1,1)}\\
    +\text{ terms that do not contain }\lambda_1^{kt-2}\lambda_2\lambda_3=\\
    s_{(kt)}+s_{(kt-2,2)}+\Pl_{(kt-2,1,1)} s_{(kt-2,1,1)}+\text{ terms that do not contain }\lambda_1^{kt-2}\lambda_2\lambda_3.
    \end{multline*}
    But the number of times that $\lambda_1^{kt-2}\lambda_2 \lambda_3$ appears in $s_{(kt)}+s_{(kt-2,2)}$ is exactly $2$ as well. Therefore, $\Pl_{(kt-2,1,1)}=0$.
    \item $\lambda = (kt-3,2,1)$, with $t\geq 3$. Again, we need to compute the coefficient of $\lambda_1^{kt-3}\lambda_2^{2}\lambda _3$ in $f$. There are $4$ ways to produce this term in $f$ according to~\cref{eq:pleth_ineq}: $(\lambda_1^t)^{k-1}\times (\lambda_1^{t-3}\lambda_2^2\lambda_3)$, $(\lambda_1^t)^{k-2}\times (\lambda_1^{t-2}\lambda_2^2)\times (\lambda_1^{t-1}\lambda_3)$,
    $(\lambda_1^t)^{k-2}\times (\lambda_1^{t-1}\lambda_2)\times (\lambda_1^{t-2}\lambda_2\lambda_3)$, and $(\lambda_1^t)^{k-3}\times (\lambda_1^{t-1}\lambda_2)^2\times (\lambda_1^{t-1}\lambda_3)$. Similarly, we have that
    \begin{align*}
        f = s_{(kt)}+s_{(kt-2,2)}+s_{(kt-3,3)}+\Pl^{k,t}_{(kt-3,2,1)}s_{(kt-3,2,1)}+ \cdots+\text{ terms that do not contain }\lambda_1^{kt-2}\lambda_2^2\lambda_3.
    \end{align*}
    One can easily see that $\lambda_1^{kt-2}\lambda_2^2\lambda_3$ appears once in $s_{(kt)}$ and $s_{(kt-3,3)}$ and twice in $s_{(kt-2,2)}$. This follows from~\cref{eq:transition_m_to_s} and $K_{(kt),(kt-3,2,1)} =K_{(kt-3,3),(kt-3,2,1)}=1$, and $K_{(kt-2,2),(kt-3,2,1)}=2$. Therefore, all of the 4 copies of $\lambda_1^{kt-2}\lambda_2^2\lambda_3$ that exist in $f$ are already captured by the two row Young diagrams, and we have $\Pl_{(kt-3,2,1)}^{k,t}=0$.
    \item $\lambda = (kt-4,3,1)$, and $t\geq 4$. The similar logic as the previous parts will show that $\Pl_{(kt-4,3,1)}^{k,t}=0$
    \item $\lambda = (kt-4,2,2)$, and $t\geq 3$. In this case we have that $\Pl_{(kt-4,2,2)}^{k,t}=1$
    \item $\lambda = (kt-5,3,2)$, and $t\geq 3$. Similar logic shows that $\Pl_{(kt-5,3,2)}^{k,t}=1$
\end{itemize}
And lastly, for the case of $t=2$:
\begin{itemize}
    \item When $t=2$, then it is easy to see that $Q(2,k,a)=1+\lfloor a/2 \rfloor$. Hence, 
    \begin{equation}\label{eq:pleth-k-2}
    \Pl_{(2t-a,a)}^{2,k} = 1 \text{ if }a\text{ is even and }0\text{ otherwise.}
    \end{equation}
\end{itemize}

\subsection{The Cauchy identity and \texorpdfstring{$\tr \left[\rho_\lambda (RC) \right]$}{} for \texorpdfstring{$l(\lambda)=2$}{}}\label{sec:RC_calc}\label{sec:citw}
In this section, we are mainly concerned with computing $\tr \left[\rho_\lambda (RC) \right]$ when $l(\lambda)\leq 2$. 
\par
Again, start with the Schur-Weyl duality and consider the action of permutations on $(\CC^d)^{\ot kt}$ for large $d$. Let $X=\text{diag}(x_1,x_2,\cdots, x_d)$ be a fixed $d\times d$ diagonal matrix. We have that
\begin{equation}\label{eq:2rowres}
\tr[RC X^{\ot kt}]=\sum_{\lambda\vdash kt, l(\lambda)\leq \min(k,t)} \tr[\rho_\lambda(RC)]s_\lambda(x_1,x_2,\cdots, x_d). \end{equation}
Now, set $x_3=x_4=\cdots=x_d=0$. This will introduce a method for restricting the above sum to the Young diagrams with $l(\lambda)\leq 2$, as $s_\lambda(x_1,x_2,0,0,\cdots, 0)$ automatically vanishes if $l(\lambda)\geq 2$.
\par 
By directly computing the left-hand side of~\cref{eq:2rowres} using the Cauchy identity~\cite{bump2004lie}, we obtain the following result:
\begin{thm}[\texorpdfstring{$\tr \left[\rho_\lambda (RC) \right]$}{} when \texorpdfstring{$l(\lambda)\leq 2$}{}]\label{thm:RC-formula}
Let $k$ and $t$ be positive integers and $a\geq 0$. Then,
\begin{equation} \label{eq:RC-formula-from-Q}
\tr[\rho_{(kt-a,a)}(RC) ] = \left \{ \begin{array}{ccc} Q(k,t,a) -Q(k,t,a-1) &&\text{if}\quad a>0\\Q(k,t,0)&&\text{if}\quad a=0 \end{array}\right. ,
\end{equation}
where $Q(k,t, a)$ is defined as: 
\[ Q(k,t, a) :=k!t! \sum_{|\mu|=|\nu|=a ,l(\mu)\leq k, l(\nu)\leq t}   (\text{IB}_{\mu\nu} )^2    \frac{\prod_{i=1}^k \mu_i! (t-\mu_i)!}{\#_\mu(0)!\cdots \#_\mu(t)! }\frac{\prod_{j=1}^t \nu_i! (k-\nu_i)!}{\#_\nu(0)!\cdots \#_\nu(k)!}, \]
where $\#_\lambda(i)$ (for $i\neq 0$) is the number of times that the number $i$ appears in the partition $\lambda$, $\#_\mu(0):=k-\sum_{i=1}^t \#_\mu(i)$, $\#_\nu(0):=t-\sum_{i=1}^k \#_\nu(i)$, and $\text{IB}_{\mu\nu}$ is the number of $0$-$1$ matrices with the row sum $\mu$ and the column sum $\nu$ (see~\cref{eq:IB_formula}).  
\end{thm}
\begin{proof}
First, we construct $R$ and $C$ projectors. Let us start with $R$ and consider a set of $k$ vectors of dimension $d$, $\{ \alpha_1, \alpha_2,\cdots, \alpha_k\}$. We choose the components of all of these vectors to be random i.i.d. complex Gaussian numbers with variance $1$ and mean $0$. Then, we place the vectors in a $k\times t$ matrix with the following pattern:
\begin{equation}\label{eq:boxk}
\EE_{\{\alpha_i\}} \arraycolsep=1.2pt\def\arraystretch{1.5} \left(\begin{array}{cccccccccc}
\ {\alpha_1}{\alpha_1}^\dagger&\ot& \ {\alpha_1}{\alpha_1}^\dagger&\ot& \ {\alpha_1}{\alpha_1}^\dagger&\ot& \cdots &\ot& \ {\alpha_1}{\alpha_1}^\dagger&\ot\\
\ {\alpha_2}{\alpha_2}^\dagger&\ot& \ {\alpha_2}{\alpha_2}^\dagger&\ot& \ {\alpha_2}{\alpha_2}^\dagger&\ot& \cdots &\ot& \ {\alpha_2}{\alpha_2}^\dagger&\ot\\
\ {\alpha_3}{\alpha_3}^\dagger&\ot& \ {\alpha_3}{\alpha_3}^\dagger&\ot& \ {\alpha_3}{\alpha_3}^\dagger&\ot& \cdots &\ot& \ {\alpha_3}{\alpha_3}^\dagger&\ot\\
\vdots && \vdots&& \vdots && \ddots && \vdots&\\
\ {\alpha_k}{\alpha_k}^\dagger&\ot& \ {\alpha_k}{\alpha_k}^\dagger&\ot& \ {\alpha_k}{\alpha_k}^\dagger&\ot& \cdots &\ot& \ {\alpha_k}{\alpha_k}^\dagger&
\end{array}\right).
\end{equation}

Note that $\alpha_i \alpha_i^\dagger$ is a rank one matrix, and should not be confused with the inner product $\alpha_i^\dagger \alpha_i$ (in the Dirac's notation, if $\ket{\alpha_i} = \alpha_i$, then $\alpha_i \alpha_i^\dagger = \ketbra{\alpha_i}{\alpha_i}$). The average $\EE_{\alpha_i}$ is taken over all $d k$ complex numbers that constitute the coefficients of $\alpha_i$'s. 
\par 
\Cref{eq:boxk} can be more compactly written as $\left(\EE_\alpha \,({\alpha}{\alpha^\dagger})^{\ot t}\right)^{\ot k}$, where the tensor power inside the parentheses forms the rows of~\cref{eq:boxk}. Similarly, we can define $\left(\EE_\beta \,({\beta}{\beta^\dagger})^{\ot k}\right)^{\ot t}$ for the set of $t$ vectors $\{\beta_1,\cdots,\beta_t\}$, where the inside tensor power indicates the columns of a $k\times t$ matrix.. 
\par
It is easy to check that $\EE_\alpha \,({\alpha}{\alpha^\dagger})^{\ot t} \propto \sum_{\pi\in S_t} \pi$, as $[\EE_\alpha ({\alpha}{\alpha^\dagger})^{\ot t} ,U^{\ot t}]=0$ for all $U\in U(d)$. The proportionality constant is computed by comparing the traces: $\tr[\EE_\alpha \,({\alpha}{\alpha^\dagger})^{\ot t}] = \EE_\alpha \,({\alpha^\dagger \alpha})^t=t!\binom{d+t-1}{t}$, and $\tr[\sum_{\pi\in S_t} \pi] = t! \times \binom{d+t-1}{t}$. Hence, we conclude that, $\EE_\alpha\,( {\alpha}{\alpha})^{\ot t}=\sum_{\pi \in S_t} \pi$. To summarize, 
\[R =  \left(\EE_\alpha \,({\alpha}{\alpha^\dagger})^{\ot t}\right)^{\ot k}, \text{ and }C =  \left(\EE_\beta \,({\beta}{\beta^\dagger})^{\ot k}\right)^{\ot t}.\]
Substituting in~\cref{eq:2rowres},
\begin{multline*}
    \tr [RCX^{\ot kt}] = \tr \left[R\sqrt X^{\ot kt}C \sqrt X^{\ot kt}\right]= \tr\left[\left(\EE_\alpha \,({\alpha}{\alpha^\dagger})^{\ot t}\right)^{\ot k} \sqrt X^{\ot kt}   \left(\EE_\beta \,({\beta}{\beta^\dagger})^{\ot k}\right)^{\ot t}   \sqrt X^{\ot kt}  \right] = \\
    \EE_{\{\alpha_i\},\{\beta_i \}} \prod_{i=1,j=1}^{k,t} \left|{\alpha_i^\dagger}\,\sqrt{X} \,{\beta_j}\right|^2 =\EE_{\{\alpha_i\},\{\beta_i \}} \left| \prod_{i=1,j=1}^{k,t}\left(\alpha_i^1 \beta_j^1 \sqrt x_1+\alpha_i^2 \beta_j^2 \sqrt x_2+\cdots +\alpha_i^d \beta_j^d \sqrt x_d \right) \right|^2.
\end{multline*} 
\par
Setting $x_3=x_4=\cdots=x_d=0$ we have,
\begin{multline}\label{eq:208908}
\sum_{\lambda\vdash kt, l(\lambda)\leq 2} \tr[\rho_\lambda(RC)]s_\lambda(x_1,x_2) = \EE_{\{\alpha_i^1,\alpha_j^2,\beta_i^1,\beta_i^2\} } \left| \prod_{i=1,j=1}^{k,t}(\alpha_i^1 \beta_j^1 \sqrt x_1+\alpha_i^2 \beta_j^2 \sqrt x_2 ) \right|^2  = \\
(x_1)^{kt}\,\EE \prod_{i=1,j=1}^{k,t} \left|\alpha_i^1 \beta_j^1\right|^2 
 \left| \prod_{i=1,j=1}^{k,t}(1+\frac{\alpha_i^2 \beta_j^2 \sqrt x_2}{\alpha_i^1 \beta_j^1 \sqrt x_1} ) \right|^2.
\end{multline}
One of the Cauchy identities (see~\cite{bump2004lie,fulton2013representation}) is the following polynomial equality,
\begin{equation}\label{eq:cauchy}  \prod_{i=1,j=1}^{k,t}(1+v_i w_j ) = \sum_{\lambda\vdash kt} s_\lambda(v_1,v_2,\cdots, v_k) s_{\tilde\lambda}(w_1,w_2,\cdots, w_t).
\end{equation}
Using~\cref{eq:transition_m_to_s}, we have that $\sum_\lambda s_\lambda s_{\tilde \lambda} =\sum_{\lambda \mu\nu} K_{\lambda\mu} K_{\tilde \lambda \nu } m_\mu m_\nu =\sum_{\mu\nu} \text{IB}_{\mu\nu} m_\mu m_\nu $.
Setting $v_i = \alpha_i^2/\alpha_i^1$ and $w_j = (\beta_j^2/\beta_j^1)\sqrt{x_2/x_1}$, this leads to,
\begin{multline}  \left| \prod_{i=1,j=1}^{k,t}(1+v_i w_j )\right|^2 = \left|\sum_{\lambda\vdash kt} s_\lambda(v_1,v_2,\cdots, v_k) s_{\tilde\lambda}(w_1,w_2,\cdots, w_t)\right|^2=\\
 \left|\sum_{l(\mu)\leq k, l(\nu)\leq t} \text{IB}_{\mu\nu} m_{\mu}(v_1,\cdots,v_k)m_{\nu}(w_1,\cdots,w_t) \right|^2=\\
 \left|\sum_{l(\mu)\leq k, l(\nu)\leq t} \text{IB}_{\mu\nu} \sqrt{\frac{x_2}{x_1}}^{|\nu|} m_{\mu}\left(\frac{\alpha_1^2}{\alpha_1^1},\cdots,\frac{\alpha_k^2}{\alpha_k^1}\right)m_{\nu}\left(\frac{\beta_1^2}{\beta_1^1},\cdots,\frac{\beta_t^2}{\beta_t^1}\right) \right|^2.
\end{multline}
Combining this relation with~\cref{eq:208908}, we derive the following equation:
\begin{multline}
\sum_{\lambda\vdash kt, l(\lambda)\leq 2} \tr[\rho_\lambda(RC)]s_\lambda(x_1,x_2) =\\ \EE \left| \sum_{l(\mu)\leq k, l(\nu)\leq t} x_1^{\frac{kt-|\mu|}{2}}x_2^{\frac{|\nu|}{2}} \text{IB}_{\mu\nu} \left[(\alpha_1^1\cdots \alpha_k^1)^tm_{\mu}\left(\frac{\alpha_1^2}{\alpha_1^1},\cdots,\frac{\alpha_k^2}{\alpha_k^1}\right)  \right] \left[ (\beta_1^1\cdots \beta_t^1)^k  m_{\nu}\left(\frac{\beta_1^2}{\beta_1^1},\cdots,\frac{\beta_t^2}{\beta_t^1}\right)   \right] \right|^2.
\end{multline}
It is straightforward to see that,
\begin{multline}
\EE_{\{\alpha_i^1,\alpha_i^2\}}\,\left[(\alpha_1^1\cdots \alpha_k^1)^tm_{\mu}\left(\frac{\alpha_1^2}{\alpha_1^1},\cdots,\frac{\alpha_k^2}{\alpha_k^1}\right)  \right] \overline{\left[(\alpha_1^1\cdots \alpha_k^1)^t m_{\nu}\left(\frac{\alpha_1^2}{\alpha_1^1},\cdots,\frac{\alpha_k^2}{\alpha_k^1}\right)  \right]} =\\ 
\delta_{\mu\nu} \prod_{i=1}^k \mu_i! (t-\mu_i)! m_{\mu}(1,\cdots,1),
\end{multline}
where $m_\mu(1,\cdots,1)$ is the number of terms in the expansion of $m_\mu$. One can see that $m_\mu(1,\cdots,1) = k!/[\#_\mu(0)!\#_\mu(1)!\cdots \#_\mu(k)!]$. 
Using this definition, we conclude that:
\begin{multline}
\sum_{\lambda\vdash kt, l(\lambda)\leq 2} \tr[\rho_\lambda(RC)]s_\lambda(x_1,x_2) =\\
k!t! \sum_{l(\mu)\leq k, l(\nu)\leq t}  x_1^{{kt-|\mu|}}x_2^{{|\nu|}} (\text{IB}_{\mu\nu} )^2    \frac{\prod_{i=1}^k \mu_i! (t-\mu_i)!}{\#_\mu(0)!\cdots \#_\mu(k)! }\frac{\prod_{j=1}^t \nu_i! (k-\nu_i)!}{\#_\nu(0)!\cdots \#_\nu(t)!}.
\end{multline}
Let us define $Q(k,t, a)$ as
\[ Q(k,t, a) :=k!t! \sum_{|\mu|=|\nu|=a ,l(\mu)\leq k, l(\nu)\leq t}   (\text{IB}_{\mu\nu} )^2    \frac{\prod_{i=1}^k \mu_i! (t-\mu_i)!}{\#_\mu(0)!\cdots \#_\mu(k)! }\frac{\prod_{j=1}^t \nu_i! (k-\nu_i)!}{\#_\nu(0)!\cdots \#_\nu(t)!}. \]
Using the expansion $s_{(kt-a,a)}(x_1,x_2)$ into monomials we have
\[ \tr[\rho_{(kt-a,a)}(RC) ] = Q(k,t,a) -Q(k,t,a-1) \text{ if }a>0 \quad\text{and}\quad Q(k,t,0)\text{ when }a=0, \]
which is the desired result.
\end{proof}
When $a\geq 3$, it hard to compute $\tr \left[\rho_\lambda (RC) \right]$ using the result of~\cref{thm:RC-formula} by hand. We have computerized this calculation and report the results in the rest of this section:
\begin{itemize}
\item{\texorpdfstring{$t=2$}{} :}
This is the simplest case and explicit formulas for $\tr [\rho_\lambda(RC)]$ can be derived.
\par
In order to use~\cref{thm:RC-formula}, we have to compute $Q(k,2,a)$:
\[Q(k,2,a) =2 \times k! \sum_{|\mu|=|\nu|=a ,l(\mu)\leq k, l(\nu)\leq 2}   (\text{IB}_{\mu\nu} )^2    \frac{\prod_{i=1}^k \mu_i! (2-\mu_i)!}{\#_\mu(0)!\cdots \#_\mu(k)! }\frac{\prod_{j=1}^2 \nu_i! (k-\nu_i)!}{\#_\nu(0)!\cdots \#_\nu(2)!}.   \]
The terms in the sum are identified by two numbers $r$ and $s$: $\mu = (2^r,1^{l-2r})$, and $\nu = (a-s,s)$, with $a-k\leq r$, and $i^r$ means that the number $i$ is repeated $r$ times. After some investigation, one can easily see that the number of $k\times 2$ matrices of $0$ and $1$, with the row sum given by $\mu$ and column sum given by $\nu$ is equal to $\binom{a-2r}{s-r}$. Therefore, 
\begin{align}\label{eq:expandedbinom}
    Q(k,2,a) =2 \times k! \sum_{0\leq r\leq s \leq a/2}   \binom{a-2r}{s-r} ^2    \frac{2^{k-a+2r}}{r!(a-2r)!(k-a+r)! }\frac{s!(k-s)!(a-s)!(k-a+s)!}{1+\delta_{a,2s}}.
\end{align}
Then, it follows from~\cref{thm:RC-formula} that 
\begin{align}\label{eq:rawt2} 
\tr[\rho_{(2k-a,a)}(RC)]= Q(k,2,a)-Q(k,2,a-1)\text{ when } a>0\text{, and }1 \text{ when }a=0. 
\end{align}
This is a complicated relation, but surprisingly, it can be simplified to obtain a simple expression.
\begin{rem}\label{rem:t-2} 
We guess a simple form for $\tr[\rho_{(2k-a,a)}(RC)]$:
\begin{align}\label{eq:nonrawt2}
\tr[\rho_{(2k-a,a)}(RC)] =2^k (k!)^2 \times \delta_{a=0\text{ mod } 2}\times  2^{-a} \frac{\binom{a}{a/2}}{\binom{k}{a/2}}. 
\end{align}
Equality of~\cref{eq:rawt2} and~\cref{eq:nonrawt2} has been tested numerically for $a\leq k\leq 1000$. We are confident that~\cref{eq:nonrawt2} is valid, however, we have not been able to rigorously derive it from~\cref{eq:rawt2} as a binomial identity.
\end{rem}
\item{\texorpdfstring{$t=3$}{} :}
Now, we report the result of $\tr [ \rho_{\lambda=(kt-a,a)}(RC)]$ for $t=3$, and $2\leq k\leq 19$, and $k\geq a$. Due to $k \leftrightarrow t$ symmetry, the case of $k=2$ is already covered in the previous section. The results have been calculated using~\cref{thm:RC-formula}, but the calculations are tedious and have been computerized. Moreover, we used~\cref{eq:IB_formula} extensively to compute $\text{IB}_{\mu\nu}$.
\par 
One observes that for $2\leq k\leq 19$, 
\begin{equation}\label{eq:rct3}
    \tr[\rho_{(2k-a,a)}(RC)] = (3!)^k (k!)^3 \times \frac{1}{Q^3_{a,k}}\times p^3_{\text{RC},a}(k),
\end{equation}
where $Q^3_{a}(k)$ is defined as 
\begin{align}\label{eq:q3lkdef}
    Q^3_{a,k} =
      \binom{k}{\lfloor a/2 \rfloor}\lfloor a/2 \rfloor ! 3^{\lfloor2a/3\rfloor} \times \left\{\begin{array}{cc}
    \binom{k}{2 \lfloor \frac{a}{6}\rfloor}\frac{\left[2 \lfloor \frac{a}{6}\rfloor\right]!}{\frac{a}{2}!} 3^{-\lfloor a/6\rfloor} & \quad \text{for even }k\\
     & \\
        \binom{k}{2 \lfloor \frac{a+4}{6}\rfloor-1}\frac{\left[2 \lfloor \frac{a+4}{6}\rfloor-1\right]!}{\frac{a+3}{2}!} 3^{2-\lfloor a/6+2/3\rfloor}/2 & \quad \text{for odd }k
    \end{array}\right. ,
\end{align}
and $p^3_{\text{RC},a}(k)$ are the monic polynomials in~\cref{tab:t3}.
\begin{table}
\begin{center}
 \begin{tabular}{||c| c |} 
 \hline
 $a$& $p_{\text{RC,}a}^3(k)$\\
 \hline
$ 0 $ & $ 1 $  \\
$ 1 $ & $ 0 $  \\
$ 2 $ & $ 1 $  \\
$ 3 $ & $ 1 $  \\
$ 4 $ & $ 1 $  \\
$ 5 $ & $ 1 $  \\
$ 6 $ & $  k^2 + 35/9  k - 92/9 $  \\
$ 7 $ & $ 1 $  \\
$ 8 $ & $  k^2 + 77/9  k - 274/9 $  \\
$ 9 $ & $  k^2 + 4/9  k - 9 $  \\
$ 10 $ & $  k^2 + 143/9  k - 72 $  \\
$ 11 $ & $  k^2 + 25/9  k - 214/9 $  \\
$ 12 $ & $  k^4 + 194/9  k^3 - 15421/81  k^2 + 3310/27  k + 72640/81 $  \\
$ 13 $ & $  k^2 + 58/9  k - 467/9 $  \\
$ 14 $ & $  k^4 + 326/9  k^3 - 23149/81  k^2 - 45346/81  k + 127480/27 $  \\
$ 15 $ & $  k^4 + 130/27  k^3 - 30775/243  k^2 + 83170/243  k + 17384/81 $  \\
$ 16 $ & $  k^4 + 500/9  k^3 - 28777/81  k^2 - 244160/81  k + 471884/27 $  \\
$ 17 $ & $  k^4 + 331/27  k^3 - 52615/243  k^2 + 2743/9  k + 482308/243 $  \\
$ 18 $ & $  k^6 + 641/9  k^5 - 248621/243  k^4 - 5416259/2187  k^3 + 56013238/729  k^2 - 592053832/2187  k  $  \\
&+ 87973760/729 \\
$ 19 $ & $  k^4 + 598/27  k^3 - 78295/243  k^2 - 79210/243  k + 649976/81 $\\  
\hline
 \end{tabular}
 \end{center}
 \caption{The polynomials $p_{\text{RC},a}^3$ used in~\cref{eq:rct3}}\label{tab:t3}
\end{table}

\item Arbitrary \texorpdfstring{$t$}{} : 
Here we report the results of the calculation of $\tr \rho_{(kt-a,a)}(RC)$ for small values of $a$. We assume that $k,t\geq a$. Sadly, the polynomials get too complicated very quickly, and do not have enough space to report even the first $10$ of them. However, we can identify a pattern for the highest order terms:
\begin{equation}\label{eq:rcgent}
    \tr \rho_{(kt-a,a)}(RC) = (k!)^t (t!)^k \times 2^{1-(-1)^a}  \times \prod_{i=0}^{a-1} \left[(k-i)(t-i)\right]^{1-\lfloor \frac{a}{i+1}\rfloor}\times p_{\text{RC},a}(k,t),
\end{equation} 
where $p_{\text{RC},a}(k,t)$ are monic polynomials reported in~\cref{tab:rcgen}.
\begin{table}
\begin{center}
 \begin{tabular}{||c| c |} 
 \hline
 $a$& $p_{\text{RC,}a}(k,t)$\\
 \hline
$ 0 $ & $ 1 $  \\
$ 1 $ & $ 0 $  \\
$ 2 $ & $ 1 $  \\
$ 3 $ & $ 1 $  \\
$ 4 $ & $k^2 t^2 + k^2 t + k t^2 + 25 k t - 30 k - 30 t + 36 $  \\
$ 5 $ & $k^2 t^2 + 5 k^2 t + 5 k t^2 + 49 k t - 84 k - 84 t + 144 $  \\
$ 6 $ & $k^5 t^5 + 2 k^5 t^4 + 2 k^4 t^5 - 5 k^5 t^3 + 20 k^4 t^4 - 5 k^3 t^5 + \text{Lower order terms}$  \\
$ 7 $ & $ k^5 t^5 + 10 k^5 t^4 + 10 k^4 t^5 + 11 k^5 t^3 + 28 k^4 t^4 + 11 k^3 t^5 + \text{Lower order terms}$  \\
$ 8 $ & $k^8 t^8 + 4 k^8 t^7 + 4 k^7 t^8 - 16 k^8 t^6 - 4 k^7 t^7 - 16 k^6 t^8+ \text{Lower order terms}$  \\
$ 9 $ & $ k^9 t^9 + 14 k^9 t^8 + 14 k^8 t^9 + 12 k^9 t^7 - 4 k^8 t^8 + 12 k^7 t^9+ \text{Lower order terms}$  \\
$ 10 $ & $ k^{12} t^{12} + 5 k^{12} t^{11} + 5 k^{11} t^{12} - 46 k^{12} t^{10} - 43 k^{11} t^{11} - 46 k^{10} t^{12}+ \text{Lower order terms} $\\
\hline
 \end{tabular}
\end{center}
 \caption{The polynomials $p_{\text{RC},a}(k,t)$ used in~\cref{eq:rcgent}.}\label{tab:rcgen}
\end{table}
\end{itemize}

\subsection{The Cauchy identity and \texorpdfstring{$\tr \left[\rho_\lambda (RCRC) \right]$}{} for \texorpdfstring{$l(\lambda)=2$}{}}\label{sec:RCRC_calc}
The techniques of~\cref{sec:RC_calc} can also be used to find exact expressions for $\tr[\rho_\lambda(RCRC) ]$. This quantity is more complicated than $\tr[\rho_\lambda(RC) ]$, and the calculations are less manageable. Therefore, we can compute a smaller number of terms. 
\par
We remind the reader that our main goal is to compute the terms in the expansion~\cref{eq:expansion-formula} and show that they stay small as one increases $k$ and $t$. This goal is already achieved by the computations in~\cref{sec:RC_calc} and using the upper bounds of~\cref{sec:pleth}. However, it is beneficial to compute $\tr[\rho_\lambda(RCRC) ]$ in order to derive more accurate lower bounds for the permanent moments.
\par
To state the results in a compact form, define $\text{Rect}_a(k,t)$ to be the set of Young diagrams of $a$ boxes that fit in a rectangle of $t$ rows and $k$ columns. Let $K$ be the matrix of Kostka numbers restricted to $\text{Rect}_a(k,t)$, i.e., $(K)_{\alpha, \beta}=K_{\alpha\beta}$ where $(\cdot)_{i,j}$ indicate the matrix element in the $i$th row and the $j$th column, and $\alpha, \beta \in \text{Rect}_a(k,t)$. Similarly, define $\hat K$ to be the Kostka matrix restricted to the Young diagrams in $\text{Rect}_a(t,k)$. Moreover, define 
\begin{equation}\label{eq:omega_def}
 \Omega_{\mu}^{r,s} :=r! \frac{\prod_{i=1}^r \mu_i! (s-\mu_i)!}{\#_\mu(0)!\cdots \#_\mu(s)! }, \text{ with } \#_\mu(0) = r - l(r), 
\end{equation}  
where $\#_\mu(i)$ counts the number of times that the number $i$ appears in the Young diagram $\mu$. Lastly, let $\Delta$ to be the matrix that maps Young diagrams to their conjugate. With this rather lengthy list of definitions, we can state the main result of this section:
\begin{thm}[\texorpdfstring{$\tr \left[\rho_\lambda (RCRC) \right]$}{} when \texorpdfstring{$l(\lambda)\geq 2$}{}]\label{thm:RCRC-formula}
Let $k$ and $t$ be positive integers and $a \geq 0$. Then,
\begin{equation} \label{eq:RCRC-formula-from-Q}
\tr[\rho_{(kt-a,a)}(RC) ] = \Gamma(k,t,a) -\Gamma(k,t,a-1) \text{ if }a>0 \quad\text{and}\quad \Gamma(k,t,0)\quad\text{ when }a=0, 
\end{equation}
where $\Gamma(k,t, a)$ is defined as: 
\[ \Gamma(k,t, a) := \sum_{\alpha_1,\cdots, \alpha_8 \in \text{Rect}_a(k,t)} K_{\alpha_1 \alpha_2} \Omega_{\alpha_2}^{k,t} K_{\alpha_3\alpha_2} K_{\hat \alpha_3 \hat \alpha_4} \Omega_{\hat \alpha_4}^{t,k}K_{\hat \alpha_5 \hat \alpha_4}K_{\alpha_5\alpha_6}\Omega_{\alpha_6}^{k,t} K_{\alpha_7\alpha_6} K_{\hat \alpha_7 \hat\alpha_8 } \Omega_{\hat \alpha_8}^{t,k} K_{\hat \alpha_1\hat \alpha_8},\]
with $K_{\mu\nu}$ being the Kostka number and $\Omega_{\mu}^{a,b}$ defined in~\cref{eq:omega_def}. In the more compact matrix form, this reads as,
\[\Gamma(k,t,a) = \tr\left[(K \Omega K^T \Delta \hat K \hat \Omega \hat K^T \Delta)^2\right],\]
where $\Omega$ is the diagonal matrix with diagonal elements $\Omega^{k,t}_\mu$, and $\hat \Omega$ is the matrix with diagonal elements $\Omega^{t,k}_\mu$.
\end{thm}
\begin{proof}
To derive explicit formulas, consider expanding $ \tr[RCRCX^{\ot kt}]$ using the Cauchy identity. Skipping a few lines of algebra and following the logic of the proof of~\cref{thm:RC-formula}, we obtain, 
\begin{multline}
    \tr[ \rho_\lambda(RCRCX^{\ot kt})] = \EE \left[ \prod_{i=1,j=1}^{k,t}(\alpha_i^1 \overline{\beta}_j^1  x_1^{1/4}+\alpha_i^2 \overline{\beta}_j^2 x_2^{1/4} ) \right] 
    \left[ \prod_{i=1,j=1}^{t,k}(\beta_i^1 \overline{\gamma}_j^1  x_1^{1/4}+\beta_i^2 \overline{\gamma}_j^2 x_2^{1/4} ) \right]\times\\
    \left[ \prod_{i=1,j=1}^{k,t}(\gamma_i^1 \overline{\delta}_j^1  x_1^{1/4}+\gamma_i^2 \overline{\delta}_j^2 x_2^{1/4} ) \right] 
    \left[ \prod_{i=1,j=1}^{t,k}(\delta_i^1 \overline{\alpha}_j^1  x_1^{1/4}+\delta_i^2 \overline{\alpha}_j^2 x_2^{1/4} ) \right].
\end{multline}
Where similar to~\cref{sec:RC_calc}, the average is over vectors $\alpha^1, \beta^1, \gamma^1, \delta^1,\alpha^2, \beta^2, \gamma^2, \delta^2$ with elements distributed according to the i.i.d. standard complex Gaussian distribution. 
\par
Assume that $v^1$ and $v^2$ are vectors of dimension $k$, and $w^1$ and $w^2$ are vectors of dimension $t$. To further simplify our calculations, define $s_\lambda(v^1, v^2) : = (v^1_1\cdots v^1_k)^t s_\lambda\left(\frac{v^2_1}{v^1_1},\cdots,\frac{v^2_k}{v^1_k}\right)$. Similarly, when the dimension of the input vector is $t$, define $s_\lambda(w^1, w^2)$ by exchanging $t$ and $k$ in the above definition. With this notational remark and using the Cauchy identity, we get,
\begin{multline}
    \tr[ \rho_\lambda(RCRCX^{\ot kt})] = \\
    \EE \sum_{\lambda_1,\lambda_2,\lambda_3,\lambda_4} 
    s_{\lambda_1}(\alpha^1,\alpha^2)\overline{s_{\hat \lambda_1}(\beta^1,\beta^2)}
    s_{\lambda_2}(\beta^1,\beta^2)\overline{s_{\hat \lambda_2}(\gamma^1,\gamma^2)}
    s_{\lambda_3}(\gamma^1,\gamma^2)\overline{s_{\hat \lambda_3}(\delta^1,\delta^2)}
    s_{\lambda_4}(\delta^1,\delta^2)\overline{s_{\hat \lambda_4}(\alpha^1,\alpha^2)}\times\\
    x_1^{kt-\frac14(|\lambda_1|+|\lambda_2|+|\lambda_3|+|\lambda_4|)}x_2^{\frac14(|\lambda_1|+|\lambda_2|+|\lambda_3|+|\lambda_4|)}.
\end{multline}
Note that $\EE_\alpha s_\mu(\alpha^1,\alpha^2)\overline{s_\nu(\alpha^1,\alpha^2)} =\sum_\lambda K_{\mu\lambda} K_{\nu\lambda}\, \EE[m_\lambda(\alpha^1,\alpha^2)\overline{m_\lambda(\alpha^1,\alpha^2)}]$. 
\par 
After another few lines of algebra and using the techniques of the previous section, we have
\[\tr[ \rho_\lambda(RCRCX^{\ot kt})] =\sum_{a=0}^{kt} x_1^{kt-a}x_2^a \Gamma_a^{k,t}, \]
with 
\[ \Gamma_a^{k,t} = \sum_{\alpha_1,\cdots, \alpha_8 \in \text{Rect}_a(k,t)} K_{\alpha_1 \alpha_2} \Omega_{\alpha_2}^{k,t} K_{\alpha_3\alpha_2} K_{\hat \alpha_3 \hat \alpha_4} \Omega_{\hat \alpha_4}^{t,k}K_{\hat \alpha_5 \hat \alpha_4}K_{\alpha_5\alpha_6}\Omega_{\alpha_6}^{k,t} K_{\alpha_7\alpha_6} K_{\hat \alpha_7 \hat\alpha_8 } \Omega_{\hat \alpha_8}^{t,k} K_{\hat \alpha_1\hat \alpha_8}.\]
After basic manipulations, this can be rewritten as
\[\Gamma(k,t,a) = \tr\left[(K \Omega K^T \Delta \hat K \hat \Omega \hat K^T \Delta)^2\right].\]
\end{proof}
The matrix form of~\cref{thm:RCRC-formula} is more efficient, and is used in our computerized calculations. Here are the results:
\begin{itemize}
\item{\texorpdfstring{$t=2$}{} : \texorpdfstring{$\tr \left[\rho_\lambda(RCRC) \right]$}{}.}
Here, we report the explicit results for $\tr [\rho_\lambda(RCRC)]$ when $t=2$. For this specific case, the calculations are very simple, and we do not need to use the machinery of~\cref{thm:RCRC-formula}. As it was shown in~\cref{eq:pleth-k-2} all of the plethysm coefficients are either $0$ or $1$. Therefore, both inequalities in~\cref{thm:pleth_ineq} are tight, and $\tr [\rho_\lambda(RCRC)]=\left(\tr [\rho_\lambda(RC)]\right)^2$.
\par
We have already computed $\tr [\rho_\lambda(RC)]$ in~\cref{rem:t-2}, hence,
\begin{equation}\label{eq:rcrct2}
\tr[\rho_{(2k-a,a)}(RCRC)] =2^{2k} (k!)^4 \times \delta_{a=0\text{ mod } 2}\times  2^{-2a} \left(\frac{\binom{a}{a/2}}{\binom{k}{a/2}}\right)^2. \end{equation}
\item{\texorpdfstring{$t=3$}{} : \texorpdfstring{$\tr \left[\rho_\lambda(RCRC) \right]$}{}.}
We report the result of $\tr [ \rho_\lambda(RCRC)]$, with $2\leq k\leq 10$ and $k\geq a$.
\par 
Again, we can analyze the results of our computerized calculations of~\cref{thm:RCRC-formula} and find that they all have the form of 
\begin{equation}\label{eq:rcrct3form}
    \tr[\rho_{(2k-a,a)}(RCRC)] =k!^{2t}t!^{2k}\frac{1}{\left[Q^3_{a}(k)\right]^2}\times p^3_{\text{RCRC,}a}(k),
\end{equation}
where $Q^3_{a}(k)$ is defined in~\cref{eq:q3lkdef} and the monic polynomials $p^3_{\text{RCRC,}a}(k)$ are reported in~\cref{tab:rcrct3}.
\begin{table}
\begin{center}
 \begin{tabular}{||c| c |} 
 \hline
 $a$& $p_{\text{RCRC,a}}^3(k)$\\
 \hline
$ 0 $ & $ 1 $  \\
$ 1 $ & $ 0 $  \\
$ 2 $ & $ 1 $  \\
$ 3 $ & $ 1 $  \\
$ 4 $ & $ 1 $  \\
$ 5 $ & $ 1 $  \\
$ 6 $ & $  k^4 - 10/9k^3 + 1729/81k^2 - 7880/81k + 8464/81$  \\
$ 7 $ & $ 1 $  \\
$ 8 $ & $ k^4 + 14/9k^3 + 6037/81k^2 - 45976/81k + 75076/8 $  \\
$ 9 $ & $  k^4 - 16/3k^3 + 1582/81k^2 - 688/9k + 355/3$  \\
$ 10 $ & $  k^4 + 62/9k^3 + 18865/81k^2 - 21488/9k + 5184 $\\
\hline
 \end{tabular}
\end{center}
\caption{Value of the monic polynomial $p^3_{\text{RCRC,a}}(k)$ used in~\cref{eq:rcrct3form}.}\label{tab:rcrct3}
\end{table}
\item Arbitrary \texorpdfstring{$t$}{} : 
Lastly, we report $\tr \rho_{(kt-a,a)}(RCRC)$ for small values of $a$. We assume that $k,t\geq a$. Like the case of $\tr \left[\rho_\lambda(RC) \right]$, the polynomials are too complicated and do not fit in the page. By close inspection, we can see that:
\begin{equation}\label{eq:rcrcgent}
    \tr \rho_{(kt-a,a)}(RCRC) = \left((k!)^t (t!)^k \times 2^{1-(-1)^a}  \times \prod_{i=0}^{a-1} \left[(k-i)(t-i)\right]^{1-\lfloor \frac{a}{i+1}\rfloor} \right)^2\times p_{\text{RCRC},a}(k,t), 
\end{equation}
where $p_{\text{RCRC},a}(k,t)$ is reported in~\cref{tab:rcrcgent}. 
\begin{table}
\begin{center}
 \begin{tabular}{||c| c |} 
 \hline
 $a$& $p_{\text{RCRC,}a}(k,t)$\\
 \hline
$ 0 $ & $ 1 $  \\
$ 1 $ & $ 0 $  \\
$ 2 $ & $ 1 $  \\
$ 3 $ & $ 1 $  \\
$ 4 $ & $ k^4 t^4 + 2 k^4 t^3 + 2 k^3 t^4 + k^4 t^2 - 20 k^3 t^3 + k^2 t^4 + \text{Lower order terms}   $  \\
$ 5 $ & $ k^4  t^4 + 10  k^4  t^3 + 10  k^3  t^4 + 25  k^4  t^2 - 140  k^3  t^3 + 25  k^2  t^4 +\text{Lower order terms}   $  \\
$ 6 $ & $ k^{10}  t^{10} + 4  k^{10}  t^9 + 4  k^9  t^{10} - 6  k^{10}  t^8 - 56  k^9  t^9 - 6  k^8  t^{10} +  \text{Lower order terms}  $  \\
$ 7 $ & $ k^{10}  t^{10} + 20  k^{10}  t^9 + 20  k^9  t^{10} + 122  k^{10}  t^8 - 104  k^9  t^9 + 122  k^8  t^{10} +  \text{Lower order terms}   $  \\
\hline
 \end{tabular}
\caption{Polynomials $p_{\text{RCRC},a}(k,t)$ used in~\cref{eq:rcrcgent}} \label{tab:rcrcgent}
\end{center}
\end{table}
\end{itemize}
\subsection{Comments on \texorpdfstring{$\tr \left[\rho_\lambda (RC) \right]$}{} for general  \texorpdfstring{$\lambda$}{}}\label{sec:Comments_calc}\label{sec:comments}
We briefly comment on general methods for computing $\tr \left[\rho_\lambda (RC) \right]$ when $l(\lambda)\geq 2$. We will not be able to provide explicit formulas like the ones derived in~\cref{thm:RC-formula} and~\cref{thm:RCRC-formula}, but we can establish algorithmic procedures for computing $\tr \left[\rho_\lambda (RC) \right]$ and report some partial results. 
\par
As our starting point, we construct a class of representations of the permutation group $S_{kt}$.
Consider a Young diagram $\lambda$, where $|\lambda|=kt$. Let $\omega_\lambda$ be the set of all tuples of length $kt$, where each tuple contains $\lambda_1$ many 1's, $\lambda_2$ many 2's, an so on. It is easy to see that \[|\omega_\lambda| = \left( \begin{array}{c}kt\\ \lambda_1 ,\lambda_2,\cdots, \lambda_{kt}\end{array}\right).\]
The permutation group $S_{kt}$ naturally acts on $\omega_\lambda$ by permuting the elements of the tuples:
\[ \pi(v_1,v_2, \cdots, v_{kt}) = (v_{\pi^{-1}(1)},v_{\pi^{-1}(2)},\cdots, v_{\pi^{-1}(kt)}).\]
We can use this action to define a representation of $S_{kt}$.
Consider a vector space with the basis elements $e_v$ for every $v\in \omega_\lambda$. Clearly, the dimension of this vector space is $|\omega_\lambda|$. An element of the permutation group $\pi \in S_{kt}$ acts on this space by mapping basis element's index tuples under the permutation group: $e_v \rightarrow e_{\pi(v)}$. In this way, we construct a representation of the permutation group of dimension $|\omega_\lambda|$ which we call $\Psi_\lambda$.
\par
when $\lambda = (kt)$, $\Psi_\lambda$ is the trivial representation, as $\omega_\lambda$ has only one element. This is the only case where $\Psi_\lambda$ is irreducible. The next simple case is $\lambda = (kt-1,1)$, where $\Psi_\lambda$ is nothing but the standard $kt$ dimensional representation of the symmetric group. In general, this representation can be decomposed to irreducible representations, each one appearing with a degeneracy given by the Kostka numbers:
\begin{equation}\label{eq:psidec}
    \Psi_{\lambda}=\bigoplus_{\mu} (\rho_\mu)^{\times K_{\mu\lambda}}.
\end{equation}
See~\cite{fulton2013representation}, section 4.3 for a proof. Therefore, we have that
\[\tr [\Psi_\lambda (RC)] =  \sum_{\lambda \leq \mu} K_{\mu \lambda} \tr [\rho_\mu (RC)], \]
where the condition $\lambda \leq \mu$ (in lexicographical order) follows from~\cref{eq:zerokost}. Using the inverse Kostka numbers (defined the in paragraph above~\cref{eq:zerokostinv}) and~\cref{eq:zerokostinv} we have
\begin{equation}\label{eq:psirc}
    \tr [\rho_\lambda (RC)] =  \sum_{\lambda \leq \mu} (K^{-1})_{\mu \lambda} \tr [\Psi_\mu (RC)].
\end{equation}
\Cref{eq:psirc} shows that we can compute $ \tr [\rho_\lambda (RC)]$ if we know $ \tr [\Psi_\mu (RC)]$. More importantly, it indicates that we only need to compute $ \tr [\Psi_\mu (RC)]$ for $\lambda \leq \mu$, which limits our calculations to a small number of Young diagrams if most of the boxes of $\lambda$ are in its first row.
\par
We will shortly state the main result of this sections, generalizing~\cref{thm:RC-formula}. But before that, we need to define \emph{row and column types} of matrices:
\begin{dfn}[Row type and column type]\label{dfn:rctypemat}
Let A be a $k\times t$ matrix filled with numbers $0,1,\cdots, l$. Define a tuple $u = (u^1,u^2,\cdots, u^k)$, where each $u^i$ itself is a vector defined by,
\[u^i_n:= \text{number of times that }n\text{ appears in the }i\text{-th row of }A \quad(\text{defined for }1\leq n \leq l). \]
We call the tuple $u$ the ``row type'' of the matrix $A$. Similarly, we define the ``column type'' of $A$ as the tuple $w=(w^1,w^2,\cdots, w^t)$, with 
\[w^i_n:= \text{number of times that }n\text{ appears in the }i\text{-th column of }A \quad(\text{defined for } 1\leq n \leq l). \]
Moreover, we define the coefficients $\text{IB}^{(l)}_{vw}$ to be the number of $0,1,\cdots, l$ matrices with the row type $v$, and column type $w$.
\end{dfn}
\begin{thm}\label{thm:RC-gen-formula}
Let $l=l(\lambda)$ and $\text{RowType}(\lambda)$ be the set of all tuples $u =(u^1,\cdots,u^k)$, such that
\[ \sum_{n=1}^{l} u^i_n = \lambda_n,\quad \text{for }1\leq i\leq k\]
and $u^1\geq u^2\geq \cdots\geq u^k$ in lexicographical order. Similarly, $\text{ColumnType}(\lambda)$ is defined to be the set of all tuples $w =(w^1,\cdots,w^t)$ in lexicographical order where,
\[ \sum_{n=1}^{l(\lambda)} w^i_n = \lambda_n,\quad \text{for }1\leq i\leq t.\]
Then,
\begin{multline}\label{eq:badlylong}
\tr [ \Psi_\lambda(RC)] = k!t!\times \\ \sum_{\small\begin{array}{c}
u\in \text{RowType}(\lambda)\\ 
w\in \text{ColumnType}(\lambda)
\end{array} } 
\left(\text{IB}^{(l)}_{uw}\right)^2 
\frac{
\prod_{i=1}^k \left(u_1^i!\cdots u_l^i!\left(t-\sum_j u_j^i\right)!\right)
\prod_{i=1}^t \left(w_1^i!\cdots w_l^i!\left(k-\sum_j w_j^i\right)!\right)}
{|\text{Stab}_{S_k}(u)||\text{Stab}_{S_t}(w)|},
\end{multline}
with \[|\text{Stab}_{S_k}(u)|=\prod_{v\text{ a vector of length l}} (\text{number of times that }v \text{ appears in }u)!,\]
and 
\[|\text{Stab}_{S_t}(w)|=\prod_{v\text{ a vector of length l}} (\text{number of times that }v \text{ appears in }w)!.\]
\end{thm}
\begin{proof}
Recall that the representation $\Psi_\lambda$ is nothing but the action of the unitary permutation group on set $\omega_\lambda$. Therefore,
\[ \tr[\Psi_\lambda(rc) ]=\sum_{x \in \omega_\lambda} \langle e_x,r(c(e_x))\rangle = \sum_{x\in \omega_\lambda} \delta[{x=r(c(x))}]=\sum_{x\in \omega_\lambda} \delta[r(x)=c(x)].\]
Hence,
\[ \sum_{r\in R, c\in C} \tr[\Psi_\lambda(rc) ]=\sum_{x\in \omega_\lambda} \left(\sum_{r\in R, c\in C} \delta[r(x)=c(x)]\right).\]
If $u$ and $w$ are row and column types of a matrix $x$ (as defined in~\cref{dfn:rctypemat}), then one can see that
\[\sum_{r\in R, c\in C} \delta[r(x)=c(x)] = \text{IB}^{(l)}_{uw} 
\prod_{i=1}^k \left(u_1^i!\cdots u_l^i!\left(t-\sum_j u_j^i\right)!\right)
\prod_{i=1}^t \left(w_1^i!\cdots w_l^i!\left(k-\sum_j w_j^i\right)!\right).\]
It is easy to see that the number of times that row and column types $u$ and $v$ show up in $\sum_{x\in \omega_\lambda}$ is $\text{IB}^{(l)}_{uw}$, and there are 
$k!/|\text{Stab}_{S_k}(u)|$ many row types that are equal to $u$ up to permutations, and there are 
$t!/|\text{Stab}_{S_t}(w)|$ many column types with that equal to $w$ up to permutations. Therefore, we have,
\begin{multline*}
    \sum_{r\in R, c\in C} \tr[\Psi_\lambda(rc) ]= \sum_{\small\begin{array}{c}
u\in \text{RowType}(\lambda)\\ 
w\in \text{ColumnType}(\lambda)
\end{array} } \frac{k!}{|\text{Stab}_{S_k}(u)!}\frac{t!}{|\text{Stab}_{S_t}(w)!}\times \\
\text{IB}^{(l)}_{uw} 
\prod_{i=1}^k \left(u_1^i!\cdots u_l^i!\left(t-\sum_j u_j^i\right)!\right)
\prod_{i=1}^t \left(w_1^i!\cdots w_l^i!\left(k-\sum_j w_j^i\right)!\right).
\end{multline*}
which is the same as~\cref{eq:badlylong}.
\end{proof}
\par
We can use~\cref{thm:RC-gen-formula} and~\cref{eq:psirc} to compute $\tr [\rho_\lambda (RC)]$ for a few Young diagrams of depth $3$. If $\lambda = (kt-2,1,1)$,$(kt-3,2,1)$, $(kt-4,3,1)$, then $\tr [\rho_\lambda (RC)]=0$. This is because $\Pl^{k,t}_{(kt-2,1,1)}=\Pl^{k,t}_{(kt-3,2,1)}=\Pl^{k,t}_{(kt-4,3,1)}=0$ as was shown in~\cref{sec:pleth-ex}. The first interesting case that we can compute with the techniques of this section is $\lambda = (kt-4,2,2)$, which leads to 
\begin{equation}\label{eq:lambda22}
    \tr [ \rho_{(kt-4,2,2)} (RC)] = \frac{( k-2) ( t-2)}{(k-1) k^2 ( t-1) t^2} \quad\text{ for }k,t\geq 3.
\end{equation}
Again, this calculation is very lengthy and computerized. Note that we do not have an explicit formula for $\text{IB}_{uw}^{(l(\lambda))}$ in~\cref{eq:badlylong}, and it has been computed by brute-force counting. Because $\Pl^{k,t}_{(kt-2,1,1)}=1$, $\tr [ \rho_{(kt-4,2,2)} (RCRC)] =(\tr [ \rho_{(kt-4,2,2)} (RC)] )^2$.
\par
Independently, in~\cref{eq:expres}, we report the results of a brute force calculation of $\tr[\rho_\lambda(RC)]$ for all of the Young diagrams contributing to the cases $t=3,k=3$, and $t=3, k=4$. Because those calculations are brute force, we do not know their dependence on $k$ and $t$.
\subsection{Moment lower bounds}\label{sec:lbmgc}
Having computed $\tr [\rho_{\lambda} (RCRC)]$ for a number of Young diagrams, we can use~\cref{thm:main_expansion_formula} to lower bound $\EE_{M\sim \Gauss k} |\Perm(M)|^{2t}$. We remind the reader that in order to exactly compute all moments of permanents, we need to know the quantity $\tr [\rho_{\lambda} (RCRC)]$ for all Young diagrams of $kt$ boxes and depth at most $\min(k,t)$. But as we observed, the contribution of Young diagrams to the permanent moment rapidly vanishes as we decrease the number of boxes in the first row. This strongly suggest that the few Young diagrams for which we have computed $\tr [\rho_{\lambda} (RCRC)]$ could provide a good lower bound for the moments of permanents.
\par 
When evaluating the permanent moment expansion~\cref{thm:main_expansion_formula}, we do not attempt to incorporate all of the terms that we have computed. Rather, we use only $4$ or $5$ terms that give a simple lower bound, and remind the reader that one can obtain slightly better lower bounds by including more terms.
\par
The first term that we have computed correspond to the trivial representation: \[f^{(kt)} \tr \rho_{(kt)} (RCRC) = k!^{2t}t!^{2k}.\] The next term, $f^{(kt-1,1)} \tr \rho_{(kt-1,1)} (RCRC)$, is always zero. The third term is \[f^{(kt-2,2)} \tr \rho_{(kt-2,2)} (RCRC) =  (1/2k^2t^2 - 3/2kt) (1/(k^2t^2))=1/2-3/(2kt),\] and so on. If we include the fourth term (see~\cref{tab:rcrcgent}) we obtain: 
\begin{multline}\label{eq:upb2} \EE_{M\sim \Gauss k} |\Perm(M)|^{2t} \geq f^{(kt)} \tr \rho_{(kt)}+f^{(kt-1,1)} \tr \rho_{(kt-1,1)}+\\f^{(kt-2,2)} \tr \rho_{(kt-2,2)}+f^{(kt-3,3)} \tr \rho_{(kt-3,3)} = 
k!^{2t} t!^{2k}\left(3/2 + 7/(6 kt) -16/(k^2t^2)+40/(3k^3t^3)\right).
\end{multline}
Note that this expression is valid for when $\min(t,k)\geq 3$. \par
In~\cref{sec:Comments_calc}, we also studied a few Young diagrams of three rows. We observed the first non-zero contribution comes form $ \rho_{(kt-4,2,2)} (RC)$. This term gives the following contribution to $\EE_{M\sim \Gauss k} |\Perm(M)|^{2t}$ (when $t,k\geq 3$):
 \[f^{(kt-4,2,2)}(\tr[\rho_{(kt-4,2,2)}(RCRC)]) =k!^{2t}t!^{2k}\frac1{12}\frac{(-2 + k)^2 (-2 + t)^2 (-1 + k t) (3 + k t) (4 + k t)}{(-1 + k)^2 k^3 (-1 + t)^2 t^3}. \]
This term is relatively small and goes to $1/12 \times k!^{2t}t!^{2k}$ as $t,k\rightarrow\infty$. 
\par
We can include the above term in~\cref{eq:upb2}, and prove the following theorem:
\begin{thm}[Main lower bound on the moments of permanents]\label{thm:mainlowerbound}
Let $M$ be drawn from the distribution of $k\times k$ random i.i.d. complex Gaussian matrices. Let $k,t\geq 3$, then  
\[\EE_{M\sim \Gauss k} |\Perm(M)|^{2t} \geq \frac{k!^{2t}t!^{2k}}{(kt)!} \left[3/2 + 7/(6 kt) -16/(k^2t^2)+40/(3k^3t^3)\right]\]
Moreover, if $k,t\geq 4$,
\[\EE_{M\sim \Gauss k} |\Perm(M)|^{2t} \geq 1.625\, \frac{k!^{2t}t!^{2k}}{(kt)!}\]
\end{thm}
We will later conjecture that this lower bound is close to being tight (up to a small multiplicative factor). 
\subsection{Exact moment calculations} \label{sec:alg}
We introduce an algorithm that can be used to explicitly compute the moments of permanents and report the exact moments for the cases of $k=3$, $1\leq t\leq 100$ and $k=4$, $1\leq t\leq 10$. Sadly, our techniques are ineffective for larger moments. (We remind the reader the $k\leftrightarrow t$ symmetry of the moments, so our calculations cover the cases of $t=3$, $1\leq k\leq 100$ and $t=4$, $1\leq k\leq 10$ as well.)
\par 
We comment that it is possible to compute the permanent moments using direct sampling from the permanent distribution. Unfortunately, such computations typically take a very long time to converge for large moments, as the permanent distribution has a very large tail (note that the first interesting case that is not yet computed using our algorithm is the $10$-th moment of the permanent of $5\times 5$ matrices). This, in addition to the fact that the computation of the individual permanent samples is very expensive, makes direct sampling a rather inefficient approach for calculating the permanent moments. Hence, we found it necessary to develop a new algorithm.
\par
To start the technical discussion of our algorithm, define $\rho_{\mathrm{std}}(\pi)$ for $\pi \in S_k$ to be the standard $k\times k$ dimensional representation of the permutation $\pi$, i.e., $(\rho_{\mathrm{std}}(\pi))_{ij}=\delta_{\pi(i),j}$. Now, consider a $k\times k$ matrix $X$ and expand $\Perm(X)^t$ as a polynomial in its matrix elements $X_{ij}$ :
\begin{equation}\label{eq:momo}
 \left[\Perm (X)\right]^t = \left[\sum_{\pi \in S_k} \prod_{i=1}^k X_{i \pi(i)}\right]^t=\sum_{\pi_1,\cdots,\pi_t\in S_k} \prod_{i,j=1}^{k} X_{ij}^{\sum_{n=1}^t (\rho_{\mathrm{std}}(\pi_n))_{ij}}.
\end{equation}
Before proceeding, we make a few definitions:
\begin{dfn}\label{def:ap}
    A~\emph{(weak) $t$-magic square} is a $k\times k$ matrix of integer entries where the row and column sums are equal to $t$. We define $A_+(t,k)$ to be the set of all $k\times k$ weak $t$-magic squares.   
\end{dfn}
As a result of the Birkhoff-von Neumann theorem (or more fundamentally, the Hall's marriage theorem), every matrix in $A_+(t,k)$ can be decomposed as sum of $t$ permutation matrices in their standard representation. The decomposition, however, can be non-unique. 
\begin{dfn}\label{def:qa}
Let $A \in A_+(t,k)$. Define $Q(A)$ to be the set of all Birkhoff-von Neumann decompositions of $A$:
\[Q(A):=\Big\{(\pi_1,\pi_2,\cdots ,\pi_t) \in S_k^{\times t}\,\Big|\, \sum_{i=1}^t \rho_{\mathrm{std}}(\pi_i)=A \Big\}. \]
Similarly, define $\hat Q(A)$ as:
\[\hat Q(A):=\Big\{(a_\pi)_{\pi \in S_k} \in \ZZ_{\geq 0}^{\times t!}\,\Big|\, \sum_{\pi \in S_k} a_\pi \rho_{\mathrm{std}}(\pi)=A \Big\}. \]
\end{dfn}
In simple words $\hat Q(A)$ characterizes the solutions to the Birkhoff-von Neumann problem ignoring the order of the permutations, therefore,
\begin{equation}
    |Q(A)| = \sum_{(a_1,\cdots,a_{k!})\in \hat Q(A)}  \binom{t}{a_1,\cdots ,a_{k!}}.
\end{equation}

To simplify~\cref{eq:momo}, note that the matrix $\sum_{n=1}^t (\rho_{\mathrm{std}}(\pi_n))$ is a (weak) $t$-magic square. Therefore, we have
\begin{equation}\label{eq:momo1}
 \left[\Perm (X)\right]^t = \sum_{A \in A_+(t,k)} |Q(A)| \prod_{i,j=1}^k X_{ij}^{A_{ij}}.
\end{equation}
This immediately leads to
\begin{multline}\label{eq:perm_gen_form}
    \EE_{M \sim \Gauss k} |\Perm(M)|^{2t} = \sum_{A \in A_+(t,k)} |Q(A)|^2 \prod_{i,j=1}^k A_{ij}!=\\
    \sum_{A \in A_+(t,k)} \left(\sum_{(a_1,\cdots,a_{k!})\in \hat Q(A)} \binom{t}{a_1,\cdots ,a_{k!}}\right)^2 \prod_{i,j=1}^k A_{ij}!.
\end{multline}

We implemented an algorithm to compute this formula for $k=3$ and $k=4$, with results reported in the table below:
\par
\begin{center}
{\small \begin{tabular}{|c||c|c||c|c|}\hline
    $t$ & $\EE \left| \Perm (M_{k=3})\right|^{2t}$ &$\frac{\EE \left| \Perm (M_{k=3})\right|^{2t}}{(t!^{2k}k!^{2t}/(kt)!)}$ & $\EE \left| \Perm (M_{k=4})\right|^{2t}$ & $\frac{\EE \left| \Perm (M_{k=4})\right|^{2t}}{(t!^{2k}k!^{2t}/(kt)!)}$ \\\hline
1&6&                        1.000& 24&1.000  \\
2&144&                      1.250& 2880&1.367\\
3&8784&                     1.464& 1092096&1.629\\
4&1092096&                  1.629& 1031049216&1.780\\
5&241920000&                1.752& 2076785049600&1.853\\
6&87157555200&              1.840& 7993079444275200&1.880\\
7&47800527667200&           1.901& 54084131717382144000&1.882\\
8&37952477724672000&        1.942& 602679058543248880435200&1.873\\
9&41935572001986969600&     1.969& 10493450241312304614762086400&1.860\\
10&62462555284423311360000& 1.984& 273409548213807664837794201600000&1.845  \\\hline
\end{tabular}
}
\end{center}
See~\cref{sec:algt3} for a larger table for $k=3$.
\par 
As we can see in~\cref{fig:first}, the lower bound of~\cref{thm:mainlowerbound} is very close to the final exact value of the permanent moment when $t=3$, $3\leq k\leq 100$ (or $k=3$, $3\leq t\leq 100$) and $t=4$, $3\leq k\leq 10$ (or $k=4$, $3\leq t\leq 10$). For the regime of interest $k,t\geq 3$, the moment calculation shows a tendency to get closer to our $1.625 (k!)^{2k}(t!)^{2t}/((kt)!)$ lower bound as one increases $k$, and $t$.
\subsection{Conjectured upper bound and the moment growth conjecture}\label{sec:conjsection}
Now, we state our conjectured permanent moment upper bound and the permanent moment growth conjecture:
\begin{conj}[Permanent moment growth conjecture]\label{conj:maintext}
Let $M$ be drawn from the distribution of $k\times k$ random i.i.d. complex Gaussian matrices. If $3 \leq t,k\in \mathbb N$, then 
\begin{equation}
    \EE_{M \sim \Gauss k} |\Perm(M)|^{2t} \leq C \frac{k!^{2t} t!^{2k}}{(kt)!}.
\end{equation} 
The numerics suggest that $C\leq 2$.
Moreover, 
\begin{equation}
    \EE_{M \sim \Gauss k} |\Perm(M)|^{2t}/ \left(\frac{k!^{2t} t!^{2k}}{(kt)!}\right)\rightarrow C_k,
\end{equation} 
as $t$ gets large, for constants $C_k\leq 2$. As we increase $k$, the $C_k$ coefficients should asymptote a constant: $C_k \rightarrow C_\infty$. 
\end{conj}
\par
{\bf Evidence for~\cref{conj:maintext}:}
\begin{itemize}
    \item Using the explicit results reported in~\cref{sec:RCRC_calc} and~\cref{sec:RC_calc}, we can compute the contribution of $f^\lambda \tr [\rho_{(kt-l,l)}(RCRC)]$ to the average permanent moment. It is easy to see that this contribution becomes smaller as one increases $l$. In fact $f^\lambda \tr [\rho_{(kt-l,l)}(RCRC)]/ \left({k!^{2t} t!^{2k}}/{(kt)!}\right)$ is subleading in $k$ and $t$ when $l$ is odd. For even $l$, the order $1$ term is equal to $1/l!$. Computation of a few Young diagrams of more than $2$ rows indicate that they decay quickly as well.
    \item In~\cref{eq:expres}, we have explicitly computed the contribution of different Young diagrams for $k=3,t=3$ and $k=4, t=3$. The results show that the main contribution to the average permanent comes from the Young diagrams with a dominating first row.
    \item In~\cref{sec:alg} we computed the moments of permanents. Our results strongly supports~\cref{conj:maintext}. See~\cref{fig:first}.
    \item Lastly, we give a new heuristic argument for showing that $\EE_{M\sim \Gauss k} |\Perm(M)|^{2t}$ is close to $k!^{2t}t!^{2k}/(kt)!$. 
    \par 
    We start by defining two probability distributions $p_1$ and $p_2$ on the space of weak magic squares $A_+(t,k)$ (see~\cref{def:ap}). If $A\in A_+(t,k)$, define
    \begin{align}\label{eq:defprobs1}
        p_1(A)&:=|Q(A)|/k!^t, \quad \text{and,}\\
        p_2(A)&:=\binom{kt}{t,t,\cdots,t}^{-1}\prod_{i=1}^k \binom{t}{A_{i1},A_{i2},\cdots,A_{ik} }.\label{eq:defprobs2}
    \end{align}
    \begin{lem} Both $p_1$ and $p_2$ are probability distributions.
    \end{lem}
    \begin{proof}
    It is straightforward to see that $p_1$ is a probability distribution:
    \[\sum_{A \in A_+(t,k)} p_1(A) =\frac{1}{k!^t}\sum_{A \in A_+(t,k)} |Q(A)|=\frac{1}{k!^t} \Perm(J_k)^t=1,  \]
    where we used~\cref{eq:momo1} and $J_k$ is the $k\times k $ matrix of all ones.
    \par 
    To see that $p_2$ is a probability distribution, we need to show that
    \[\sum_{A \in A_+(t,k)} \left(\prod_{i,j=1}^k A_{ij}!\right)^{-1} = \frac{(kt)!}{t!^{2k}}.\]
    \par
    Consider the quantity $\exp(\sum_{i,j=1}^t X_{ij})$. By expanding the exponential function we obtain
    \[\exp\left(\sum_{i,j=1}^t X_{ij}\right)=\sum_{A_{11},A_{12},\cdots,A_{kk}=0}^{\infty}\,\,\prod_{i,j=1}^k X_{ij}^{A_{ij}}/A_{ij}!.\]

    If we set $X_{ij}=x_iy_j$, we have 
    \[\sum_{A \in A_+(t,k)} \left(\prod_{i,j=1}^k A_{ij}!\right)^{-1} =  \text{coefficient of }x_1^t\cdots x_k^t \,y_1^t\cdots y_k^t\text{ in } \exp\left(\sum_{i,j=1}^t x_{i}y_{j}\right).\]
    But,
    \begin{equation}\label{eq:insidelemm}
     \exp\left(\sum_{i,j=1}^t x_{i}y_{j}\right)=\exp\left(\left[\sum_{i=1}^t x_{i}\right]\left[\sum_{i=1}^t y_{i}\right]\right)=\sum_{n=0}^\infty \frac{1}{n!} \left[\sum_{i=1}^t x_{i}\right]^n\left[\sum_{i=1}^t y_{i}\right]^n.
     \end{equation}
    Now, the coefficient of $x_1^t\cdots x_k^t \,y_1^t\cdots y_k^t$ in~\cref{eq:insidelemm} is simply given by
    \[\frac{1}{(kt)!}\binom{kt}{t,t,\cdots, t}\binom{kt}{t,t,\cdots, t} = (kt)!/t!^{2k}.\]
    This concludes the proof of the lemma.
    \end{proof}
    Consider the process of choosing $t$ permutations $\pi_1,\cdots,\pi_t$ randomly, and constructing the matrix $A=\sum_{i=1}^t \rho_{\mathrm{std}}(\pi_i)$. When $t$ is very large, and in the limit of large $k$, one can heuristically assume that different matrix elements of $A$ are approximately independent, given the constraint that they form a magic square. This is equivalent to saying that all rows of $A$ are approximately independent.
    \par 
    Given that the probability of obtaining a particular vector $(v_1,\cdots, v_k)$, (with $\sum_i v_i=t$) is proportional to $ \binom{t}{v_1,\cdots ,v_k}$,
    the approximate independence of rows gives
    \[p_1(A)\appropto \prod_{i=1}^k \binom{t}{A_{i1},A_{i2},\cdots,A_{ik} }.  \]
    After normalizing the right hand side of this relation, we conclude that 
    \begin{equation}\label{eq:apprtwoprob}
        p_1(A) \approx p_2(A).
    \end{equation}
    We emphasize that the approximation in~\cref{eq:apprtwoprob} is not very tight, and we believe that it is only valid with multiplicative errors and in most probable cases. 
    \par 
    Now, we use~\cref{eq:perm_gen_form} and the definitions in~\cref{eq:defprobs1,eq:defprobs2} to conclude that
    \begin{multline} 
    \EE_{M \sim \Gauss k} |\Perm (M)|^{2t} =   \sum_{A \in A_+(t,k)} |Q(A)|^2 \prod_{i,j=1}^k A_{ij}! =\\
    \frac{k!^{2t}t!^{2k}}{(kt)!}\sum_{A\in A_+(t,k)} \frac{p_1(A)^2}{p_2(A)}\approx\frac{k!^{2t}t!^{2k}}{(kt)!} \exp\left[{D_2(p_1 \| p_2)}\right],
    \end{multline}
    where $D_2(p_1 \| p_2)$ is the $2$-R\'enyi divergence, a measure of distance between the two probability distributions $p_1$ and $p_2$. This relation provides a yet another proof for $\EE_{M \sim \Gauss k} |\Perm (M)|^{2t}\geq \frac{k!^{2t}t!^{2k}}{(kt)!}$, and suggests that $\EE_{M \sim \Gauss k} |\Perm (M)|^{2t}$ should be close to $\frac{k!^{2t}t!^{2k}}{(kt)!}$. In other words, the permanent moments are directly related to $D_2(p_1 \| p_2)$, which is nothing but a measure of independence of matrix elements of matrices drown from the probability distribution $p_1$.
\end{itemize}
\subsection{Concentration results and from the moment growth conjecture}\label{sec:concentration}
Now, we are going to assume the moment growth conjecture and compute the tail of the log-permanent distribution. In fact, we need to assume a different form of the conjecture that we define below:
\begin{conj}[Non-integer moment growth conjecture]\label{conj:ext}
Let $3\leq t \in \RR, 3\leq k\in \NN$. Then 
\[\EE_{M \sim \Gauss k} |\Perm(M)|^{2t} =C_{k,t} \frac{k!^{2t}t!^{2k }}{(kt)!}, \]
where $\lim_{k\rightarrow \infty} \log(C_{k,t})/k=0$.
\end{conj}
\begin{rem}
Note that in~\cref{conj:ext} we did not use our strongest lower bounds for the integer moments (which show that $C_{k,t}>1.625$ for $4\leq k,t\in \NN$), nor used our strongest upper bound conjecture (which implies that $C_{k,t}\leq 2$). We made the weakest conjecture that could give the desired result.
\end{rem}
\par 
To study the tail of the permanent distribution, we should properly normalize is. As it was mentioned before, $\EE_{M \sim \Gauss k} |\Perm(M)|^{2}=k!$. Therefore, we choose the normalization factor to be $1/\sqrt{k!}$. We take the logarithm of the distribution to obtain better behaved distribution. To summarize, we consider the following random variable:
\begin{equation}\label{eq:dfnY}
   Y_k := \frac{1}{k} \log \,\left[\frac{|\Perm(M)|}{\sqrt{k!}}\right]\quad \text{for}\quad M\sim \Gauss{k}. 
\end{equation}  
Our goal is to estimate $p_{Y_k}(y)$ (the probability density function of $Y_k$) for $y$ ranging from infinity to near $y=0$.
\par 
We follow the standard literature of large deviation theory (see, e.g.,~\cite{touchette2011basic}), and compute the rate function of $p_Y(y)$, i.e., a function $I(y)$ such that: \[p_{Y_k}(y) = \exp(-k I(y)+o(k)).\]
Here, we used the small $o$ notation where $x=o(k)$ means that $\lim_{k\rightarrow \infty} x/k=0$.
\par 
We start our analysis by computing the~\emph{scaled cumulant generating function} defined as,
\[ \lambda_k(t) :=\frac{1}{k} \log \EE_{y \sim Y_k} [ e^{2t k y}].\]
We will see in a few lines that the key quantity for computing the rate function is $\lambda(t):=\lim_{k\rightarrow \infty} \lambda_k(t)$. Using~\cref{conj:ext} for $t\geq 3$, we have,
\begin{equation*} 
\lambda(t):=\lim_{k \rightarrow \infty}\lambda_k(t)=\lim_{k \rightarrow \infty}\frac{1}{k} \log \EE_{M\sim \Gauss k}\left[\frac{|\Perm(M)|}{\sqrt{k!}}\right]^{2t}=\lim_{k \rightarrow \infty}\frac{1}{k} \log\left[ C_{k,t}\frac{k!^{t}t!^{2k}}{(kt)!}\right]=
2 \log(t!)-t \log(t).
\end{equation*}
In fact, we know more about $\lambda(t)$. Using exact $0$th, $2$nd, and $4$th moment calculations one can see that $\lambda(0)=\lambda(1)=\lambda(2)=0$. Given that $\lambda$ is convex function (the cumulant generating function is always convex), we conclude that $\lambda(t)=0$ for all $t\in [0,2]$. A similar convexity argument shows that $0\leq \lambda(t) \leq \lambda(3)=\log(4/3)$ for $2\leq t\leq 3$. To summarize:
\par
\begin{equation}\label{eq:expstir}
    \lambda(t)\left\{\begin{array}{ll}
        =0                    & 0\leq t\leq 2 \\
        \in [0,\log(4/3)]     &  2<t<3\\
        =2 \log(t!)-t\log (t) & 3\leq t
    \end{array} \right. .
\end{equation}
See~\cref{fig:lambdapl}.
\begin{figure}
    \centering
    \includegraphics[width=9cm]{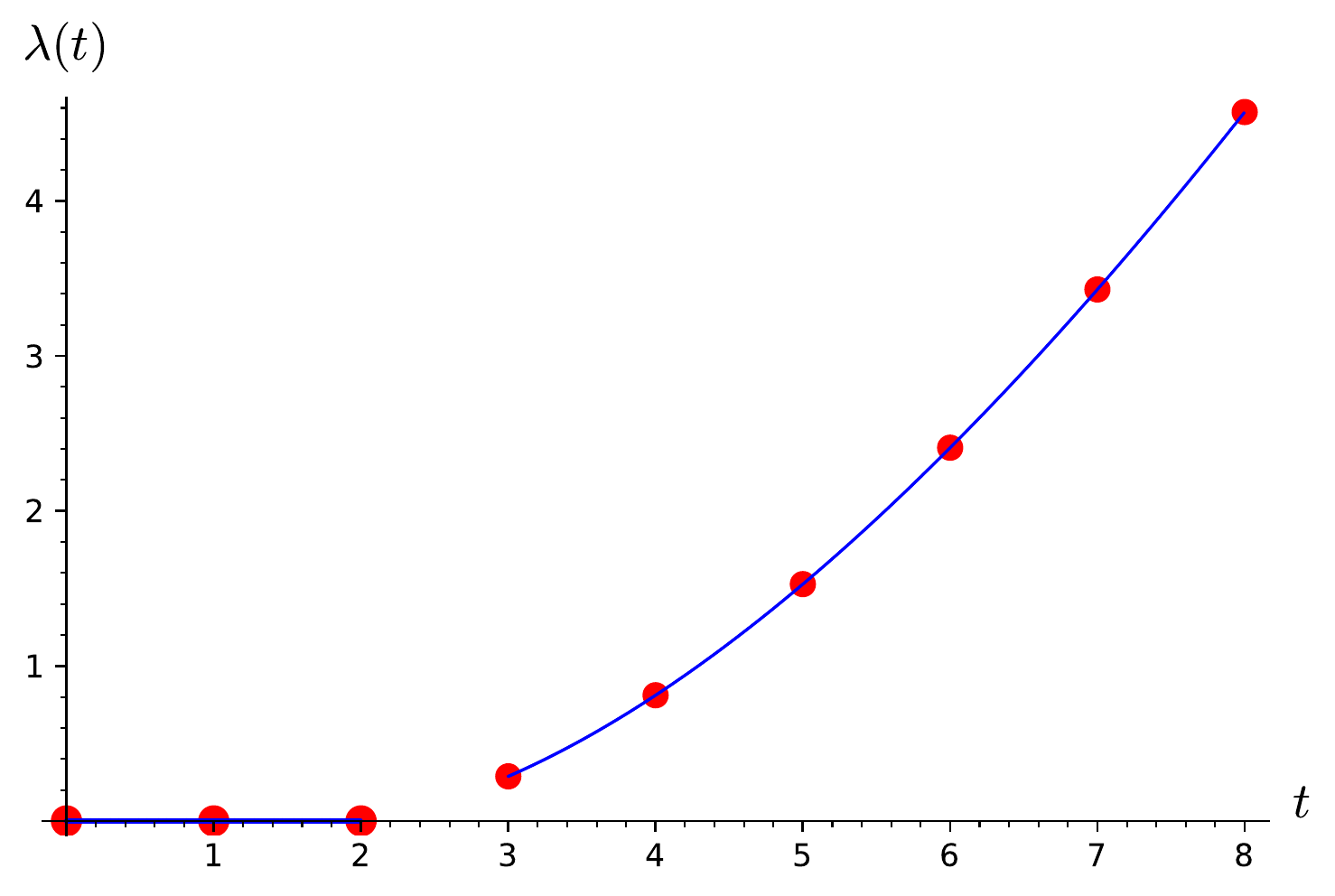}
    \caption{Plot of $\lambda(t)$. Red dots indicate the value of $\lambda(t)$ for integer $t$. The first three dots are consequences of known formulas for the $0$th, $2$nd, and $4$th moments of the permanent of random matrices. Other dots for $t\geq 3$ are consequences of~\cref{conj:maintext}. Blue line is the predicted form of $\lambda(t)$. Since $\lambda$ is a convex function, it should be identically zero for $t\in [0,2]$ to pass through the first three red dots. The piece of the blue curve for $t\geq 3$ follows from~\cref{conj:ext}. We do not have a prediction for $\lambda(t)$ when $2<t<3$ except that $0=\lambda(2)\leq\lambda(t)\leq \lambda(3)=\log(4/3)$.}
    \label{fig:lambdapl}
\end{figure}

Now, assuming that $k$ is very large, we have,
\[e^{k \lambda_k(t)}=\EE_{y \sim Y_k} \left[e^{2 t k y}\right]=\int{p_{Y_k}(y) \,e^{2t k y} \,dy} \approx  \int{e^{k(2t y- I(y))} \,dy}\approx e^{k \sup_{y}(2t y -I(y))}, \]
We conclude that 
\[ \lim _{k\rightarrow \infty} \lambda_k(t)=\lambda(t)=\sup_y (2t y -I(y)).\]
In other words, $\lambda(t)/2$ is the Legendre transform of $I(y)/2$. This is a well known and common technique for computing the rate function using the cumulant generating function (see~\cref{rem:gartner}). Inverting this transformation gives,
\begin{equation}\label{eq:legtr} I(y)=\sup_{t}\left(2 t y -\lambda(t) \right),
\end{equation}
Assume that the supremum is attained at $t=t^*$ ($t^*$ depends on $y$). Taking the first derivative when $t^*\geq 3$, we obtain
\begin{equation}\label{eq:derivativeleg}
2y=\lambda'(t^*)=2\psi(t^*+1)-\log(t^*)-1,
\end{equation}
where $\psi(t):=\frac{d}{d t}\log (\Gamma(t))$ is the ``digamma'' function. 
\par 
When $t^*$ is large, it is known that $\psi(t^*+1)\approx \log(t^*)$, therefore, $t^*\approx e^{2y+1}$. Using the Stirling's approximation, $\lambda(t^*)=2\log(t^*!)-t^*\log(t^*)\approx t^* \log(t^*)-2t^* \approx 2yt^*-e^{2y+1}$.
Therefore,
\[ I(y)=2yt^*-\lambda(t^*)=e^{2y+1}\quad \text{for large }y.\]

For generic $y$, we assume that $I(y)=e^{2y+1}\omega(y)$, where $\omega(y)$ is the correction of $I(y)$, deviating from $1$ for small $y$. Given a value of $y$, we can solve for $t^*$ in~\cref{eq:derivativeleg}, and compute $\omega(y)$. The resulting function is plotted in~\cref{fig:omegaplot}. We can numerically observe that when $y>0.21$, then $t^*>3$, therefore this calculation is only valid for $y>0.21$.
\par 
On the other hand, when $0\leq 6 y \leq \lambda(3)=\log(4/3)$, we know that the supremum of~\cref{eq:legtr} is attained at a value of $t\leq 3$, because for $t>3$ the value of $2ty - \lambda(t)$ is negative. Therefore, for this range of $y$, 
\[4y=2 t y -\lambda(t)\big |_{t=2} \leq I(y)=\sup_{0\leq t\leq 3}\left(2 t y -\lambda(t) \right)\leq \sup_{0\leq t\leq 3}(2ty)\leq 6y. \]
Strictly speaking, this result does not rely on any of our conjectures, rather, it relies on~\cref{thm:mainlowerbound} and the assumption of existence of a rate function. 
\par 

    \begin{figure}
        \centering
        \includegraphics[width=11cm]{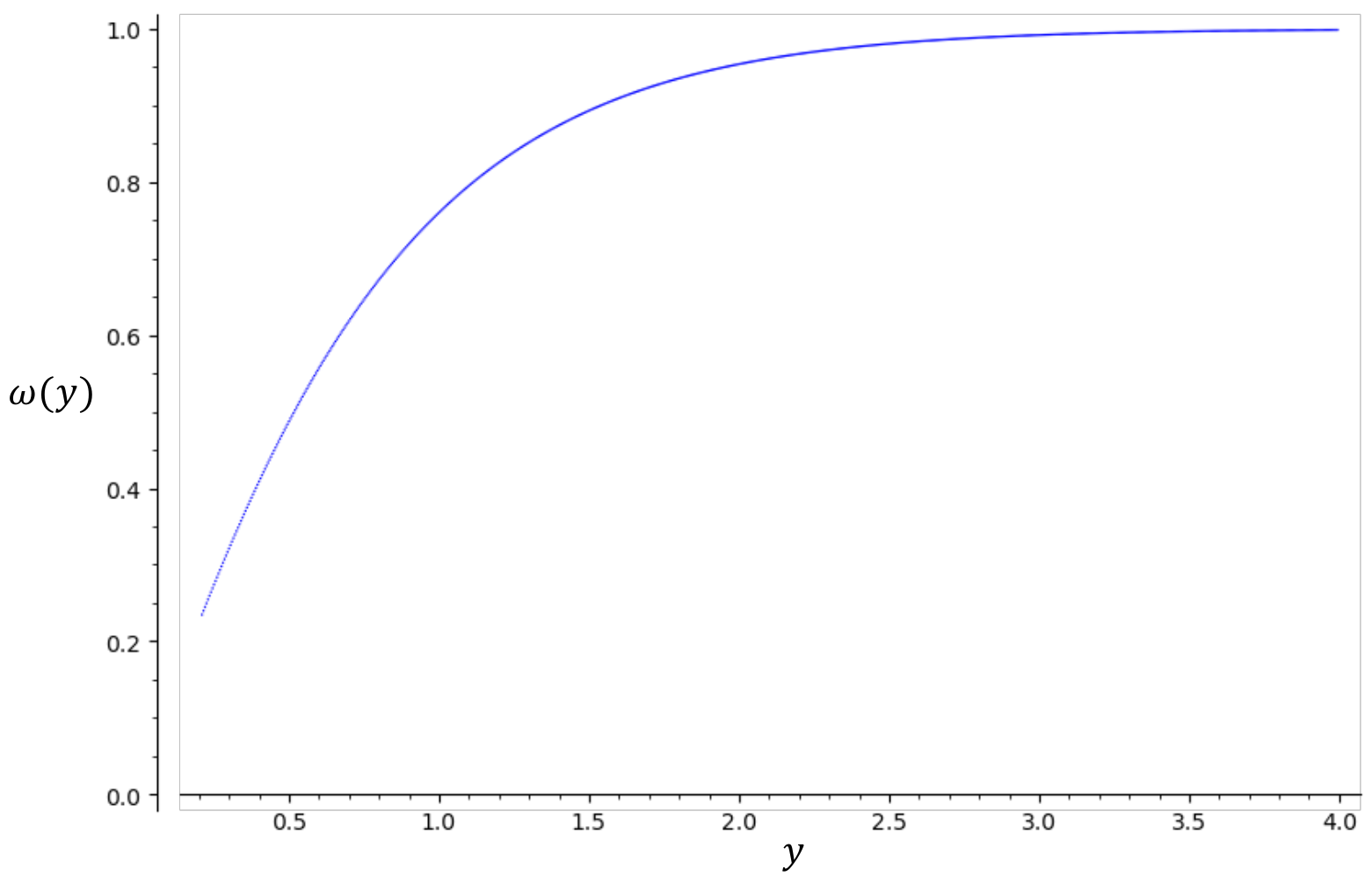}
        \caption{Plot of $\omega(y)$ defined in~\cref{thm:tail-dist}. When $y$ is small, $\omega(y)$ gets closer to zero. This is expected since it is believed that the most of the mass of the random permanent probability distribution is at around $\sqrt{k!}$, which translates to $y=0$ in~\cref{eq:dfnY}. }
        \label{fig:omegaplot}
    \end{figure}
\par 
To summarize, we have the following theorem:
\begin{thm}\label{thm:tail-dist}
Assuming~\cref{conj:ext} and the existence of a rate function $I(y)$, the probability density function of the normalized log-permanent, $p_{Y_k}$, has the form of,
\[p_{Y_k}(y) = e^{-k I(y)+o(k)}, \quad\text{with }\quad I(y)=e^{2y+1}\omega(y) \quad\text{ for }y>0.21,\]
where $\omega(y)\approx 1$ for large $y$, and is defined as
\[\omega(y) = 2^{-2y-1} \sup_{t}\left(2yt-(2\log(t!)-t \log(t)\right). \]
Moreover, only assuming the existence of the rate function and no extra assumption (such as~\cref{conj:ext}), we have,
\[ 4y\leq I(y)\leq 6y \quad\text{for}\quad 0\leq y\leq \frac16\log(4/3)\approx 0.048.\]
\end{thm}

\begin{rem} \label{rem:gartner}
The argument presented in this section was rather imprecise. However (except for~\cref{conj:ext} which is the underlying assumption of this section), the rest can be made precise. The G\"artner-Ellis theorem (see~\cite{touchette2011basic,dembod1996large}) states that when $\lambda(\alpha)=\lim_{k\rightarrow \infty} \lambda_k(\alpha)$ exists and is differentiable, the distribution satisfies the large deviation principle, i.e.,  $\lim_{k\rightarrow \infty} \frac{1}{k}P_{Y_k} (y)\rightarrow -I(y)$ for some rate function $I(y)$. Furthermore, the rate function is given by the Legendre transform of $\lambda(\alpha)$. This means that under mild assumptions, the assumption of existence of the rate function can be dropped in~\cref{thm:tail-dist}. 
\end{rem}

\subsubsection{Comparison to the case of determinants}
In this section, we run a similar analysis for the case of determinants, and obtain a quantitatively different behavior.
\par 
Similar to the previous section, define,
\[Z_k:= \frac{1}{k} \log \,\left[\frac{|\det(M)|}{\sqrt{k!}}\right]\quad \text{for}\quad M\sim \Gauss{k} . \]
If we repeat the argument of the previous section, and use the formula for the moments of determinant of random Gaussian matrices (see~\cref{eq:detgg}), we can easily see that the scaled cumulant generating function $\lim_{k\rightarrow \infty} \frac1k \log \EE_{z\sim Z} e^{2ktz}$ vanishes. Consequently, the rate function for determinant diverges as $k\rightarrow \infty$. In other words, the tail of the distribution decays faster than $e^{-ck}$ for some $c$.  
\par 
This shows that in the limit of large $k$, the tail of the probability density function of log-permanent of random matrices (when normalized by $1/k$) is longer and has more mass than the tail of the probability density function of log-determinants (when normalized by $1/k$). Note that it is not in priory clear if the same result holds if we normalize these distributions with a factor other than $1/k$.
\par 
We can normalize the log-determinant distribution with a different normalization factor. If we define 
\[Z_k':= \frac{1}{\log(k)} \log \,\left[\frac{|\det(M)|}{\sqrt{k!}}\right]\quad \text{for}\quad M\sim \Gauss{k}, \]
then the scaled cumulant generating function does not vanish as $k \rightarrow \infty$:
\begin{multline}
\lambda(t) = \lim_{k\rightarrow\infty} \frac{1}{\log(k)} \log\left(\EE_{z\sim Z_k'} [e^{2\log(k) t z}]\right)=\\
\lim_{k\rightarrow\infty} \frac{1}{\log(k)}\log\EE_{M\sim \Gauss k} \left(\frac{|\det(M)|}{\sqrt{k!}}\right)^{2t}=\lim_{k\rightarrow\infty} \frac{1}{\log(k)} \log \left(\frac{\prod_{i=1}^{t}\prod_{j=1}^{k} (i+j-1)}{k!^{t}}\right)= \frac{t(t-1)}2
\end{multline}
 
Any other normalization factor will lead to a vanishing or diverging $\lambda$. Again, the rate function for of the determinant is given by the Legendre transform of $\lambda (t)$:
\[p_{Z_k'}(z):=e^{-\log(k) I(z)+o(\log(k))},\quad\text{where} \quad I(z)/2=\sup_{t\geq 0}(t z - \lambda(t)/2). \]
This Legendre transform is easy to compute, and gives the result of $I(z)=2(z+1/4)^2$. Therefore, we have,
\begin{equation}\label{eq:detld}
P_{Z_k'}(z)\approx e^{-2\log(k)(z+1/4)^2+o(\log(z))},\quad \text{for  }z>0. \end{equation}
This mimics the law of random determinants~\cite{nguyen2014random,girko2012theory}. 

\section{Permanents of minors of random unitary matrices}\label{sec:permsub}
In this section, we study the moments of permanents of minors of Haar random unitary matrices:
\[\EE_{M \sim \Usub d k} |\Perm (M)|^{2t}.\]
Using~\cref{{eq:main_avg_formula}} this can be expanded as 
\begin{align*}
\EE_{M\sim \Usub d k} \left| \Perm(M) \right|^{2t} = \sum_{\lambda \vdash kt, l(\lambda)\leq \min(k,t)} \frac{1}{\WD_\lambda(d)} \left(\frac{f^\lambda}{(kt)!}\right)^2\tr \left[\rho_\lambda(RCRC) \right].
\end{align*}
In~\cref{sec:permiid} we have extensively studied the terms $\tr \left[\rho_\lambda(RCRC) \right]$. We will observe that unlike the case of Gaussian i.i.d. matrices, many Young diagrams can contribute to the above sum. Therefore, an analog of~\cref{conj:maintext} cannot be easily stated in this case. Nevertheless, in~\cref{sec:lbhj}, we prove a different lower bound for the moments of permanents, and argue that it should be almost tight for the specific case of $k=d$ (i.e., the moments of the permanents of random unitary matrices) as a consequence of the Hunter-Jones conjecture.
\par 
First, we study the case of $t=1$ and $t=2$:
\begin{itemize}
    \item $t=1$. In this case, the only contribution comes from $\lambda = (k)$. Therefore, 
    \begin{equation} \label{eq:hmt1}
    \EE_{M\sim \Usub d k} \left| \Perm(M) \right|^{2} =  \frac{1}{\WD_{(k)}(d)} \left(\frac{f^{(k)}}{k!}\right)^2\tr \left[\rho_{(k)}(RCRC) \right]=\frac{1}{\WD_{(k)}(d)}=\frac{1}{\binom{k+d-1}{k}}. 
    \end{equation}
    \item $t=2$. This case is slightly more complicated, but we have an explicit formula for $\tr \left[\rho_{(kt)}(RCRC) \right]$ from~\cref{rem:t-2}, which leads to the following formula:
    \begin{multline}\label{eq:hht2}
    \EE_{M\sim \Usub d k} \left| \Perm(M) \right|^{4} = \sum_{a=0}^k \frac{1}{\WD_{(2k-a,a)}(d)} \left(\frac{f^{(2k-a,a)}}{(2k)!}\right)^2\tr \left[\rho_{(2k-a,a)}(RCRC) \right]=\\
    (d-1)!(d-2)!\sum_{a=0}^{\lfloor k/2 \rfloor} 4^{k-2a}
    \frac{2k-4a+1}{2k-2a+1}\binom{2a}{a}\binom{2k-2a}{k-a}^{-1}\frac{(k!)^2}{(2a+d-2)!(2k-2a+d-1)!}.
    \end{multline}
    This formula is very complicated, but it is possible to read its asymptotic growth in some cases (see~\cref{conj:hj} and the paragraphs that follows).
    \item $t\geq3$. We have not been able to derive an explicit formula for this case. However, some lower bounds can be derived following the logic of~\cref{sec:lbmgc} (see~\cref{sec:lbhj} for a different type of lower bound). For the specific cases of $t=3$ and $k=3,4$, we can use the explicit results of~\cref{eq:expres} combined with~\cref{eq:generalized-expansion-formula} to find expressions for the moments of permanents (See~\cref{conj:hj} and what follows). 
\end{itemize}
\subsection{Lower bounds}\label{sec:lbhj}
Now, we proceed to prove a lower bound for the moments of permanents of minors of Haar random unitary matrices. Let $t\leq k$, and start by the expansion formula~\cref{eq:generalized-expansion-formula}, and define $p_\lambda $ and $q_\lambda$ as follows:
\begin{equation}\label{eq:defpq}
    p_\lambda := \frac{\Pl_\lambda^{t,k}\WD_\lambda(d)}{\binom{\binom{d+k-1}{k}+t-1}{t}}, \quad\text{and},\quad q_\lambda:= \frac{f^\lambda}{(kt)!}\tr\left[\rho_\lambda(RC)\right],\quad \text{ for }\lambda\vdash kt, l(\lambda)\leq t.
\end{equation} 
\begin{lem}
Both $p_\lambda$ and $q_\lambda$ defined in~\cref{eq:defpq} are probability distributions.
\end{lem}
\begin{proof}
It is obvious that both quantities are positive, therefore, it only remains to show that they sum to $1$. Moreover, based on the argument at the end of~\cref{sec:expproof}, it is easy to see that both $q_\lambda$ and $p_\lambda$ are zero if $l(\lambda)\geq t$. Therefore, we can implicitly remove the constraint $l(\lambda)\leq t$.
\par
We start by proving that $p_\lambda$ is a probability distribution. Using the decomposition of~\cref{eq:pleth}, it is straightforward to see that 
\[\dim \Sym_t(\Sym_k(\CC^d)) = \sum_{\lambda\vdash kt} \Pl^{t,k}_\lambda \dim \left(V_\lambda^{U(d)}\right)= \sum_{\lambda\vdash kt} \Pl^{t,k}_\lambda \WD_\lambda(d).\]
Substituting $\dim \Sym_t(\Sym_k(\CC^d)) = \binom{\binom{d+k-1}{k}+t-1}{t}$ gives the desired result for $p_\lambda$.
\par
To prove the similar fact for $q_\lambda$, note that using the decomposition of the regular representation we have that 
\[\sum_\lambda q_\lambda =  \sum_\lambda \frac{f^\lambda}{(kt)!}\tr [\rho_\lambda(RC) ]=\frac{\tr(RC)}{(kt)!} = \sum_{r\in R, c\in C} \frac{\tr(rc)}{(kt)!},\]
where $r$ and $c$ are in the regular representation. Since $rc=e$ if and only if $r=c=e$, there is only one term that contributes to the above sum. This shows that $\sum_{r\in R, c\in C} \frac{\tr(rc)}{(kt)!} =  \frac{\tr(e)}{(kt)!}=1$.
\end{proof}
Now, it is easy to see that 
\begin{multline*}
\EE_{M\sim \Usub d k} \left| \Perm(M) \right|^{2t} = \sum_{\lambda \vdash kt, l(\lambda)\leq \min(k,t)} \frac{1}{\WD_\lambda(d)} \left(\frac{f^\lambda}{(kt)!}\right)^2\tr \left[\rho_\lambda(RCRC) \right]\geq  \\
\sum_{\lambda \vdash kt, l(\lambda)\leq t} \frac{1}{\Pl_\lambda^{t,k}\WD_\lambda(d)} \left(\frac{f^\lambda \tr \left[\rho_\lambda(RC) \right]}{(kt)!}\right)^2 = \binom{\binom{d+k-1}{k}+t-1}{t}^{-1}\sum_{\lambda \vdash kt, l(\lambda)\leq t}\frac{q_\lambda^2}{p_\lambda}\geq\\ \binom{\binom{d+k-1}{k}+t-1}{t}^{-1},
\end{multline*}
where the first inequality is proved in~\cref{thm:pleth_ineq}, and the second one is the Cauchy-Schwarz inequality  $(\sum_\lambda q_\lambda^2/p_\lambda)(\sum_\lambda p_\lambda) \geq (\sum_\lambda p_\lambda)^2 = 1$. Therefore, we proved the following theorem:
\begin{thm}[Lower bound for moments of permanents of minors of Haar random unitary matrices]\label{thm:minorbound}
Let $k,t$ be integers, $k\leq d$, and $a:= \min(k,t), b:=\max(k,t)$. Then, we have the following inequality for the moments of permanents of minors of Haar random unitary matrices:
\begin{equation}\label{eq:minorbound}
    \EE_{M\sim \Usub d k} \left| \Perm(M) \right|^{2t} \geq \binom{\binom{d+b-1}{b}+a-1}{a}^{-1}.
\end{equation}
\end{thm}

\begin{conj}[Hunter-Jones conjecture]\label{conj:hj}
Nick Hunter-Jones conjectured that:
\begin{equation}\label{eq:hjj}
\EE_{M\sim U(d)} \left| \Perm(M) \right|^{2t} \approx \frac{t!}{\binom{2d-1}{d}^t}.
\end{equation}
\end{conj}
Since 
\[\frac{t!}{\binom{2d-1}{d}^t}\approx \binom{\binom{2d-1}{d}+t-1}{t}^{-1},\]
this conjecture means that~\cref{eq:minorbound} is close to being tight.
\par
Note that the Hunter-Jones conjecture can also inform the asymptotics of $\tr[\rho_\lambda(RCRC)]$, as the tightness of Cauchy-Schwarz inequality in the proof of~\cref{thm:minorbound} indicates that $p_\lambda \sim q_\lambda$ when $q_\lambda$ is relatively large. 
\par 
Now, we provide evidence for~\cref{conj:hj}. When $t=1$,~\cref{eq:hmt1} shows that~\cref{eq:hjj} is in fact an equality. For $t=2$, we can use~\cref{eq:hht2} and see that the~\cref{conj:hj} holds with great accuracy for large $d$. See~\cref{fig:relerr} for a comparison of our conjecture and exact results. When $t=3$, we can exactly compute the unitary permanent moments for $d=k=3$ and $d=k=4$ using the results of~\cref{eq:expres} and~\cref{eq:generalized-expansion-formula}. We obtain the following numbers:
\begin{align}
    \EE_{M\sim U(3)} \left| \Perm(M) \right|^{6} &= 323/57750,\\
    \EE_{M\sim U(4)} \left| \Perm(M) \right|^{6} &= 578047/4138509375.
\end{align}
The corresponding relative errors are:
\begin{align*} 
\left(\EE_{M\sim U(3)} \left| \Perm(M) \right|^{6}- \frac{3!}{\binom{2\times 3-1}{3}^3}\right)/\EE_{M\sim U(3)} \left| \Perm(M) \right|^{6} \approx -0.072,\\
\left(\EE_{M\sim U(4)} \left| \Perm(M) \right|^{6}- \frac{3!}{\binom{2\times 4-1}{4}^3}\right)/\EE_{M\sim U(4)} \left| \Perm(M) \right|^{6} \approx -0.0019,
\end{align*}
which again support~\cref{conj:hj}.
\begin{figure}
    \centering
    \includegraphics[width=10cm]{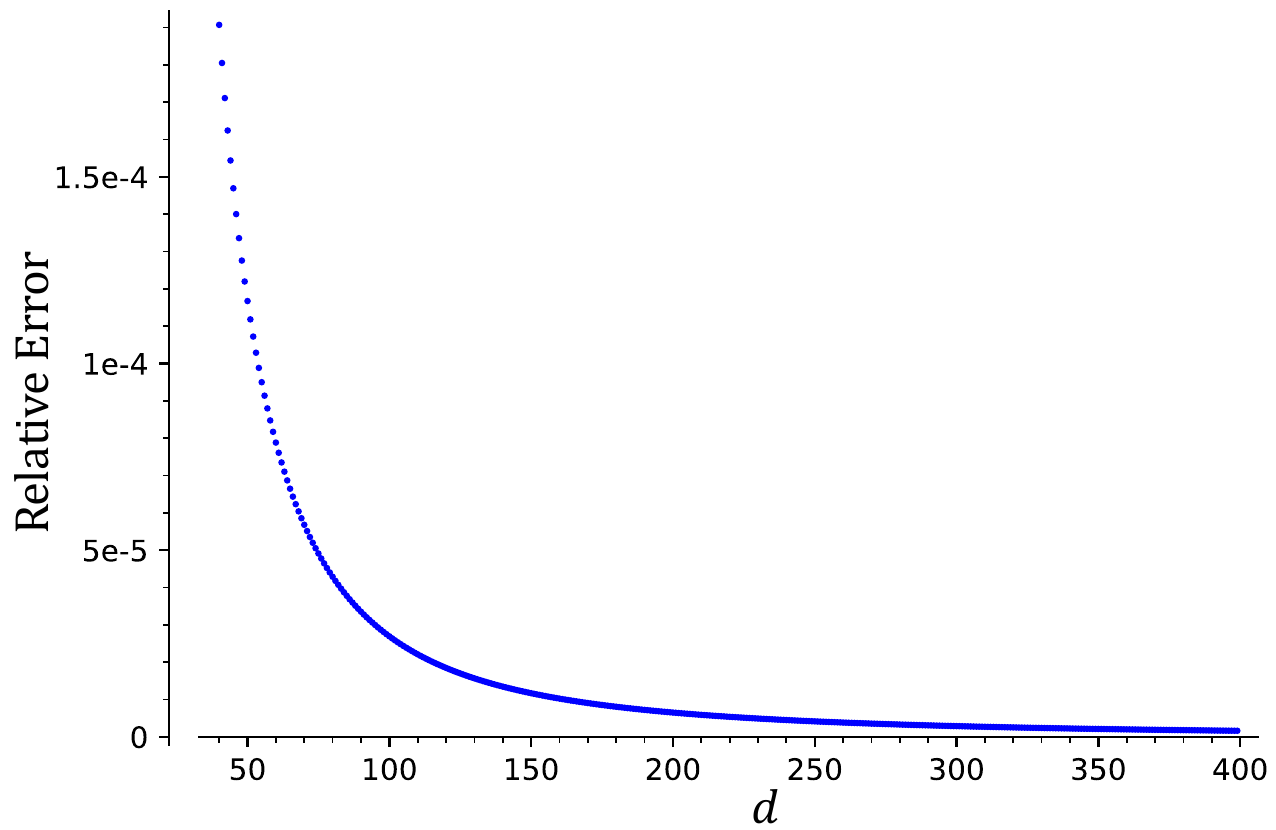}
    \caption{Relative error of~\cref{conj:hj}. Vertical axis is the relative error, $1-\frac{\EE_{M\sim \Usub d d} \left| \Perm(M) \right|^{4}}{{2!}/{\binom{2d-1}{d}^2}}$, and the horizontal axis is $d$.   }
    \label{fig:relerr}
\end{figure}

\section{Moments of the determinant}
In this section, we exactly compute the moments of determinants of random matrices and prove that they constitute a lower bound for the moments of permanents. This lower bound gives better estimate for some small permanent moments, but quickly fails as one considers higher moments. In such cases, we should use the main bounds that we discussed above.
\par
Recall the definition of the determinant:
\[\det(M) = \sum_{\pi \in S_t} \prod_{i=1}^k \text{sgn}(\pi) M_{i \pi (i)}. \]
Skipping a few lines of algebra and repeating the arguments of~\cref{sec:intro_1}, we get
\begin{multline}
\EE_{M\sim \Gauss k} |\!\det(M)|^{2 t} = \frac1{(kt)!} \sum_{r_1 ,r_2 \in R, c_1,c_2 \in C}  \text{sign}(c_1) \text{sign}(c_2) \tr (r_1 c_1 r_2 c_2)\leq  \\ \frac1{(kt)!} \sum_{r_1 ,r_2 \in R, c_1,c_2 \in C}  \tr (r_1 c_1 r_2 c_2) =\EE |\Perm(M)|^{2 t}. 
\end{multline}
Let us indicate the rectangular Young diagram of $k$ rows and $t$ columns by ``$\rectangle$''. Then, $\sum_{r \in R, c \in C}  \text{sign}(c) cr $ is nothing but the Young symmetrizer of the Young diagram $\rectangle$, i.e., $c_{\rectangle}$. We have, 
\begin{equation}\label{eq:detgg}
\EE_{M \sim \Gauss k} |\!\det(M)|^{2 t}  = \frac{1}{(kt)!} \tr(c_{\rectangle}c_{\rectangle}) =\frac{1}{(kt)!}\hook(\!\rectangle) \tr(c_{\rectangle}) = \hook(\rectangle) = \prod_{i=1 \cdots k, j = 1 \cdots t} (i+j-1),
\end{equation}
where we used~\cref{eq:cchook}, and the fact that $\tr c_\lambda$ is always equal to $(kt)!$ as the identity element appears only once in the expansion on the Young symmetrizer. 
This lower bound gives the correct value of permanent moments when $\min(k,t)\leq 2$, but as mentioned earlier, quickly fails for the other values.
\subsection{Exact moments of determinants of Haar random unitary minors}\label{sec:detmom}
The distribution of the determinant of random Gaussian matrices has been extensively studied before, and its moments and the limiting distributions are well known. In this section, we consider the (absolute value of) determinant of truncated random unitaries matrices, and for the first time, prove a simple formula for its moments. This section is a quick digression from the rest of the paper as it is not immediately relevant to the problem of permanents.
\par 
As we mentioned before, we indicate the distribution of $k \times k$ minors of Haar random unitary $d\times d$ matrices by $\Usub d k$. This distribution should take radically different forms as one increases the value of $k$ from $1$ to $d$:
\begin{itemize}
    \item If $k\ll d$, as pointed out in~\cref{rem:gauss_usub}, the distribution $\Usub d k$ is very similar to the distribution of (sub-normalized) random Gaussian matrices. Therefore, the distribution of determinants should be given by the \emph{logarithmic law}~\cite{girko2012theory,nguyen2014random}.
    \item When $k=d$, the absolute value of the determinant is always equal to $1$. Therefore, the distribution is a delta function at $1$ 
\end{itemize}
It is interesting to study how the distribution of determinants of random matrices evolve from the logarithmic law to a delta function, and how the moments of determinant behave in this transition.
\par
The calculations leading to exact moments of determinants are parallels of the calculations that led to~\cref{eq:generalized-expansion-formula} and will be omitted here. The analog of~\cref{eq:generalized-expansion-formula} for the case of determinants is 
\begin{align*}
\EE_{U \sim \Usub d k} \left| \det \left(M\right) \right|^{2t} = \sum_{\lambda \vdash kt, l(\lambda)\leq \min(k,t)} \frac{1}{\WD_\lambda(d)} \left(\frac{f^\lambda}{(kt)!}\right)^2\tr \left[\rho_\lambda(c_{\rectangle}c_{\rectangle}) \right],
\end{align*}
However, the formula simplifies significantly in this case. Note that the support of $c_{\rectangle}$ lies only in the irrep with $\lambda = \rectangle$. Hence, $\lambda={\rectangle}$ is the only term that survives in the sum and we obtain the following relation:
\begin{align*}
\EE_{U \sim \Usub{d}{k}} \left| \det \left(M\right) \right|^{2t} = \frac{1}{\WD_{\rectangle}(d)} \left(\frac{f^{\rectangle}}{(kt)!}\right)^2\tr \left[\rho_{\rectangle}(c_{\rectangle}c_{\rectangle}) \right].
\end{align*}
Using $c_{\rectangle}c_{\rectangle}=\hook({\rectangle})c_{\rectangle}$, and $\rho_{\rectangle}(c_{\rectangle})= \hook({\rectangle})$, this simplifies to
\begin{align}\label{eq:detform}
\EE_{U \sim \Usub{d}{k}} \left| \det \left(M\right) \right|^{2t} = \frac{1}{\WD_{\rectangle}(d)} \left(\frac{1}{\hook({\rectangle})}\right)^2 \hook({\rectangle})^2 =\frac{1}{\WD_{\rectangle}(d)}=\prod_{i=1,j=1}^{k,t}\frac{i+j-1}{(d-k)+i+j-1}.
\end{align}
This relation clearly shows that when $d=k$, all of the moments are equal to $1$, indicating that the distribution is concentrated at the value of $1$. Moreover, as $d\rightarrow \infty$, the normalized moments $\EE_{U \sim \Usub{d}{k}} \left| \det \left(\sqrt d M\right) \right|^{2t} $ reproduce the result of~\cref{eq:detgg}. The exact formula~\cref{eq:detform} is complement to the results of~\cite{zyczkowski2000truncations} which studies the eigenvalue distribution of truncated random unitary matrices.

\section{Summary and open questions}
In this paper, we systematically studied the moments of permanents of random matrices. We analyzed the both cases of the random Gaussian i.i.d. matrices as well as submatrices of random unitary matrices. We proved strong lower bounds for all moments of permanents for the cases of random Gaussian and random unitary matrices, and argued that the lower bound are close the true value of the permanent moments. With our results and quantitative conjectures, we derived concentration results for the log-permanent of random Gaussian matrices. This work leads to new questions, which we discuss below:
\begin{itemize}
    \item The first obvious open question is proving/disproving our moment growth conjecture and our moment upper bound conjecture for the permanent of random Gaussian matrices (\cref{conj:maintext}). Is it possible to bound the individual terms in the expansion~\cref{eq:generalized-expansion-formula}? or is it possible to show that the matrix elements of the weak magic squares defined at end of~\cref{sec:conjsection} are nearly independent?
    \item The second question is proving the Hunter-Jones conjecture for the permanent of random unitary matrices (\cref{conj:hj}). Similar to the Gaussian case, we provided evidence for this conjecture and proved a closely matching lower bound. The proof of this lower bound is different and more subtle than the Gaussian case, but, the exact distribution of random unitary permanents could be easier derive. In particular, this distribution cannot have a large tail because the permanent of a unitary matrix is always smaller than $1$.
    \item Using basic large deviation techniques, we computed the rate function of the distribution of $\frac 1 k \log\left[|\Perm(M)|/\sqrt{k!}\right]$ for $M\sim \Gauss k$, in the limit of $k\rightarrow \infty$. But, the rate function is not differentiable at $0$, which suggests that one should look at the less harshly normalized quantities, such as $\frac 1 {\log(k)} \log\left[|\Perm(M)|/\sqrt{k!}\right]$ for a more detailed computation of the log-permanent distribution. Sadly, such calculation seems to require the knowledge of (non-integer) moments smaller than $2$, which is beyond our work. 
    \item There could be a roadblock for computing the full permanent distribution function from the integer moment information. To argue in details, recall the Stieltjes moment problem:
\begin{dfn}[Stieltjes moment problem]\label{dfn:stm}
Let $m_0,m_1,\cdots$ be a list of positive real numbers. What are the necessary and sufficient conditions such that the probability distribution $p(x)$ exists where
\[ \int_0^\infty p(x) x^t = m_t.\]
Moreover, when the probability distribution exists, is it unique?
\end{dfn}
We have computed the moments from the permanent probability distribution, therefore, we are not concerned about the existence part of the problem. However, if the solution to the moment problem is not unique, then, it is not possible to read-off the permanent distribution purely from the moment data. 
\par 
We do not know if the moment problem of random permanents is non-unique, but given the fast growth of permanent moments we find it unlikely to be unique.
\item We have computed a large number of moments of $3\times 3$ and $4\times 4$ matrices. Is it possible to directly guess a general formula for such moments? An exact analytical formula can fully determine the permanent distribution, and resolve the permanent anti-concentration conjecture. One may also modify the matrix ensemble with the hope of finding simple moment formulas. For instance, the moments of permanents of $3\times 3$ matrices where the first row elements are random phases and the other matrix elements are random complex Gaussians has a simple analytical formula. Is it possible to find such hybrid ensembles where the permanent moments are easily computable for all matrix sizes?   
\end{itemize}

\par
\ssection{Acknowledgments}
The author would like to thank 
Scott Aaronson, 
Alex Arkhipov,
David Gross, 
Nick Hunter-Jones, 
Richard Keung,
Saeed Mehraban, 
Terence Tao, 
Van Vu, and
Michael Walter 
for illuminating discussions. The author is supported by the Walter Burke Institute for Theoretical Physics and IQIM at Caltech.

 \clearpage

\renewcommand\refname{References}
\printbibliography
\appendix 
\clearpage
\section{Explicit moment calculations for \texorpdfstring{$t=3,k=3$}{}, and \texorpdfstring{$t=3,k=4$}{}}\label{eq:expres}
In this section, we report the values of $\tr[\rho_\lambda(RC)]$ for the cases of $t=3$, and $k=3,4$. One can easily see that (e.g. using the sagemath software) all of the plethysm coefficients for these Young diagrams are $0$ or $1$, therefore, $\tr_\lambda(RCRC)=(\tr_\lambda(RC))^2$.
For $t=3,k=3$, there are $12$ Young diagrams with the total number of $9$ boxes, and depth $\leq 3$. Among these, only $5$ of them are non-zero, which are reported below:
\newcommand{\centered}[1]{\begin{tabular}{@{}l@{}} #1 \end{tabular}}

\begin{center}
    \begin{tabular}{|c|c|c|c|}
    \hline
    $\lambda$ & Shape of $\lambda$ & $\tr[\rho_\lambda(RC)]$& General formula from~\cref{sec:RC_calc}\\
    \hline
    \hline
     \centered{ $ (9)$ } & \raisebox{-0.45\totalheight}{ \includegraphics[height=1cm]{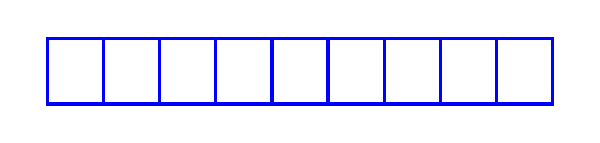}}& $2^63^6 $ &$k!^tt!^k$\\\hline
    $(7,2)$  &\raisebox{-0.45\totalheight}{ \includegraphics[height=1.5cm]{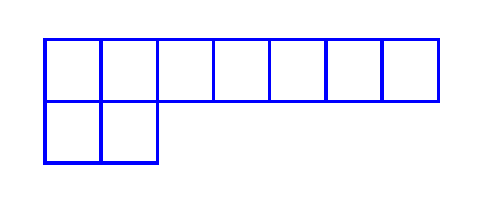}}& $2^63^4 $ &$k!^tt!^k/(td)$\\\hline
    $(6,3)$  & \raisebox{-0.45\totalheight}{ \includegraphics[height=1.5cm]{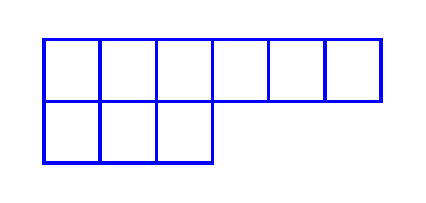}}& $2^83^2 $ &$k!^tt!^k\times 4/(t^2d^2)$\\\hline
    $(5,2,2)$&\raisebox{-0.45\totalheight}{  \includegraphics[height=2cm]{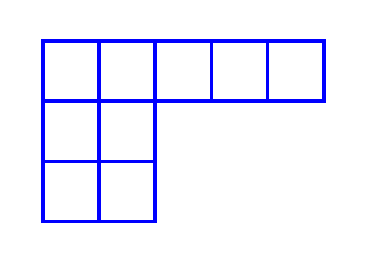}}& $2^43^2 $ & $k!^tt!^k\times \frac{( k-2) ( t-2)}{(k-1) k^2 ( t-1) t^2}$           \\\hline
    $(4,4,1)$&\raisebox{-0.45\totalheight}{  \includegraphics[height=2cm]{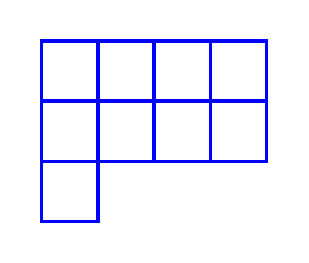}}& $2^63^2 $ & \\
    \hline
    \end{tabular}
\end{center}
For $t=3$ and $k=4$, only $9$ Young diagrams contribute:
\begin{center}
    \begin{tabular}{|c|c|c|c|}
    \hline
    $\lambda$ & Shape of $\lambda$ & $\tr[\rho_\lambda(RC)]$& General formula from~\cref{sec:RC_calc}\\
    \hline
    \hline
     \centered{ $(12)$ } & \raisebox{-0.45\totalheight}{ \includegraphics[height=1cm]{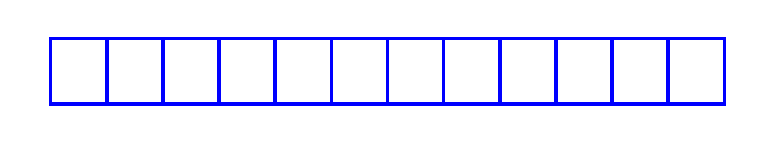}}& $2^{13}  3^7 $ &$k!^tt!^k$\\\hline
    $(10,2)$             &\raisebox{-0.45\totalheight}{ \includegraphics[height=1.5cm]{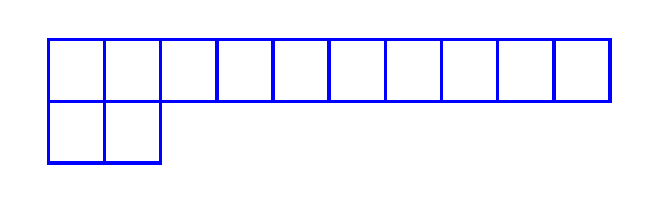}}& $2^{11}  3^6$ &$k!^tt!^k/(td)$\\\hline
    $(9,3)$              & \raisebox{-0.45\totalheight}{ \includegraphics[height=1.5cm]{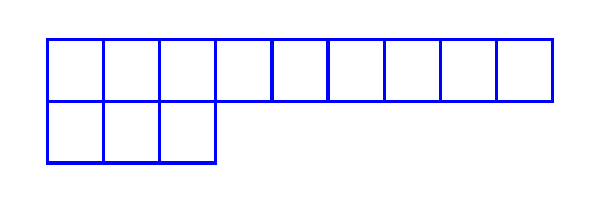}}& $2^{11}  3^5 $ &$k!^tt!^k\times 4/(t^2d^2)$\\\hline
    $(8,4)$              &\raisebox{-0.45\totalheight}{  \includegraphics[height=1.5cm]{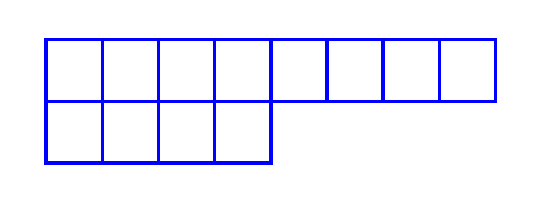}}& $2^{12}  3^4$ & $k!^tt!^k\times \frac29\frac{k}{k^3 - k^2}$                \\\hline
    $(6,6)$&\raisebox{-0.45\totalheight}{  \includegraphics[height=1.5cm]{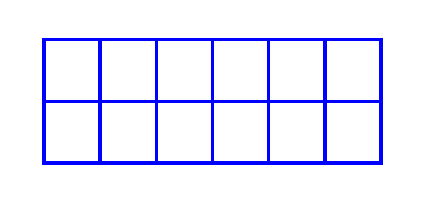}}& $2^{12}  3^3$ & \\
    \hline
    $(8, 2, 2)$&\raisebox{-0.45\totalheight}{  \includegraphics[height=2cm]{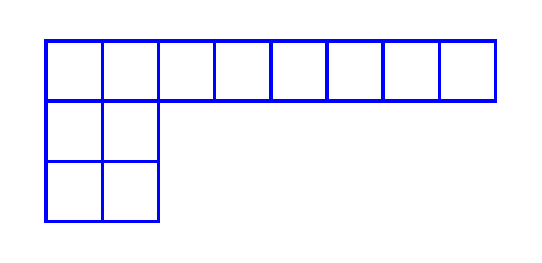}}& $2^9  3^4 $ & $k!^tt!^k\times \frac{( k-2) ( t-2)}{(k-1) k^2 ( t-1) t^2}$ \\
    \hline
    $(7, 4, 1)$&\raisebox{-0.45\totalheight}{  \includegraphics[height=2cm]{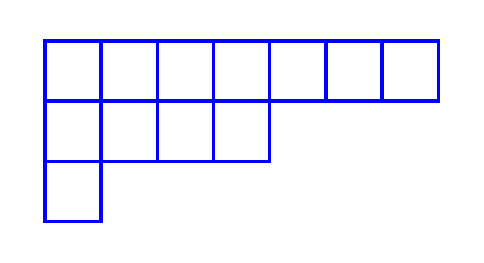}}& $2^{10}  3^4 $ & \\
    \hline
    $ (6, 4, 2)$&\raisebox{-0.45\totalheight}{  \includegraphics[height=2cm]{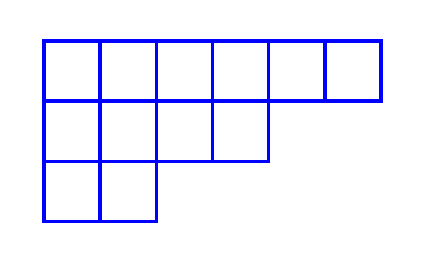}}& $2^{11}  3^2$ & \\
    \hline
    $(4, 4, 4)$&\raisebox{-0.45\totalheight}{  \includegraphics[height=2cm]{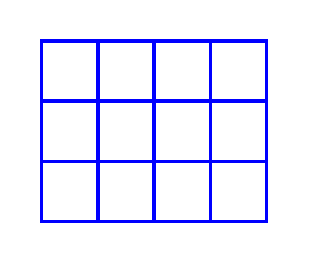}}& $2^9  3^3 $ &\\
    \hline
    \end{tabular}
\end{center}
\section{Table of the permanent moments for \texorpdfstring{$k=3$}{}}\label{sec:algt3}
In this section, we will report the result of the sixth moment of permanents of $k\times k$ matrices. The numbers are generically too large, so we have tried to divide them by a product of factorials to fit them into a table. Unfortunately, we could not find or guess an explicit formula for the moments of permanents.
\par
We see that the permanent moments stay above $1.625 \times k!^{2t}t!^{2k}/(kt)!$ and below $2 \times k!^{2t}t!^{2k}/(kt)!$ when $t\geq 4$, as expected.  
{\tiny
\begin{center}
\begin{tabular}{|c|c|c|}\hline
    $t$ & $\EE \left| \Perm (M_{k=3})\right|^{2t} / t!(t-1)!\lfloor\frac t3\rfloor! \lfloor\frac t4\rfloor! \lfloor\frac t5\rfloor! \lfloor\frac t7\rfloor!$ &   $\EE \left| \Perm (M)\right|^{2t} / (k!^{2t}t!^{2k}/((kt)!))$\\\hline
    $ 1 $ & $ 6 $ & $ 1.00000000000000 $\\
$ 2 $ & $ 72 $ & $ 1.25000000000000 $\\
$ 3 $ & $ 732 $ & $ 1.46433470507545 $\\
$ 4 $ & $ 7584 $ & $ 1.62974751371742 $\\
$ 5 $ & $ 84000 $ & $ 1.75215049154092 $\\
$ 6 $ & $ 504384 $ & $ 1.84007733807704 $\\
$ 7 $ & $ 6586272 $ & $ 1.90141045844492 $\\
$ 8 $ & $ 46690560 $ & $ 1.94268621025552 $\\
$ 9 $ & $ 238845888 $ & $ 1.96908588019468 $\\
$ 10 $ & $ 1976434560 $ & $ 1.98461658724633 $\\
$ 11 $ & $ 35140412736 $ & $ 1.99232746105098 $\\
$ 12 $ & $ 55727266176 $ & $ 1.99451020893753 $\\
$ 13 $ & $ 1131093100800 $ & $ 1.99286917695807 $\\
$ 14 $ & $ 12202705912320 $ & $ 1.98865931503188 $\\
$ 15 $ & $ 18602849959680 $ & $ 1.98279525913578 $\\
$ 16 $ & $ 112390279723008 $ & $ 1.97593613726976 $\\
$ 17 $ & $ 2862815239454976 $ & $ 1.96855071453618 $\\
$ 18 $ & $ 12779353218700800 $ & $ 1.96096697625898 $\\
$ 19 $ & $ 359090847778602240 $ & $ 1.95340958290772 $\\
$ 20 $ & $ 528199187933007360 $ & $ 1.94602797836253 $\\
$ 21 $ & $ 773206609416933888 $ & $ 1.93891735688319 $\\
$ 22 $ & $ 24793217701106190336 $ & $ 1.93213421167903 $\\
$ 23 $ & $ 827889140509167744000 $ & $ 1.92570779661754 $\\
$ 24 $ & $ 598830989173253022720 $ & $ 1.91964852166174 $\\
$ 25 $ & $ 4317372411030828230400 $ & $ 1.91395405910916 $\\
$ 26 $ & $ 161375039894370452806656 $ & $ 1.90861374900449 $\\
$ 27 $ & $ 694064861294677860653568 $ & $ 1.90361174706010 $\\
$ 28 $ & $ 992499042379082029747200 $ & $ 1.89892924766034 $\\
$ 29 $ & $ 41060506623285608097730560 $ & $ 1.89454603039995 $\\
$ 30 $ & $ 29223129728659662019768320 $ & $ 1.89044151500365 $\\
$ 31 $ & $ 1286825956126314742827177984 $ & $ 1.88659546158441 $\\
$ 32 $ & $ 7297375514997659075363782656 $ & $ 1.88298841727536 $\\
$ 33 $ & $ 30980004273557163387248716800 $ & $ 1.87960198342170 $\\
$ 34 $ & $ 1488013910740373372116559769600 $ & $ 1.87641895752227 $\\
$ 35 $ & $ 2098693489678321945943778048000 $ & $ 1.87342338927132 $\\
$ 36 $ & $ 985153598609641299723155927040 $ & $ 1.87060057907881 $\\
$ 37 $ & $ 51256810213850056430278938193920 $ & $ 1.86793703936844 $\\
$ 38 $ & $ 2735164895775501868131145173811200 $ & $ 1.86542043302811 $\\
$ 39 $ & $ 11507626102899101257666168771829760 $ & $ 1.86303949906575 $\\
$ 40 $ & $ 8059299314330223666526571739832320 $ & $ 1.86078397238716 $\\
$ 41 $ & $ 462284296222289424199930047945627648 $ & $ 1.85864450235242 $\\
$ 42 $ & $ 323011260972275492313076690037274624 $ & $ 1.85661257315070 $\\
$ 43 $ & $ 19389095613689669633863411518989030400 $ & $ 1.85468042788833 $\\
$ 44 $ & $ 108154047860222226745809571120370442240 $ & $ 1.85284099748334 $\\
$ 45 $ & $ 50225014425692664920894784672307691520 $ & $ 1.85108783490726 $\\
$ 46 $ & $ 3215644392659577376710736031790702329856 $ & $ 1.84941505494173 $\\
$ 47 $ & $ 210167065857927773530182212723814604111872 $ & $ 1.84781727937288 $\\
$ 48 $ & $ 73000894982941520392127414736732633907200 $ & $ 1.84628958739195 $\\
$ 49 $ & $ 709397333815298882841078973866373982976000 $ & $ 1.84482747088061 $\\
$ 50 $ & $ 4920133537744546672264161440843391131136000 $ & $ 1.84342679421239 $\\
$ 51 $ & $ 20458927216447983273834531539671709940295680 $ & $ 1.84208375818503 $\\
$ 52 $ & $ 113346353260925291399743562400627050525736960 $ & $ 1.84079486770090 $\\
$ 53 $ & $ 8314574422219217877655745462771240829039820800 $ & $ 1.83955690282776 $\\
$ 54 $ & $ 34500154563812199832753196200423476103949352960 $ & $ 1.83836689289419 $\\
$ 55 $ & $ 238432041038408167094756668766278897709419315200 $ & $ 1.83722209330029 $\\
$ 56 $ & $ 164676676203627468970357206882080223885360414720 $ & $ 1.83611996475186 $\\
$ 57 $ & $ 681999861767638242983015292245103844475980738560 $ & $ 1.83505815465391 $\\
$ 58 $ & $ 54573997356138363766693363305148008723045636096000 $ & $ 1.83403448042581 $\\
$ 59 $ & $ 4439790035100283845678937980662381414861625010790400 $ & $ 1.83304691452571 $\\
$ 60 $ & $ 101975349958352280490764879669744920199211594076160 $ & $ 1.83209357099465 $\\
$ 61 $ & $ 8567936319705929309459776317134905109846249700950016 $ & $ 1.83117269335212 $\\
$ 62 $ & $ 731296854673252328188674061496432915272299105098268672 $ & $ 1.83028264369378 $\\
$ 63 $ & $ 335413015049091357299052094007521449520575754258022400 $ & $ 1.82942189285912 $\\
$ 64 $ & $ 1845167454673977226018766301655167043034750179509534720 $ & $ 1.82858901155228 $\\
$ 65 $ & $ 12682256746680458760243753700787130977083129728863846400 $ & $ 1.82778266231236 $\\
$ 66 $ & $ 52276868762074913021526303076355218701929049349712609280 $ & $ 1.82700159224222 $\\
$ 67 $ & $ 4810420494471412979007411287036242626786207046400382812160 $ & $ 1.82624462641455 $\\
$ 68 $ & $ 26415219830028198673102330954653415134032765654140895232000 $ & $ 1.82551066188405 $\\
$ 69 $ & $ 108743820803513005688074019486373054938055526153150572134400 $ & $ 1.82479866224234 $\\
$ 70 $ & $ 74580752041543444800113724539276965641643010019336359116800 $ & $ 1.82410765265959 $\\
$ 71 $ & $ 7260478986629636082091206658585523413380603478120191670026240 $ & $ 1.82343671536324 $\\
$ 72 $ & $ 1658543716135460363333612862531407314663868817934410027171840 $ & $ 1.82278498550972 $\\
$ 73 $ & $ 165882266747924865206904084016310339756148556672966090106880000 $ & $ 1.82215164741010 $\\
$ 74 $ & $ 16812159284890750979749243220614306695110886412329191789076480000 $ & $ 1.82153593107479 $\\
$ 75 $ & $ 4603531624150913959448362425199232529734268515131705294196736000 $ & $ 1.82093710904641 $\\
$ 76 $ & $ 25202209012410533540246834563319118233637390214181198060814336000 $ & $ 1.82035449349305 $\\
$ 77 $ & $ 241366977197967677538973657807799949747740103309025947202792980480 $ & $ 1.81978743353741 $\\
$ 78 $ & $ 990370612570418445427360530484689979937185043796494879714417049600 $ & $ 1.81923531279978 $\\
$ 79 $ & $ 106975541100328991711318407253896890333906467020352210539219357532160 $ & $ 1.81869754713491 $\\
$ 80 $ & $ 36555154227540053054153934850010826302616194014823661503146548264960 $ & $ 1.81817358254549 $\\
$ 81 $ & $ 149851736673363398929001194123847945118734199135584279281761945747456 $ & $ 1.81766289325595 $\\
$ 82 $ & $ 16785663926886993783588316183126473147402708126885324650447584814825472 $ & $ 1.81716497993260 $\\
$ 83 $ & $ 1902626661050162065250048422831769419698530758031523578321484155933491200 $ & $ 1.81667936803728 $\\
$ 84 $ & $ 30923475731107214452096082991518502460176669039705846470147052753223680 $ & $ 1.81620560630270 $\\
$ 85 $ & $ 211033905115271417271660172990407703232255473542069106030658226809241600 $ & $ 1.81574326531946 $\\
$ 86 $ & $ 24764371396919770832019746470968483879210345845263185710296564094429429760 $ & $ 1.81529193622492 $\\
$ 87 $ & $ 101346920171814647629695589724332790530159075926733679909289241095523860480 $ & $ 1.81485122948567 $\\
$ 88 $ & $ 552866781404733560486999559143123682510616270452709684858371540835172352000 $ & $ 1.81442077376579 $\\
$ 89 $ & $ 67088937039147459499754442509276049267011574928837849289798555373109533081600 $ & $ 1.81400021487398 $\\
$ 90 $ & $ 15241683690608024309993834064099842381432707503708748880342411792796460646400 $ & $ 1.81358921478300 $\\
$ 91 $ & $ 145398351752565018216111572213953306791698137456938745871197983890113710981120 $ & $ 1.81318745071601 $\\
$ 92 $ & $ 792401796457255403746361349635606869469307241749908813534695742680107311431680 $ & $ 1.81279461429416 $\\
$ 93 $ & $ 3238115619049237198790651919466157440602271940770629820254905680155048148992000 $ & $ 1.81241041074103 $\\
$ 94 $ & $ 414522052599544861492950324749471785188557424644373479510407159418931153941299200 $ & $ 1.81203455813922 $\\
$ 95 $ & $ 2821947459113873562606816047459173388991579282609847003837813920767783874330624000 $ & $ 1.81166678673537 $\\
$ 96 $ & $ 480171136596231557363608952686919350292034173913395737093202851833349060716134400 $ & $ 1.81130683828988 $\\
$ 97 $ & $ 63388819518945706718347449300458544859519346709249567830809222050220147678932172800 $ & $ 1.81095446546787 $\\
$ 98 $ & $ 603761225868438398106902388365133121350001348602049530759138537188486716137062400000 $ & $ 1.81060943126867 $\\
$ 99 $ & $ 2464066580622273744937507904234552825676575325097370269243681346381530312880332800000 $ & $ 1.81027150849064 $\\
$ 100 $ & $ 670288255898885382815113952829631306581975481222939788400366747112035699423887360000 $ & $ 1.80994047922909 $\\\hline
\end{tabular}
\end{center}
}
\end{document}